\documentclass[a4paper,fleqn]{CSML}
\pdfoutput=1

\usepackage{lastpage}
\newcommand{\dOi}{13(4:15)2017}
\lmcsheading{}{1--\pageref{LastPage}}{}{}%
{Nov.~03, 2016}{Nov.~22, 2017}{}

 \usepackage{amsmath}
  \usepackage{amssymb}
  \usepackage{marvosym,wasysym}
  \usepackage{stmaryrd,mathtools,multirow}
\usepackage{xspace}

\usepackage{hyperref}
\hypersetup{hidelinks}

\providecommand{\phantomsection}{}

\AtBeginDocument{\let\textlabel\label}
\makeatletter
\newcommand{\mylabel}[2]{\raisebox{.7\normalbaselineskip}{\phantomsection}#1%
  \def\@currentlabel{#1}\textlabel{#2}}
\makeatother

\usepackage{thm-restate}

\newcommand{\eg}{\emph{e.g.}\xspace}
\newcommand{\ie}{\emph{i.e.}\xspace}

\newcommand{\viz}{\emph{viz.}\xspace}

\newcommand{\etal}{\emph{et al.}\xspace}

\newcommand{\rel}[1]{\ensuremath{\mathrel{#1}}}

\newcommand{\sconf}[1]{\ensuremath{[\,#1\,]}}
\newcommand{\dconf}[1]{\ensuremath{\langle\, #1\,\rangle}}
\newcommand{\conf}[1]{\ensuremath{(#1)}}

\newcommand{\bisim}{\ensuremath{\leftrightarroweq}}
\newcommand{\bbbisim}{\ensuremath{\leftrightarroweq_b}}
\newcommand{\dbisim}{\ensuremath{\leftrightarroweq_d}}
\newcommand{\ebisim}{\ensuremath{\leftrightarroweq_e}}
\newcommand{\gbbbisim}{\ensuremath{\leftrightarroweq_{(x,y)}}}
\newcommand{\gbbbisimargs}[1]{\ensuremath{\leftrightarroweq_{(#1)}}}
\newcommand{\gbbsim}{\ensuremath{\leq_{(x,y)}}}
\newcommand{\gbbsimeq}{\ensuremath{\sim_{(x,y)}}}
\newcommand{\gbbedbisim}{\ensuremath{\leftrightarroweq_{(x,y)}^{ed}}}
\newcommand{\gbbedbisimargs}[1]{\ensuremath{\leftrightarroweq_{(#1)}^{ed}}}

\newcommand{\wbisim}{\ensuremath{\leftrightarroweq_w}}
\newcommand{\bbedbisim}{\ensuremath{\leftrightarroweq_b^{ed}}}
\newcommand{\sbg}{\ensuremath{\equiv_s}}
\newcommand{\bbg}{\ensuremath{\equiv_b}}
\newcommand{\lbbg}{\ensuremath{\equiv_{lb}}}
\newcommand{\gbbg}{\ensuremath{\equiv}}
\newcommand{\gxybbg}{\ensuremath{\equiv_{E(x,y)}}}

\newcommand{\bbedg}{\ensuremath{\equiv_b^{ed}}}
\newcommand{\gdxybbg}{\ensuremath{\equiv_{E(x,y)}^{ed}}}

\newcommand{\chal}{\textsl{Challenge}}
\newcommand{\chaldagger}{\textsl{Challenge}_{\dagger}}
\newcommand{\mov}{\textsl{Match}_{\dagger}}
\newcommand{\pos}{\textsl{Position}}
\newcommand{\rew}{\textsl{Reward}}

\newcommand{\spoiler}{\ensuremath{\textsc{Spoiler}}\xspace}
\newcommand{\duplicator}{\ensuremath{\textsc{Duplicator}}\xspace}

\usepackage{tikz}
\usetikzlibrary{arrows,automata,shapes,calc}

\newcommand{\struct}[1]{\langle #1 \rangle}
\newcommand{\pijl}[1]{\xrightarrow{\smash{#1}}}
\newcommand{\taustar}{\twoheadrightarrow}
\newcommand{\taustarm}{\twoheadleftarrow}
\newcommand{\tauplus}{\twoheadrightarrow^+}

\renewcommand{\sb}{{sb}\xspace}
\newcommand{\lbb}{{lbb}\xspace}
\newcommand{\bb}{{bb}\xspace}
\newcommand{\gb}{{gb}\xspace}

\newcommand{\bbed}{{bbed}\xspace}
\newcommand{\gbed}{{gbed}\xspace}

\newcommand{\drawinitarrow}{\ensuremath{{\bullet}\!\!\rightarrow}}

\ACMCCS{[{\bf Theory of computation}]: Concurrency; Algorithmic game theory; Semantics and reasoning.}

\usepackage{hyperref}
\hypersetup{hidelinks}
\theoremstyle{plain} 

\begin{document}

\title[Games for Bisimulations and Abstraction]{Games for Bisimulations and Abstraction}

\author[D.~de~Frutos~Escrig]{David de Frutos Escrig}
\address{Dpto. Sistemas Inform{\'a}ticos y Computaci{\'o}n -
Facultad CC. Matem{\'a}ticas \\
Universidad Complutense de Madrid,
Spain}
\email{defrutos@sip.ucm.es}

\author[J.J.A.~Keiren]{Jeroen J.A. Keiren}
\address{Open University in the Netherlands
    \and
    Radboud University, 
    Nijmegen, The Netherlands}
\email{Jeroen.Keiren@ou.nl}

\author[T.A.C.~Willemse]{Tim A.C. Willemse}
\address{Eindhoven University of Technology, Eindhoven, The Netherlands}
\email{T.A.C.Willemse@TUe.nl}

\begin{abstract}
Weak bisimulations are typically used in process algebras where silent steps are used to abstract from internal behaviours. They facilitate relating implementations to specifications. When an implementation fails to conform to its specification, pinpointing the root cause can be challenging. In this paper we provide a generic characterisation of branching-, delay-, $\eta$- and weak-bisimulation as a game between \spoiler and \duplicator, offering an operational understanding of the relations. We show how such games can be used to assist in diagnosing non-conformance between implementation and specification. Moreover, we show how these games can be extended to distinguish divergences.
\end{abstract}

\maketitle

\section{Introduction}\label{sec:introduction}

Abstraction is a powerful, fundamental concept in process theories. It
facilitates reasoning about the conformance between an implementation
and specification of a (software) system, described by a transition
system. Essentially, it allows
one to ignore (\ie, abstract from) implementation details that
are unimportant from the viewpoint of the specification.

There is a wealth of behavioural equivalences (and preorders), each
treating abstraction in slightly different manners~\cite{vanglabbeek_linear_1993}.
Some prototypical equivalences have been incorporated in
contemporary tool sets that implement verification technology for
(dis)proving the correctness of software systems. These equivalences include
branching bisimilarity~\cite{vanglabbeek_branching_1996} and branching
bisimilarity with explicit divergence~\cite{vanglabbeek_branching_2009},
which are both used in tool sets such as CADP~\cite{garavel_cadp_2013},
$\mu$CRL~\cite{blom_mcrl_2001}, and mCRL2~\cite{cranen_overview_2013}, and
weak bisimilarity~\cite{milner_calculus_1980} which is used in, for example, FDR~\cite{fdr}.

The key idea behind these weak behavioural equivalences is that
they abstract from `internal' events (events that are \emph{invisible}
to the outside observer of a system). At the same time, they preserve
essential parts of the branching structure of the transition systems.
The behavioural equivalences differ in the extent to which the
branching structure is preserved.  Allowing abstraction from invisible
actions can make it difficult to explain why a particular pair of
states is not equivalent, as one must somehow capture the loss of
potential future computations in the presence of internal actions.
While (theoretical) tools such as distinguishing formulae can help
to understand why two states are distinguishable, these may not
always be very accessible and, to date, the idea of integrating such formulae
in tool sets seems not to have caught on, as witnessed by their
absence in the four prominent tool sets for analysing labelled transition
systems, \viz 
CADP~\cite{garavel_cadp_2013}, $\mu$CRL~\cite{blom_mcrl_2001}, mCRL2~\cite{cranen_overview_2013}
and FDR~\cite{fdr}; the only tool supporting distinguishing formulae seems
to be the (education-focussed) tool CAAL~\cite{AndersenAEHLOSW15}, which
uses both games and formulae.
We note that distinguishing formulae and games for behavioural equivalences
are complementary. Formulae offer a high-level explanation for the inequivalence
of two systems in terms of their
`capabilities',
whereas the game characterisation allows
for explaining their inequivalence at the level of the structure of the
transition systems.  

\subsection*{Contributions}
We provide game-based views on four weak
behavioural equivalences.  Such a game-based view provides an alternative point of
view on the traditional coinductive definitions. Moreover, we argue,
using some examples, how such games can be used to give an operational
explanation of the inequivalence of states following the ideas
in~\cite{stevens_practical_1998}, thereby remaining close to the
realm of transition systems.
Our first main contribution can be summarised as follows:
\begin{quote}
\normalfont\emph{Contribution~1.} We show that branching-,
$\eta$-, delay-, and weak-bisimilarity can be characterised
by Ehrenfeucht-Fra\"iss\'e games~\cite{thomas_ehrenfeucht_1993}
in Theorem~\ref{thm:correspondence}.
\end{quote}
We do not obtain just a collection of (more or less) similar
game-based characterisations. Instead, we present a generic,
parameterised characterisation that captures all these (closely related)
semantics in a single framework, resulting in
a single definition for all four games. In addition,
this unified presentation enables a generic proof of the correctness
of our game characterisation.

We furthermore study the notion of divergence for all four weak behavioural
equivalences. A divergence
can be understood as a computation that only involves unobservable behaviour.
From the perspective of an observer, a system that diverges can thus appear to be
`stuck'. For that reason, it may be undesirable to abstract from 
divergences. While coinductive definitions of divergence-aware weak behavioural
equivalences treat divergence---in some sense---as a separate issue, we show that
divergence is obtained through a natural modification of our generic
game-based characterisation of these equivalences.
Our second main contribution is therefore as follows:
\begin{quote}
\emph{Contribution~2}. We generalise 
branching bisimilarity with explicit divergence~\cite{vanglabbeek_branching_2009} 
to branching-, $\eta$-, delay-, and weak-bisimilarity with explicit
divergence and we show that all four notions can be characterised
by Ehrenfeucht-Fra\"iss\'e games in Theorem~\ref{th:bbbedbisim_is_bbedg}.
\end{quote}
This paper is an extended and enhanced version
of~\cite{branching_games_2016}, where we defined 
games for branching bisimilarity, and branching
bisimilarity with explicit divergence and proved that these games
characterise branching bisimilarity and branching bisimilarity with
explicit divergence, respectively.
A major difference between the games of the current paper and the game-based characterisations of~\cite{branching_games_2016} is
that in \emph{ibid.}\ we employed the stuttering 
condition in the definition of branching bisimulation, whereas in this paper we need to
resort to a slightly different mechanism for dealing with the four weak
behavioural equivalences in a unified way. The stuttering condition underlies
branching bisimilarity and increases its local coinductive 
character, thus allowing to arrive at a `simple' game.  In Theorem~\ref{thm:branching_special_case},
we nevertheless show how the branching
bisimulation game of~\cite{branching_games_2016} can be recovered from our
generic game. The proof of this theorem exposes the close connection between
the two games.

\subsection*{Related work}
The idea of using two-player games for bisimulation originates with Stirling, who described \emph{strong bisimulation} games \cite{stirling_bisimulation_1999}. 
Yin \etal recently presented a branching bisimulation game, in the context of normed
process algebra \cite{yin_branching_2014}, but their game uses
moves that consist of sequences of silent steps, rather than single steps.
As argued convincingly by Namjoshi~\cite{namjoshi_simple_1997}, local
reasoning using single steps often leads to simpler arguments. 

The two-player game for divergence-blind stuttering bisimilarity~\cite{denicola_three_1995} (a relation for
Kripke structures that in essence is the same as branching bisimilarity),
provided by Bulychev \etal in~\cite{bulychev_computing_2007} comes closer
to our work on games for branching bisimilarity. However, their game-based definition
is sound only
for transition systems that are essentially free of divergences, so that in 
order to deal with transition systems containing divergences they need an
additional step that precomputes and eliminates
these divergences. Such a preprocessing step is a bit artificial, and makes
it hard to present the user with proper diagnostics, as it fundamentally modifies 
the structure of the transition system. Moreover, infinite state
transition systems cannot be dealt with this way.
As far as we are aware, ours is the first work that tightly
integrates dealing with divergences in a game-based characterisation of a
behavioural equivalence.

In a broader context, the type of games we study in this paper have been
used for characterising different types of relations in settings other
than labelled transition systems.  For instance,
in~\cite{Hutagalung16}, simulation games on B\"uchi automata are
defined with the purpose of approximating the inclusion of trace
closures of languages accepted by finite-state automata.  Similarly,
in~\cite{Etessami:01}, several algorithms based on games characterising
a variety of simulation relations on the state space of a B\"uchi
automaton are studied. In~\cite{Fritz06}, a delayed simulation
game for parity games is defined and investigated, and in~\cite{CKW:17},
the branching bisimulation games we studied in~\cite{branching_games_2016}
are lifted to the setting of parity games.

There is a large number of behavioural equivalences in the context of abstraction, each having its own applications and properties. An extensive overview of these was given by van Glabbeek~\cite{vanglabbeek_linear_1993}. In practice, behavioural equivalences are used that can be computed efficiently, yet that still allow to equate a large number of systems. Weak bisimilarity~\cite{milner_calculus_1980} and branching bisimilarity~\cite{vanglabbeek_branching_1996} are the dominant equivalences in use. The notions of $\eta$-bisimilarity~\cite{baeten_another_1987} and delay bisimilarity~\cite{milner_modal_1981} are between weak- and branching bisimilarity in terms of distinguishing power.

Divergence has been studied as a concept orthogonal to bisimulation.  Different notions of divergence have been defined in the literature, differing in their abilities to, for example, distinguish livelocks (infinite, internal computations) from deadlock. Milner, \eg studied a notion of observation equivalence with divergence \cite{milner_modal_1981}. In the branching bisimulation setting, the notions of branching bisimilarity with explicit divergence \cite{vanglabbeek_branching_1996} and divergence sensitive branching bisimilarity \cite{denicola_three_1995} have been studied. For the first, different, yet equivalent characterisations have been explored by van Glabbeek \etal \cite{vanglabbeek_branching_2009}.\footnote{Note that the literature uses the same names for different notions of sensitivity to divergence, but also uses different names for the same notion of divergence, so when studying these issues one needs to tread carefully.}

\subsection*{Outline of the paper}
Labelled transition systems, and strong- and branching bisimilarity are introduced in Section~\ref{sec:preliminaries}.
Subsequently, we introduce bisimulation games, using Stirling's games as an example, in Section~\ref{sec:games}. In the same
section, we present Bulychev's version of branching bisimulation games, and explain its shortcomings, and we introduce the branching bisimulation games from~\cite{branching_games_2016}. Next, in Section~\ref{sec:generic}, we introduce the coinductive definitions of weak-, delay-, and $\eta$-bisimilarity, and present a generic bisimulation definition that covers all three bisimulations as well as branching bisimulation.
In Section~\ref{sec:generic_games} we introduce a generic bisimulation game that is parameterised such that it covers all four bisimulations. We discuss the concepts of divergence and simulation, and the changes to our generic bisimulation game required to capture these in Section~\ref{sec:extensions}. We illustrate our game using a small example in Section~\ref{sec:smallapplication}. Finally, we conclude in Section~\ref{sec:conclusions}.

\section{Preliminaries}\label{sec:preliminaries}

We are concerned here with relations on labelled transition systems
that include both {\em observable} transitions, and {\em internal} transitions
labelled by the special action $\tau$.
\begin{defi}
A \emph{Labelled Transition System} (LTS) is a structure $L = \struct{S,A,\to}$ where:
\begin{itemize}
\item $S$ is a set of states, 
\item $A$ is a set of actions containing a special action $\tau$,
\item ${\to} \subseteq S \times A \times S$ is the transition relation.
\end{itemize}
\end{defi}
As usual, we write $s \pijl{a} t$ to stand for $(s,a,t) \in\to$. The
reflexive-transitive closure of the $\pijl{\tau}$ relation is denoted by
$\taustar$ and the transitive closure of the $\pijl{\tau}$ relation is
denoted $\tauplus$.  Given a relation $R \subseteq S \times S$ on states, we simply write 
$s\rel{R} t$ to represent $(s,t) \in R$. 
We say that a labelled transition system is \emph{divergent} if it admits an
infinite sequence $s \pijl{\tau} s_1 \pijl{\tau} s_2 \cdots$ from some state
$s \in S$, and we say it is \emph{non-divergent} if it contains no such sequence.

Strong bisimulation, due to Park \cite{park_concurrency_1981}, is arguably the finest meaningful
behavioural equivalence defined for labelled transition systems. 
\begin{defi}[{\cite{park_concurrency_1981}}]
\label{def:strong_bisimulation}
A symmetric relation $R \subseteq S \times S$ is said to be a \emph{strong bisimulation} whenever for all $s \rel{R} t$, if $s \pijl{a} s'$ then there exists a state $t'$ such that $t \pijl{a} t'$ and $s' \rel{R} t'$. We write $s \bisim t$ and say that $s$ and $t$ are strongly bisimilar if and only if there is a strong bisimulation relation $R$ such that $s \rel{R} t$.
\end{defi}

In this paper we are mainly concerned with weaker bisimulations, that allow to abstract from \emph{internal} transitions. The finest such
bisimulation is \emph{branching bisimulation}, that was introduced by van Glabbeek and Weijland in~\cite{vanglabbeek_branching_1996}. 

\begin{defi}[{\cite{vanglabbeek_branching_1996}}]
\label{def:branching_bisimulation}
A symmetric relation $R \subseteq S \times S$ is said
to be a \emph{branching bisimulation} whenever $s \rel{R} t$ and $s \pijl{a} s'$ imply either:
\begin{itemize}
  \item $a = \tau$ and $s' \rel{R} t$, or
  \item there exist states $t', t_1$ such $t \taustar t_1 \pijl{a} t'$, $s\rel{R}t_1$ and $s'\rel{R}t'$;
\end{itemize}
We write $s \bbbisim t$ and say that $s$ and $t$ are branching bisimilar, 
iff there is a branching bisimulation $R$ such that $s \rel{R} t$.
Typically we simply write $\bbbisim$ to denote {\em branching bisimilarity}.
\end{defi}

Basten \cite{basten_branching_1996}, finally showed that branching bisimilarity is indeed an equivalence relation.
\begin{thm}[{\cite[Corollary 13]{basten_branching_1996}}]
\label{th:bbbisim_is_equivalence}
Branching bisimilarity, $\bbbisim$, is an equivalence relation.
\end{thm}


Branching bisimilarity has the {\em stuttering 
property}, see~\cite[Lemma~2.5]{vanglabbeek_branching_1996}. This means that the condition 
$s\rel{R}t''$ imposed 
 on the weak transition $t \taustar t''$ at Definition~\ref{def:branching_bisimulation} can be
strengthened, so that all the intermediate states $t''_i$ along the weak transition $t \taustar t''$,
and not just the final state $t''$, will
have to satisfy the condition $s\rel{R}t''_i$. This fact is quite useful when studying other 
properties of branching bisimilarity, and in particular plays a central role when developing 
the characterisation of branching bisimilarity by means of a branching bisimulation game.

\begin{defi}[{\cite[Lemma~2.5]{vanglabbeek_branching_1996}}, \cite{sangiorgi_introduction_2012,vanglabbeek_branching_2009}]
\label{def:stuttering property}
A relation $R$ has the {\em stuttering property} if, 
whenever $t_0 \pijl{\tau} t_1  \cdots \pijl{\tau} t_k$ with 
 $t_0 \rel{R} t_k$, then $t_i \rel{R} t_j$, for all $0 \leq i,j \leq k$.
\end{defi}


\section{Games for Strong and Branching Bisimulation}\label{sec:games}
In the previous section we have reiterated the coinductive definitions of strong- and branching bisimulation. There are several alternatives for defining behavioural equivalences. Relations can, for example be defined using a fixed point characterisation, or using a two-player game. In this paper we are concerned with such two-player games. We first give some background on the games, and recall strong bisimulation games, before we proceed to the branching bisimulation games that we introduced in \cite{branching_games_2016}.

\subsection{Games}
The games we consider in this paper are, essentially, Ehrenfeucht-Fra{\"\i}ss{\'e} games, whose use in computer science has
been reviewed in \cite{thomas_ehrenfeucht_1993}.
They are instances of
two-player infinite-duration games with $\omega$-regular winning
conditions, played on game arenas that can be represented by graphs. In these games
each vertex is assigned to one of two players, here called \spoiler and \duplicator. The players
move a token over the vertices as follows. The player that `owns' the vertex where the token is pushes
it along an edge to an adjacent vertex, and this continues as long as possible, possibly forever.
The winner of the play is decided from the resulting sequence of vertices visited by the token,
depending on the predetermined winning criterion. We say that a player \emph{can win from a 
given vertex} if she has a strategy  such that any play with the token initially at that 
vertex will be won by her.
The games that we consider here are \emph{memoryless} and \emph{determined}: every
vertex is won by (exactly) one player, and the winning player has 
a \emph{positional} winning strategy, meaning that she can decide her winning 
moves based only on the vertex where the token
currently resides, without inspecting the previous moves of the play. A strategy
of a player
can thus be given by a function that maps each vertex owned by this player to
one of its adjacent vertices. A play
is \emph{consistent} with a strategy of a given player whenever on this play, the
successor of each vertex owned by this player is given by the strategy.
A strategy of a player induces a
\emph{solitaire game} by restricting the edges in the graph that 
emanate from this player's set of vertices to those used by her
strategy.
Note that a play that is consistent both with a strategy of \spoiler and a strategy
of \duplicator is a unique path in the graph.
Winning strategies can be efficiently
computed while solving the game. We refer to~\cite{gradel_automata_2002}
for a more in-depth treatment of the underlying theory.

\subsection{Stirling Games for Strong Bisimulation}
The most well-known, and simplest game-based characterisation of a behavioural
equivalence is Stirling's notion of a \emph{strong bisimulation game} \cite{stirling_bisimulation_1999}.

The game is played on pairs of states $(s,t)$ in a labelled transition system, where \duplicator tries
to prove that states $s$ and $t$ are strongly bisimilar, whereas \spoiler tries to disprove this.
To achieve this, \spoiler plays first, and challenges \duplicator with a move from either $s$ or $t$, and
\duplicator is required to match this move. If play continues indefinitely, this means \duplicator is always
able to match \spoiler's challenges, and the states are equivalent. If not, this means \spoiler found a way
to distinguish the states.
The game is formally defined as follows.

\begin{defi}
\label{def:strong_bisimulation_game}
A \emph{strong bisimulation (sb) game} on an LTS $L$ is played by players \spoiler
and \duplicator on an arena of \spoiler-owned configurations $\sconf{(s,t)}$ and \duplicator-owned
configurations $\dconf{(s,t), c}$, where $(s,t) \in \pos$ and $c \in \chal$, and $\pos = S \times S$ is the set of
\emph{positions} and $\chal = (A \times S)$ the set of \emph{pending challenges}.  
\begin{itemize}
  \item \spoiler moves from a configuration $\sconf{(s,t)}$ by:
  \begin{enumerate}
    \item selecting $s \pijl{a} s'$ and moving to $\dconf{(s, t), (a, s')}$, or
    \item selecting $t \pijl{a} t'$ and moving to $\dconf{(t, s), (a, t')}$.
  \end{enumerate}
  \item \duplicator responds from a configuration $\dconf{(u, v), (a, u')}$ by playing
  $v \pijl{a} v'$ and continuing in configuration $\sconf{(u', v')}$.
\end{itemize}
Finite plays are won by \duplicator if and only if \spoiler gets stuck. All infinite plays are won by \duplicator.
Full plays of the game start in a configuration $\sconf{(s,t)}$; 
we say that \duplicator wins the \sb-game for a position $(s,t)$, if the configuration $\sconf{(s,t)}$
is won by her; in this case, we write $s \sbg t$. Otherwise, we say that \spoiler wins that game.
\end{defi}
We illustrate the above game using a small example.
\begin{exa} Consider the four-state transition system depicted below (left). 
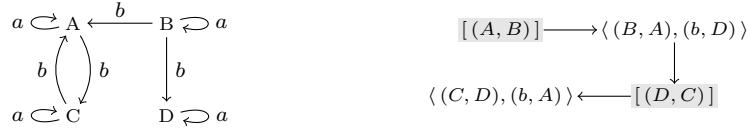
\begin{center}
\begin{minipage}{.40\textwidth}
  \centering
  \begin{tikzpicture}
    [node distance=55pt,inner sep = 1pt,minimum size=10pt]

     \tikzstyle{state}=[circle, draw=none,node distance=35pt]
     \tikzstyle{transition}=[->,>=stealth']

      \node[state] (A) {\tiny A};
      \node[state] [right of=A] (B) {\tiny B};
      \node[state] [below of=A] (C) {\tiny C};
      \node[state] [below of=B] (D) {\tiny D};
      \draw [->] (A) edge[bend left] node[right] {\scriptsize $b$} (C);
      \draw [->] (C) edge[bend left] node[left] {\scriptsize $b$} (A);
      \draw [->] (A) edge[loop left] node[left] {\scriptsize $a$} (A);
      \draw [->] (C) edge[loop left] node[left] {\scriptsize $a$} (C);

      \draw [->] (B) edge node[right] {\scriptsize $b$} (D);
      \draw [->] (B) edge[loop right] node[right] {\scriptsize $a$} (B);
      \draw [->] (D) edge[loop right] node[right] {\scriptsize $a$} (D);
      \draw [->] (B) edge node[above] {\scriptsize $b$} (A);

  \end{tikzpicture}
\end{minipage}
\begin{minipage}{.40\textwidth}
\vspace{5pt}
\centering
  \begin{tikzpicture}
    [node distance=25,inner sep = 1pt,minimum size=5pt,initial text={}]

     \tikzstyle{state}=[rectangle, draw=none]
     \tikzstyle{transition}=[->,>=stealth']

      \node[state,fill=gray!20] (sAB) {\tiny $\sconf{(A,B)}$};
      \node[state,right of=sAB,xshift=40pt] (BAbD) {\tiny $\dconf{(B,A),(b,D)}$};
      \node[state,below of=BAbD,fill=gray!20] (sDC) {\tiny $\sconf{(D,C)}$};
      \node[state,left of=sDC,xshift=-40pt] (CDbA) {\tiny $\dconf{(C,D),(b,A)}$};

      \draw [->] (sAB) edge (BAbD);
      \draw [->] (BAbD) edge (sDC);
      \draw [->] (sDC) edge (CDbA);

  \end{tikzpicture}
\vspace{5pt}
\end{minipage}
\end{center}
Observe that states $A$ and $B$ are not strongly bisimilar. This can
be seen by the (solitaire) game, depicted right, in which \spoiler plays
her winning strategy, first challenging \duplicator to match a $b$-transition,
and, when she does so by moving to $C$, switch positions and challenging her 
to match another $b$-transition (which \duplicator cannot).

States $A$ and $C$, on the other hand, are strongly bisimilar. While
it is easy to check that, indeed, the relation $R =\{(A,C),(C,A)\}$
is a strong bisimulation relation, the fact that both states are
strongly bisimilar also follows from the solitaire game depicted below, in
which \duplicator plays her winning strategy.
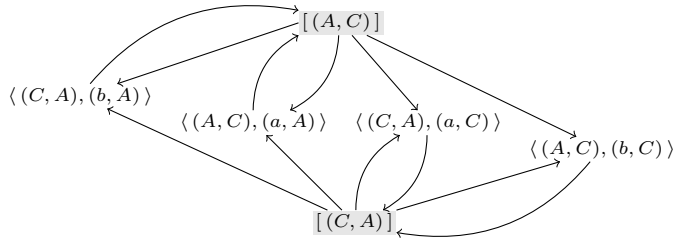
\begin{center}
\begin{minipage}{.60\textwidth}
\vspace{5pt}
\centering
  \begin{tikzpicture}
    [node distance=25,inner sep = 1pt,minimum size=5pt,initial text={}]

     \tikzstyle{state}=[rectangle, draw=none]
     \tikzstyle{transition}=[->,>=stealth']

      \node[state,fill=gray!20] (sAC) {\tiny $\sconf{(A,C)}$};
      \node[state,below left of=sAC,xshift=-80pt,yshift=-10pt] (CAbA) {\tiny $\dconf{(C,A),(b,A)}$};
      \node[state,below left of=sAC,xshift=-15pt,yshift=-20pt] (ACaA) {\tiny $\dconf{(A,C),(a,A)}$};
      \node[state,below right of=sAC,xshift=80pt,yshift=-30pt] (ACbC) {\tiny $\dconf{(A,C),(b,C)}$};
      \node[state,below right of=sAC,xshift=15pt,yshift=-20pt] (CAaC) {\tiny $\dconf{(C,A),(a,C)}$};
      \node[state,below left of=CAaC,xshift=-10pt,fill=gray!20,yshift=-20pt] (sCA) {\tiny $\sconf{(C,A)}$};

      \draw [->] (ACaA) edge[bend left] (sAC.south west);
      \draw [->] (sAC) edge[bend left] (ACaA);
      \draw [->] (sAC) edge (ACbC);

      \draw [->] (ACbC) edge[bend left] (sCA);
      \draw [->] (sCA) edge (ACbC);

      \draw [->] (CAaC) edge[bend left] (sCA);
      \draw [->] (CAbA) edge[bend left] (sAC.north west);
      \draw (sAC.west) edge[->] (CAbA);

      \draw [->] (sCA) edge[bend left] (CAaC);
      \draw [->] (sCA) edge (CAbA);

      \draw [->] (sAC) edge (CAaC);
      \draw [->] (sCA) edge (ACaA);

  \end{tikzpicture}
\vspace{5pt}
\end{minipage}
\end{center}
\end{exa}

\subsection{Branching Bisimulation Games for Non-Divergent LTSs}
The first attempt to define two-player games for weak bisimulations was made by Bulychev \etal \cite{bulychev_computing_2007}.
Their work defines a game for \emph{(divergence blind) stuttering equivalence}, which is the counterpart of branching bisimulation
for Kripke structures (which have state labels instead of edge labels). The following definition is a direct translation of 
the one from \cite{bulychev_computing_2007} to labelled transition systems. In fact, this definition weakens the strong bisimulation game
in two ways: first, it allows a $\tau$-transition to be mimicked by \duplicator by not moving, and second, when \spoiler challenges
with an $a$-transition, \duplicator can reply by making a $\tau$ move.
These are the first and third move for \duplicator in the following definition.
This game is limited to labelled transition systems in which there
are no divergences (\ie, no infinite sequences of $\tau$-transitions are possible).

\begin{defi}
\label{def:bulychev_game}
A \emph{limited branching bisimulation (\lbb) game} on an LTS $L$ is played by players \spoiler and \duplicator
on an arena of \spoiler-owned configurations $\sconf{(s,t)}$ and \duplicator-owned
configurations $\dconf{(s,t), c}$, where $(s,t) \in \pos$ and $c \in \chal$, and $\pos$ and $\chal$ are as before.
\begin{itemize}
  \item \spoiler moves from a configuration $\sconf{(s,t)}$ by:
  \begin{enumerate}
    \item selecting $s \pijl{a} s'$ and moving to $\dconf{(s, t), (a, s')}$, or
    \item selecting $t \pijl{a} t'$ and moving to $\dconf{(t, s), (a, t')}$.
  \end{enumerate}
  \item \duplicator responds from a configuration $\dconf{(u, v), (a, u')}$ by:
  \begin{enumerate}
    \item not moving if $a = \tau$ and continuing in configuration $\sconf{(u', v)}$, or
    \item playing $v \pijl{a} v'$ if available, and continuing in configuration $\sconf{(u', v')}$, or
    \item moving $v \pijl{\tau} v'$ if possible, and continuing in configuration $\sconf{(u, v')}$.
  \end{enumerate}
\end{itemize}
Finite plays are won by \duplicator if and only if \spoiler gets stuck. All infinite plays are won by \duplicator.
We say that a configuration is won by a player when she has a strategy that wins all plays starting in it. Full plays
of the game start in a configuration $\sconf{(s,t)}$; we say that \duplicator wins the \lbb-game for a position $(s,t)$
if the configuration $\sconf{(s,t)}$ is won by her; in this case we write $s \lbbg t$. Otherwise we say that \spoiler wins that game.
\end{defi}

By adapting Bulychev \etal's proof of~\cite[Theorem 1]{bulychev_computing_2007} to LTSs and branching bisimulation, we can show that this game-based definition
characterises branching bisimulation, provided the labelled transition system does not have divergences, \ie\ is void of infinite $\tau$-paths. 
\begin{thm}
Let $L = \struct{S,A,\to}$ be a non-divergent LTS. We have for all $s,t \in S$ that $s \bbbisim t$ if and only if $s \lbbg t$.
\end{thm}

We conclude this section with an example showing that the definition does not characterise
branching bisimilarity for arbitrary LTSs (\ie those that exhibit divergence).
\begin{exa}
Consider the following LTS, in which the states are labelled with their name.
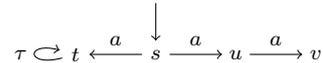
\begin{center}
  \begin{tikzpicture}
    [node distance=30,inner sep = 1pt,minimum size=10pt,initial text={}]

     \tikzstyle{state}=[circle, draw=none]
     \tikzstyle{transition}=[->,>=stealth']

      \node[state, initial above] (s) {\scriptsize $s$};
      \node[state] [left of=s] (t) {\scriptsize $t$};
      \node[state] [right of=s] (u) {\scriptsize $u$};
      \node[state] [right of=u] (v) {\scriptsize $v$};
      \draw [->] (s) edge node[above] {\scriptsize $a$} (t)
                 (s) edge node[above] {\scriptsize $a$} (u)
                 (t) edge[loop left] node[left] {\scriptsize $\tau$} (t)
                 (u) edge node[above] {\scriptsize $a$} (v);
  \end{tikzpicture}
\end{center}
Observe that there is a divergence because of the $\tau$-loop in $t$. Furthermore, $t$ and $u$ are not branching bisimilar since $u \pijl{a} v$ can
never be mimicked from $t$. However, \duplicator can win the \lbb-game starting in $\sconf{(t,u)}$ according to Definition~\ref{def:bulychev_game} (regardless of \spoiler's moves). When starting in $\sconf{(t,u)}$, \spoiler can choose to play on $u$ or on $t$. We distinguish these two cases:
\begin{itemize}
  \item $\sconf{(t,u)} \to \dconf{(u,t),(a,v)}$, then \duplicator can respond with $\dconf{(u,t),(a,v)} \to \sconf{(u,t)}$ by playing $t \pijl{\tau} t$ according to her third move, or
  \item $\sconf{(t,u)} \to \dconf{(t,u),(\tau,t)}$, then \duplicator can respond with $\dconf{(t,u),(\tau,t)} \to \sconf{(t,u)}$ by not moving according to her first move.
\end{itemize}
In both cases, the game ends up in a configuration for which we have shown \duplicator can respond, and the game will continue indefinitely, hence \duplicator wins, even though $t \not  \bbbisim u$.
\end{exa}

\subsection{Branching Bisimulation Games}

Since all vertices on a strongly connected component reachable through $\tau$-transitions are branching
bisimilar, any finite LTS can be preprocessed and turned into a non-divergent
LTS. Such preprocessing
is often part of a state space minimisation algorithm. However, for use cases such as
debugging a specification or an implementation, it is desirable to avoid any preprocessing, and stay as close
as possible to any user-provided specification. Furthermore, for infinite LTSs, such a preprocessing
step is impossible altogether. In \cite{branching_games_2016} we therefore investigated an
alternative game-based characterisation for branching bisimulation. We introduce this alternative in this
section.

First, let us take a closer look at Definition~\ref{def:bulychev_game}. Observe that, when \spoiler challenges \duplicator
by playing, \eg, $s \pijl{a} s'$, \duplicator can get away with not ever playing an $a$-transition, by playing a $\tau$-transition instead.
Intuitively, this means that \duplicator never truly answers to \spoiler's challenge, since this challenge is forgotten after she plays the
$\tau$-transition in her third move. The key idea in our branching bisimulation game is to preserve an unanswered challenge after \duplicator's move.
To facilitate this, we also include the challenge in \spoiler's configurations. Next, we also introduce a reward for \duplicator that she earns whenever
she actually matches \spoiler's challenge. The need for this reward is illustrated by the following example.

\begin{exa}
\label{ex:need_for_challenges}
Consider the first LTS depicted in Figure~\ref{fig:need_for_challenges}.
Observe that $s_0$ and $t_0$ are branching bisimilar. Suppose
\spoiler tries (in vain) to disprove that $s_0$ and $t_0$ are
branching bisimilar and challenges \duplicator by playing $s_0
\pijl{a} c_1$.  \duplicator may respond with an infinite sequence
of $\tau$-steps, moving between $t_0$ and $t_1$, so long as \spoiler
sticks to her challenge.  In this way she would win the play following the
rules in Definition~\ref{def:bulychev_game}, but such \emph{procrastinating} 
behaviour of \duplicator is not rewarded in our game. Instead, \duplicator has 
to eventually move to $c_1$, matching the challenge, if she wants to win the play.  

\begin{figure}[ht]
\centering
\begin{tikzpicture}
\scriptsize
\node (u_)  {$s_0$};
\node[right of=u_,xshift=20pt] (c1_) {$c_1$};
\node[left of=u_,xshift=-20pt] (c2_) { $c_2$};
\node[below of=c1_] (v_) {$t_0$};
\node[below of=c2_] (w_) {$t_1$};

\path[->]
  (u_) edge node[above] {$a$} (c1_) edge node[above] {$b$} (c2_)
  (v_) edge[bend left] node[above] {$\tau$} (w_) edge node [right] {$a$}  (c1_)
  (w_) edge[bend left] node[below] {$\tau$} (v_) edge node [left] {$b$} (c2_)
;

\node (u) [right of=c1_, xshift=60pt]  {$u$};
\node (v1) [below of=u] {$v_1$};
\node (v0) [left of=v1,xshift=-20pt] {$v_0$};
\node (v2) [right of=v1,xshift=20pt] {$v_2$};

\path[->]
  (v0) edge node[above left] {$a$} (u)
  (v1) edge node[left] {$b$} (u)
  (v2) edge node[above right] {$c$} (u)
  (v0) edge node[above] {$\tau$} (v1)
  (v1) edge node[above] {$\tau$} (v2)
  (v2) edge [bend left] node[below] {$\tau$} (v0);
;

\end{tikzpicture}
\caption{LTSs illustrating some consequences and
subtleties of using challenges.}
\label{fig:need_for_challenges}
\end{figure}
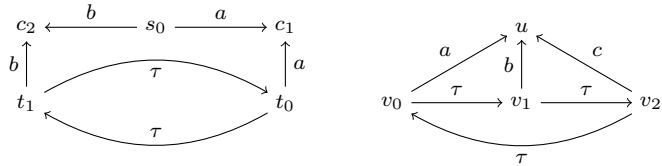

\end{exa}

\begin{defi}
\label{def:branching_bisimulation_game}
A \emph{branching bisimulation (\bb) game} on an LTS  $L$ is played
by players \spoiler and \duplicator on an arena  of \spoiler-owned
configurations $\sconf{(s,t),c,r}$ and \duplicator-owned configurations
$\dconf{(s,t),c,r}$, where $((s,t),c,r) \in \pos \times \chaldagger \times
\rew$, and $\pos$ and $\chal$ are as before, $\chaldagger = \chal \cup \{ \dagger \}$ is the set of challenges where $\dagger$ signifies absence of a challenge, and $\rew = \{ *, \checkmark \}$ is a set of \emph{rewards}.
By convention, we write $\conf{(s,t),c,r}$ if we do not care about the
owner of the configuration. 
\begin{itemize}
\item \spoiler moves from a configuration $\sconf{(s,t),c,r}$ by:
  \begin{enumerate}
    \item selecting $s \pijl{a} s'$ and moving to
$\dconf{(s,t),(a,s'),*}$ if
$c = (a,s')$ or $c = \dagger$, and to $\dconf{(s,t),(a,s'),\checkmark}$,
otherwise; or

    \item picking some $t \pijl{a} t'$ and moving to
$\dconf{(t,s),(a,t'),\checkmark}$.
  \end{enumerate}

\vspace{0.3cm}
\item \duplicator responds from a configuration $\dconf{(u,v),(a,u'),r}$ by:
\begin{enumerate}

 \item \label{bbg_terminate} not moving if $a = \tau$ and continuing in
    configuration $\sconf{(u',v),\dagger,\checkmark}$, or, 

  \item moving $v \pijl{a} v'$ if available
    and continuing in configuration $\sconf{(u',v'), \dagger, \checkmark}$, or

 \item moving $v \pijl{\tau} v'$ if available
    and continuing in configuration $\sconf{(u,v'), (a,u'), *}$.
\end{enumerate}
\end{itemize}
\duplicator wins a finite play starting in a configuration
$\conf{(s,t),c,r}$ if \spoiler gets stuck, and she wins an infinite
play if the play yields infinitely many $\checkmark$ rewards. All
other plays are won by \spoiler. 
We say that a configuration is won by a player when 
she has a strategy that wins all plays starting in it.
{\em Full} plays of the game start in a configuration  $\sconf{(s,t),\dagger,*}$;
we say that 
 \duplicator wins the \bb-game for a position $(s,t)$, if the configuration $\sconf{(s,t),\dagger,*}$
is won by her; in this case, we write $s \bbg t$. Otherwise, we say that \spoiler wins that game.

Note that by definition both players strictly alternate their moves along plays.
\end{defi}
We have not yet explained the rewards that are handed to \duplicator in some of \spoiler's moves.
The need for those are illustrated by the following example.

\begin{exa}
\label{ex:need_for_rewards}
Again consider the first LTS depicted in Figure~\ref{fig:need_for_challenges}.
Suppose \spoiler tries to disprove (again in vain) that $s_0$ and $t_0$ are branching bisimilar, 
and challenges \duplicator by playing $s_0 \pijl{b} c_2$.
The only response for \duplicator is to move $t_0 \pijl{\tau} t_1$, for which she gets a $*$ reward,
and the pending challenge $(b, c_2)$ is kept, generating the new configuration $\sconf{(s_0, t_0), (b, c_2), *}$.
If \spoiler again plays $s_0 \pijl{b} c_2$, \duplicator can win by playing $t_1 \pijl{b} c_2$. If \spoiler
changes, and plays $s_0 \pijl{a} c_1$, then \duplicator can only respond with $t_1 \pijl{\tau} t_0$, getting a $*$
reward and continuing from configuration $\sconf{(s_0, t_0), (a, c_1), *}$. \spoiler can continue switching between
the moves $s_0 \pijl{b} c_2$ and $s_0 \pijl{a} c_1$, and \duplicator will never be able to match the move
directly. However, note that $\checkmark$s are only awarded whenever \spoiler switches away from his current challenge,
and \duplicator still wins this play. If we were not to award $\checkmark$s whenever \spoiler switches away from the current
challenge, \spoiler would win the game starting in $(s_0, t_0)$, whereas $s_0 \bbbisim t_0$.

Now consider the second LTS in Figure~\ref{fig:need_for_challenges}.
In this LTS, $v_0 \bbbisim v_1 \bbbisim v_2$. Now, if the game starts in $\sconf{(v_1, v_0), \dagger, *}$, \spoiler can
play to $\dconf{(v_0, v_1), (a, u), \checkmark}$, to which \duplicator can only respond by playing to $\sconf{(v_0, v_2), (a, u), *}$.
This play can be extended, each time not leaving any choice for \duplicator, as $\sconf{(v_0, v_2), (a, u), *} \to \dconf{(v_2, v_0), (c,u), \checkmark} \to
\sconf{(v_2, v_1), (c,u), *} \to \dconf{(v_1, v_2), (b,u), \checkmark} \to \sconf{(v_1, v_0), (b,u), *} \to \dconf{(v_0, v_1), (a, u), \checkmark} \to \cdots$.
If we do not reward \duplicator when \spoiler plays the second move, this play is winning for \spoiler, whereas $v_1 \bbbisim v_0$.
\end{exa}

In \cite{branching_games_2016} we proved that the \bb-game captures branching bisimilarity. The result also follows
from our results concerning the generic games in Section~\ref{sec:generic_games}.
\begin{thm}\label{th:bbbisim_is_bbg} We have $\bbbisim\, =\, \bbg$.
\end{thm}


\section{A Generic Bisimulation for Abstraction}
\label{sec:generic}

Branching bisimilarity induces a weak semantics of systems; it
achieves this by partially abstracting from $\tau$ moves so that
they need not be matched one by one as is the case for strong
bisimulation.  A much more direct way of abstracting from $\tau$
moves is to modify strong bisimulation in such a way that each
concrete transition $\pijl{a}$ is `weakly' matched using a weak
transition $\xRightarrow{\,a~}$, where the weak transition relation
$\Longrightarrow$ is essentially $\taustar \circ \longrightarrow
\circ \taustar$.  The relation obtained this way is called \emph{weak
bisimulation}, which is the basis for weak bisimilarity.

As explained in detail in~\cite{vanglabbeek_branching_1996}, the
definition of branching bisimulation can be obtained from the
definition of weak bisimulation by imposing additional conditions
on the relation. By varying these conditions, two further behavioural
equivalences can be obtained, \viz\ \emph{delay bisimilarity} and
\emph{$\eta$-bisimilarity}.  This is nicely visualised
in~\cite{vanglabbeek_branching_1996} using the diagram of
Figure~\ref{fig:generic_picture}.

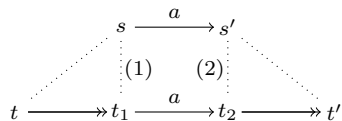
\begin{figure}[hbpt]
\centering
\begin{tikzpicture}[node distance=40pt,inner sep = 1pt, minimum size=10pt,initial text={}]
\tikzstyle{state}=[draw=none]
\scriptsize
\node[state] (t) {$t$};
\node[state,right of=t] (t1) {$t_1$};
\node[state,right of=t1] (t2) {$t_2$};
\node[state,right of=t2] (t') {$t'$};

\node[state,above of=t1,yshift=-8pt] (s)  {$s$};
\node[state,right of=s] (s')  {$s'$};

\path[->]
  (t) edge (t1)
  (t) edge [shorten >=2pt](t1)
  (t1) edge node[above] {$a$} (t2)
  (t2) edge (t')
  (t2) edge [shorten >=2pt](t');

\path[->]
  (s) edge node[above] {$a$} (s');

\path[dotted]
  (s) edge (t)
  (s) edge node[right] {(1)} (t1)
  (s') edge node[left] {(2)} (t2)
  (s') edge (t')
;
\end{tikzpicture}
\caption{Weakly matching a concrete transition $s \pijl{a} s'$ by an abstract transition
from $t \xRightarrow{\,a~} t'$. By requiring (1) and (2) to hold, or not, variations
on weak bisimulations are obtained: $\eta$-bisimulation is obtained by requiring
(1), delay bisimulation is obtained by requiring (2) and branching bisimulation is
obtained by requiring both (1) and (2).
}
\label{fig:generic_picture}
\end{figure}

In this section we give a parametric bisimulation, which we refer
to as \emph{$(x,y)$-generic bisimulation}, that gives a unified
definition for branching bisimulation and the three other abstract bisimulations.
This definition in essence reflects the diagram of
Figure~\ref{fig:generic_picture}: we use parameters $x$ and $y$ to
indicate which extra conditions (if any) are imposed on the definition
of weak bisimulation.  Before we formally state our parametric
bisimulation, we briefly recall the definitions of weak bisimulation,
delay bisimulation and $\eta$-bisimulation and state a few well-known
facts about them.

\begin{defi}[{\cite{milner_calculus_1980}}]
\label{def:weak_bisimulation}
A symmetric relation $R \subseteq S \times S$ is 
a \emph{weak bisimulation} whenever $s \rel{R} t$ and 
$s \pijl{a} s'$ imply either:
\begin{itemize}
\item $a = \tau$ and $s' \rel{R} t$, or

\item there exist states $t', t_1,t_2$ such that $t \taustar t_1
\pijl{a} t_2 \taustar t'$ and $s'\rel{R}t'$.

\end{itemize}
We write $s \wbisim t$ and say that $s$ and $t$ are weak bisimilar
iff there is a weak bisimulation $R$ such that $s \rel{R} t$.
Typically we simply write $\wbisim$ to denote {\em weak bisimilarity}.\footnote{The
definition in \cite{milner_calculus_1980} uses the weak transition $s \stackrel{a}\Rightarrow s'$
to denote $s \taustar s'' \pijl{a} s''' \taustar s'$ in our notation.}

\end{defi}

Note that compared to branching bisimulation, weak bisimulation drops all
conditions imposed on states reached silently before \emph{and}
after the execution of a weakly matching $a$-transition. The difference
between the two relations is nevertheless subtle, as illustrated by the
LTS depicted in Figure~\ref{fig:example_weak_vs_branching}, taken from
Korver~\cite{korver_computing_1992}.
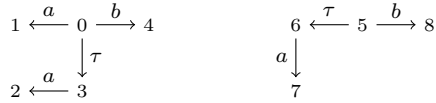
\begin{figure}[hbpt]
  \begin{center}
  \begin{tikzpicture}
    [node distance=25,inner sep = 1pt,minimum size=10pt]

     \tikzstyle{state}=[circle, draw=none]
     \tikzstyle{transition}=[->,>=stealth']

      \node[state] (t0) {\tiny 0};
      \node[state] [left of=t0] (t1) {\tiny 1};
      \node[state] [below of=t1] (t2) {\tiny 2};
      \node[state] [right of=t2] (t3) {\tiny 3};
      \node[state] [right of=t0] (t4) {\tiny 4};
      \draw [->] (t0) edge node[above] {\scriptsize $a$} (t1);
      \draw [->] (t0) edge node[above] {\scriptsize $b$} (t4);
      \draw [->] (t0) edge node[right] {\scriptsize $\tau$} (t3);
      \draw [->] (t3) edge node[above] {\scriptsize $a$} (t2);

      \node[state, right of=t0,xshift=80pt] (s0) {\tiny 5};
      \node[state] [left of=s0] (s1) {\tiny 6};
      \node[state] [below of=s1] (s2) {\tiny 7};
      \node[state] [right of=s0] (s3) {\tiny 8};
      \draw [->] (s0) edge node[above] {\scriptsize $\tau$} (s1);
      \draw [->] (s1) edge node[left] {\scriptsize $a$} (s2);
      \draw [->] (s0) edge node[above] {\scriptsize $b$} (s3);

  \end{tikzpicture}
  \end{center}
\caption{An illustration of the difference between weak bisimulation
and branching bisimulation: states $0$ and $5$ are weak bisimilar
but not branching bisimilar.}
\label{fig:example_weak_vs_branching}
\end{figure}
It is not too hard to check that states $0$ and $5$ are weak
bisimilar (and, in fact, also delay bisimilar).\label{discussion}  
Nonetheless, the argument why these states are not
branching bisimilar is more subtle and is best explained using the
branching bisimulation game of the previous section, as we also
indicated in~\cite{branching_games_2016}.  Indeed, \spoiler wins
the branching bisimulation game by moving $0 \pijl{a} 1$.  \duplicator
can only respond to this challenge by moving $5 \pijl{\tau} 6$.
Now, continuing from $\sconf{(0, 5), (a, 1),*}$ \spoiler plays her
second option and challenges \duplicator to mimic move $0 \pijl{b}
4$, something that  \duplicator cannot match.\label{weak_vs_branching_example}

\begin{defi}[{\cite{milner_modal_1981,vanglabbeek_branching_1996}}]
\label{def:delay_bisimulation}
A symmetric relation $R \subseteq S \times S$ is 
a \emph{delay bisimulation} whenever $s \rel{R} t$ and
$s \pijl{a} s'$ imply either:
\begin{itemize}
\item $a = \tau$  and $s' \rel{R} t$, or

\item there exist states $t', t_1$ such that $t \taustar t_1 \pijl{a}
t'$ and $s'\rel{R}t'$.

\end{itemize}
We write $s \dbisim t$ and say that $s$ and $t$ are delay bisimilar
iff there is a delay bisimulation $R$ such that $s \rel{R} t$.
Typically we simply write $\dbisim$ to denote {\em delay bisimilarity}.

\end{defi}

We remark that the original `delay bisimulation' originates with Milner
in~\cite{milner_modal_1981}, where it is called \emph{observation equivalence}.
It is not presented in a coinductive way but by means of a sequence of
\emph{approximations}. Moreover, the definition from~\cite{milner_modal_1981}
is sensitive to divergences, and coincides with our definition only for
non-divergent systems. The fact that \emph{weak transitions} where the execution of an
observable action can be preceded, but not followed, by a finite
sequence of internal actions, is indeed the seed for the--by now--`classical' 
notion of delay bisimulation that we present
above, and which stems from~\cite{vanglabbeek_branching_1996}.

Observe that delay bisimulation essentially renders the existentially
quantified state $t_2$ in the definition of weak bisimulation
superfluous as it requires it to coincide with $t'$ in that definition.
We note that from this it also immediately follows that every delay
bisimulation is a weak bisimulation, but not \emph{vice versa}.

\begin{defi}[{\cite{baeten_another_1987}}]
\label{def:eta_bisimulation}
A symmetric relation $R \subseteq S \times S$ is an $\eta$-{\em
bisimulation} whenever $s \rel{R} t$ and $s \pijl{a} s'$ imply
either:
\begin{itemize}
\item $a = \tau$ and $s' \rel{R} t$, or

\item there exist states $t', t_1,t_2$ such that $t \taustar t_1
\pijl{a} t_2 \taustar t'$, $s \rel{R} t_1$ and $s'\rel{R}t'$.

\end{itemize}
We write $s \ebisim t$ and say that $s$ and $t$ are $\eta$-bisimilar 
iff there is an  $\eta$-bisimulation $R$ such that $s \rel{R} t$.
Typically we simply write $\ebisim$ to denote $\eta$-{\em bisimilarity}.
\end{defi}

As shown in \cite{vanglabbeek_branching_1996},
each of the relations we define above are equivalence relations.
\begin{thm}[\cite{vanglabbeek_branching_1996}]
\label{th:equivalences_edw} 
The relations $\ebisim,\dbisim$ and $\wbisim$ are equivalence relations.

\end{thm}

From their definitions it immediately follows that the four equivalences $\bbbisim,\ebisim,\dbisim$ and $\wbisim$, are ordered according to the lattice
depicted in Figure~\ref{fig:lattice}, as was also shown in~\cite{sangiorgi_introduction_2012}. The equivalences lower in the lattice are coarser than the equivalences higher in the lattice.
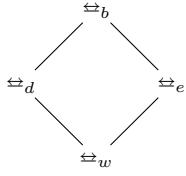
\begin{figure}[hbpt]
\centering
\begin{tikzpicture}[node distance=40pt,inner sep = 1pt, minimum size=10pt,initial text={}]
\tikzstyle{state}=[draw=none]
\scriptsize
\node[state] (bb) {$\bbbisim$};
\node[state,below right of=bb] (eb) {$\ebisim$};
\node[state,below left of=bb] (db) {$\dbisim$};
\node[state,below right of=db] (wb) {$\wbisim$};

\path[draw]
  (bb) edge (eb) edge (db)
  (wb) edge (eb) edge (db)
;
\end{tikzpicture}
\caption{The lattice of abstract behavioural equivalences.  }
\label{fig:lattice}

\end{figure}

\newcommand{\staustar}{\Longrightarrow}
\newcommand{\gentaustar}[1]{\ensuremath{\taustar}_{#1}}
\newcommand{\sgentaustar}[1]{\ensuremath{\staustar}_{#1}}
We next exploit the similarity between the four abstract bisimulations
by explicitly giving a parameterised definition
from which the original definitions can be recovered.  Before we
state this parameterised definition, we introduce some auxiliary
notation to facilitate a concise parametric definition.\medskip

Let $R \subseteq S \times S$ and let $s,s',t \in S$.
We write $s \gentaustar{o,R,t} s'$ iff $s \taustar s'$, while we write
$s \gentaustar{b,R,t} s'$ iff $s \taustar s'$, $t \rel{R} s$ and $t \rel{R} s'$. 
In this way, when $x \in \{o,b\}$, a {\em general weak transition} 
$s \gentaustar{x,R,t} s'$  is just a plain weak transition  $s \taustar s'$, when $x=o$, while, when $x=b$, we additionally impose the `context' conditions $t \rel{R} s$ and $t \rel{R} s'$.
In a similar way, we write $s \sgentaustar{x,R,t} s'$ iff
$s \taustar s'$ and either $x = o$, or $x = b$ and there is a
finite sequence of transitions $s = s_0 \pijl{\tau} s_1 \pijl{\tau} \cdots
\pijl{\tau} s_n = s'$ such that for all $i$, $t \rel{R} s_i$. 
It is clear that $s \sgentaustar{x,R,t} s'$ implies $s \gentaustar{x,R,t} s'$,
because we are not only imposing $t \rel{R} s_0$ and $t \rel{R} s_n$, but
also that $t \rel{R} s_i$ for any intermediate state. Instead, 
the implication from right to left generally does not hold,
because those additional conditions for the intermediate states are not 
(explicitely) imposed by the definition of our general weak transitions.

\begin{defi}
\label{def:generic_bisimulation}
For  $x,y \in \{o,b\}$, 
a symmetric relation $R \subseteq S \times
 S$ is an $(x,y)$-\emph{generic bisimulation}, whenever $s
 \rel{R} t$ and $s \pijl{a} s'$ imply either:
\begin{itemize}
 \item $a = \tau$ and $s' \rel{R} t$, or 
 \item there exist states $t', t_1,t_2$ such that $t
 \gentaustar{x,R,s} t_1 \pijl{a} t_2 \gentaustar{y,R,s'} t'$ and $s' \rel{R} t'$.  
\end{itemize}
 Now, for the corresponding values of $x$ and $y$, we write $s \gbbbisim
 t$ and say that $s$ and $t$ are $(x,y)$-generic bisimilar iff there is an
 $(x,y)$-generic bisimulation $R$ such that $s \!\rel{R} t$.
 Typically, we simply write $\gbbbisim$ to denote  $(x,y)$-\emph{generic
 bisimilarity}.

\end{defi}
As formally claimed by the proposition below, the above parametric
bisimulation definition captures all four abstract bisimulations.

\begin{prop}\label{prop:cases_generic_bisim}
We have the following correspondences for $R \subseteq S \times S$:
\begin{itemize}
  \item $R$ is a weak bisimulation iff it is an $(o,o)$-generic bisimulation;
  \item $R$ is a delay bisimulation iff it is an $(o,b)$-generic bisimulation;
  \item $R$ is an $\eta$ bisimulation iff it is a $(b,o)$-generic bisimulation;
  \item $R$ is a branching bisimulation iff it is a $(b,b)$-generic bisimulation.
\end{itemize}

\end{prop}
\begin{proof}
First observe that the definition of $(o,o)$-generic bisimulation coincides exactly
with the definition of weak bisimulation. Likewise, the definition of
$(b,o)$-generic bisimulation reduces to the definition of
$\eta$-bisimulation. 

As for the remaining two relations, we find that the implication
from right to left follows because the requirements 
for being a $(b,b)$-generic bisimulation (resp.\ an $(o,b)$-generic bisimulation)
are stronger than the requirements
for being a branching bisimulation (resp.\ a delay bisimulation).
For the implication from left to right we note that for any
delay bisimulation (resp.\ a branching bisimulation) $R$, then $R$ is
an $(o,b)$-generic bisimulation (resp.\ a $(b,b)$-generic bisimulation) by
simply choosing $t_2 = t'$.
\end{proof}

It follows from the above theorem that, by equipping the set
$\{o,b\}$ with an ordering $\leq$, where $x \leq x$ for all
$x$ and $o < b$, and lifting this ordering to pairs taken from
$\{o,b\} \times \{o,b\}$ we can recover the lattice of Figure~\ref{fig:lattice}.
We furthermore note that in view of Proposition~\ref{prop:cases_generic_bisim}
and Theorems~\ref{th:bbbisim_is_equivalence} and~\ref{th:equivalences_edw} we
have the following corollary.

\begin{cor}
Let $x,y \in \{o,b\}$. Then $\gbbbisim$ is an equivalence relation.
\end{cor}

Additionally, it is not hard to prove that $(x,y)$-generic bisimilarity, like branching bisimilarity, satisfies the stuttering property. 
\begin{lem}
  \label{lem:gen_bisim_stuttering_property}
Let $x,y \in \{o,b\}$. Then $\gbbbisim$ satisfies the stuttering property (Definition~\ref{def:stuttering property}).
\end{lem}
\begin{proof}
The proof is essentially a translation to the general case of that given in \cite[Lemma 4.9.2]{sangiorgi_introduction_2012},  for the particular case of branching
bisimulation. 
Let $t_0 \pijl{\tau} t_1 \cdots \pijl{\tau} t_k$ with $t_0 \gbbbisim t_k$.
We define relation $R$ as follows:
\[
R = \{ (t_0, t_i), (t_i, t_0) \mid 0 \leq i < k \} \cup \gbbbisim
\]
Let us see that $R$ is an $(x,y)$-generic bisimulation relation, by proving the transfer condition for all pairs $t_0 \mathrel{R} t_i$.
\begin{itemize}
\item
When we have 
 $t_i\pijl{a} t_i'$, we immediately obtain
  $t_0 \gentaustar{x,R,t_i} t_i \pijl{a} t_i' \gentaustar{y,R,t_i'} t_i'$,  which is trivially true because $t_i \mathrel{R} t_i$. Moreover, it is equally trivial that $t_i' \rel{R} t_i'$.
\item
When we have 
 $t_0\pijl{a} t_0'$, from $t_0 \gbbbisim t_k$ we infer that
 we  have either:
\begin{itemize}
 \item $a = \tau$ and $t_0' \rel{R} t_k$, or 
 \item there exist states $t_k', t_k'',t_k'''$ such that 
 $t_k \gentaustar{x,R,t_0} t_k' \pijl{a} t_k'' \gentaustar{y,R,t_0'} t_k'''$ and $t_0' \rel{R} t_k'''$.  
\end{itemize}
In both cases we can start the matching computation from $t_i$ by means of the weak transition 
$t_i \taustar t_k$, thus getting either:
\begin{itemize}
 \item $a = \tau$ and we have   
  $t_i \gentaustar{x,R,t_0} t_{k-1} \pijl{a} t_k \gentaustar{y,R,t_0'} t_k$, because 
  $t_0 \rel{R} t_{k-1}$ and $t_0' \rel{R} t_k$, or 
 \item there exist states $t_k', t_k'',t_k'''$ such that 
 $t_i \gentaustar{x,R,t_0} t_k' \pijl{a} t_k'' \gentaustar{y,R,t_0'} t_k'''$ and $t_0' \rel{R} t_k'''$,
\end{itemize}
thus concluding the proof. \qedhere
\end{itemize}
\end{proof}
Using the stuttering property, we find
that we may rephrase $(x,y)$-generic bisimilarity as per the following theorem.
Essentially, this says we may require the intermediate states along the stuttering paths of the transfer condition to be related without changing the definition.
\begin{thm}\label{th:alternative_gen_bisim}
 For every $x,y \in \{o,b\}$ and $\bar{s},\bar{t} \in S$, we have
 $\bar{s} \gbbbisim \bar{t}$ iff $\bar{s} \rel{R} \bar{t}$ 
for some symmetric relation $R \subseteq S \times
 S$ satisfying that whenever $s \rel{R} t$ and $s \pijl{a} s'$ 
 we have:
\begin{itemize}
 \item $a = \tau$ and $s' \rel{R} t$, or 
 \item there exist states $t', t_1,t_2$ such that $t
 \sgentaustar{x,R,s} t_1 \pijl{a} t_2 \sgentaustar{y,R,s'} t'$ and $s' \rel{R} t'$.  
\end{itemize}
\end{thm}
\begin{proof}
Immediate. Even if the original definition of the parameterised bisimulations used 
$\gentaustar{x,R,s}$ and $\gentaustar{y,R,s'}$ 
instead of 
$\sgentaustar{x,R,s}$ and $\sgentaustar{y,R,s'}$, 
the additional conditions added by the latter relations follow from
those imposed by the former ones per
Lemma~\ref{lem:gen_bisim_stuttering_property}.
\end{proof}

As a consequence of the above theorem, we may henceforth and without
loss of generality, use a stronger definition for $(x,y)$-generic bisimilarity
than the one stated in Definition~\ref{def:generic_bisimulation},
\viz\ the one implied by Theorem~\ref{th:alternative_gen_bisim}.


\section{Generic Bisimulation Games}
\label{sec:generic_games}

As we illustrated in the previous section, branching bisimulation
is one of four bisimulations for providing a weak semantics to
systems. Rather than defining a dedicated game for each of these
bisimulations, we show that these four relations are captured by
instances of a generic bisimulation game, which we introduce in
Section~\ref{sec:generic_game_def}. In particular, in
Section~\ref{sec:generic_game_sound_complete} we show that this
game is sound and complete for $(x,y)$-generic bisimilarity.  Moreover, in
Section~\ref{sec:generic_specific} we shall see that the 
branching bisimulation game that we defined
in Section~\ref{sec:games} specialises our generic
bisimulation game. Finally, in Section~\ref{sec:variations}
we discuss two variations on the game for characterising
$(x,y)$-generic bisimilarity.

\subsection{A Generic Bisimulation Game}
\label{sec:generic_game_def}

As we explained in Section~\ref{sec:games}, \spoiler's and \duplicator's
role in the games we consider is to \emph{refute} and \emph{prove},
respectively, that states in an LTS can be related.  In general,
given some game characterising a relation on states, we can
obtain a weaker relation by either restricting \spoiler's capabilities,
or by offering more liberal options to \duplicator.  The game we
present in this section essentially does the latter: it will---depending
on which relation we wish to capture---give \duplicator fewer or
more options to respond to \spoiler's challenges.
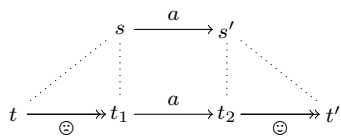
\begin{figure}[h]
\centering
\begin{tikzpicture}[node distance=40pt,inner sep = 1pt, minimum size=10pt,initial text={}]
\tikzstyle{state}=[draw=none]
\scriptsize
\node[state] (t) {$t$};
\node[state,right of=t] (t1) {$t_1$};
\node[state,right of=t1] (t2) {$t_2$};
\node[state,right of=t2] (t') {$t'$};

\node[state,above of=t1,yshift=-8pt] (s)  {$s$};
\node[state,right of=s] (s')  {$s'$};

\path[->]
  (t) edge node [below] {$\frownie$} (t1)
  (t) edge [shorten >=2pt](t1)
  (t1) edge node[above] {$a$} (t2)
  (t2) edge node [below] {$\smiley$} (t')
  (t2) edge [shorten >=2pt](t');

\path[->]
  (s) edge node[above] {$a$} (s');

\path[dotted]
  (s) edge (t)
  (s) edge (t1)
  (s') edge (t2)
  (s') edge (t')
;
\end{tikzpicture}
\caption{Diagram that illustrates the different scenarios that need to
be considered when `weakly' matching a challenge by \spoiler}
\label{fig:diagram}
\end{figure}
Before we move to defining our generic bisimulation game, consider
the diagram above (see also Figure~\ref{fig:generic_picture}). It
illustrates that a `challenge' $s \pijl{a} s'$, posed by \spoiler, is
to be matched `weakly' through $t \taustar t_1 \pijl{a} t_2 \taustar
t'$. Intuitively, we can capture this weak match by recording the progress
made by \duplicator using a `pebble' which \duplicator can push along
the $\tau$ or $a$-edges in the transition system.
Moreover, we need to record whether \duplicator is pushing the
pebble along the `$\frownie$-decorated' $\tau$-path or the
`$\smiley$-decorated' $\tau$-path if we wish to prevent her from
meeting \spoiler's challenge by taking multiple (or no) $a$-transitions.
The relation one intuitively obtains in this way characterises
weak bisimilarity. 

On the other hand, we must disallow \duplicator to merely push the
pebble when we wish to characterise branching bisimilarity. Instead,
she must move to a configuration with an updated position so that
\spoiler can renew her challenge for this new position. Indeed, since
branching bisimilarity has the stuttering property, \duplicator
should be able to move along $\tau$-paths in which every state is
related to the state from which \spoiler posed her challenge. A
similar observation can be made for $\eta$-bisimilarity: in this
case, \duplicator must update her position while traversing the
$\tau$-path to the state from which a matching $a$-action can be
performed, while, after the match, she may move `freely' by pushing
the pebble to a state that matches \spoiler's target state.  In
delay bisimilarity, \duplicator can first push the pebble along the
$\frownie$-decorated $\tau$-path to a state from which a matching
action can be performed, but she must update the position
immediately upon choosing this matching transition.

This means that we can tweak \duplicator's powers by enabling or
disabling those rules that would allow her to merely push pebbles.
That is, in case of branching bisimilarity, she is never allowed
to use such rules, whereas for weak bisimilarity, she can always
employ such rules. For the other two relations, we can disable
the rules associated with pushing the pebble essentially by
relying on whether \duplicator is on a $\frownie$ or $\smiley$
$\tau$-path. Below, we make these ideas more precise.\medskip


\begin{defi}
\label{def:generalised game} For each $E \subseteq \{ \smiley, \frownie \} $,
a \emph{generic bisimulation (gb) game} (also noted as $E$-gb
game when we want to stress the corresponding set $E$) 
on an LTS $L$ is played by players \spoiler and \duplicator
on an arena  of \spoiler-owned configurations $\sconf{(s,t),c,m,r}$
and \duplicator-owned configurations $\dconf{(s,t),c,m,r}$, where
$(s,t) \in \pos$, $c \in \chaldagger$, $m \in \mov$ and $r \in \rew$, and $\pos$, $\chaldagger$ and $\rew$ are as before, and $\mov = (S \times \{\frownie,\smiley\}) \cup \{\dagger\}$ is the set of \emph{partial matches}.
We write $\conf{(s,t),c,m,r}$ if we do not care about
the owner of the configuration. \spoiler's and \duplicator's possible moves
are given by the following rules: 
\begin{itemize}
\item From a configuration $\sconf{(s,t),c,m,r}$, \spoiler can:

  \begin{enumerate}
    \item \label{gbbg_move_unchanged} move to $\dconf{(s,t),c,m,*}$ if $c \neq \dagger$, or

    \item for some $s \pijl{a} s'$, move to either:
       \begin{enumerate}
   \item\label{gbbg_spoiler_fresh_challenge} $\dconf{(s,t),(a,s'),(t,\frownie),*}$, if $c = \dagger$, or
\item\label{gbbg_spoiler_renew_challenge}
       $\dconf{(s,t),(a,s'),(t,\frownie),\checkmark}$, if $c \neq (a,s')$
      \end{enumerate}
    \item\label{gbbg_spoiler_switches} for some $t \pijl{a} t'$, move to:
$\dconf{(t,s),(a,t'),(s,\frownie),\checkmark}$.

  \end{enumerate} 
\item[]

\item From a configuration $\dconf{(u,v),(a,u'),(\bar{v},f),r}$, \duplicator can:
\begin{enumerate}

 \item \label{gbbg_terminate} 
move to $\sconf{(u',\bar{v}),\dagger,\dagger,\checkmark}$ when $a = \tau$, or

 \item if $f = \frownie$ and $\bar{v} \pijl{a} v'$, move to one of the following:
\begin{enumerate}
   \item \label{gbbg_match_and_move}  $\sconf{(u',v'),(a,u'),(v',\smiley),*}$, in any case, or
   \item \label{gbbg_match_move_and_terminate} $\sconf{(u',v'),\dagger,\dagger,\checkmark}$, in any case, or
   \item \label{gbbg_match_and_postpone} $\sconf{(u,v),(a,u'),(v',\smiley),*}$, only if  $\smiley \in E$
 \end{enumerate} 
     
 \item for some $\bar{v} \pijl{\tau} v'$, move to one of the following:
 \begin{enumerate}
   \item \label{gbbg_push_pebble_and_move} $\sconf{(u,v'),(a,u'),(v',f),*}$,  in any case, or
   \item \label{gbbg_push_pebble_and_terminate} $\sconf{(u',v'),\dagger,\dagger,\checkmark}$, only if $f = \smiley$
   \item \label{gbbg_push_pebble_and_wait} $\sconf{(u,v),(a,u'),(v',f),*}$, only if  $f \in E$
 \end{enumerate} 

\end{enumerate}

\end{itemize}
\duplicator wins a finite play starting in a configuration
$\conf{(s,t),c,m,r}$ if \spoiler gets stuck, and she wins an infinite
play if the play yields infinitely many $\checkmark$ rewards.  All
other plays are won by \spoiler.  A player wins a configuration
when she has a strategy that wins all plays starting in it.
\emph{Full} plays of the game start in a configuration
$\sconf{(s,t),\dagger,\dagger,*}$; we say that \duplicator wins the
game for a position $(s,t)$, if the configuration
$\sconf{(s,t),\dagger,\dagger,*}$ is won by her.  In this case, we
write $s \gbbg_E t$. Otherwise, we say that \spoiler wins that game.
\end{defi}
Note that the moves in the game, like in the branching bisimulation
game of Definition~\ref{def:branching_bisimulation_game}, alternate
between \spoiler and \duplicator. \spoiler typically issues a challenge,
which, in case of rule~\eqref{gbbg_move_unchanged} consists of a previously
coined challenge, in case of 
rule~\eqref{gbbg_spoiler_fresh_challenge} is a fresh challenge, in
case of rule~\eqref{gbbg_spoiler_renew_challenge} overrides a pre-existing
challenge,
and in case of rule~\eqref{gbbg_spoiler_switches}, is a fresh challenge but
issued from the second state in the position of the configuration. Intuitively,
when \duplicator is to respond from a configuration, she will try to
`weakly' match the challenge using any of her rules.  
Using rule~\eqref{gbbg_terminate}, she can decide
to instantly meet a $\tau$-challenge and leave the task of disproving that the states
in the resulting position are not related, to \spoiler. Alternatively, 
\duplicator can also face a challenge by making (at least) one (explicit) move.
The actual matching
(or cancelling) of \spoiler's challenge is captured in rules~\eqref{gbbg_match_and_move}-\eqref{gbbg_match_and_postpone}.
Additionally, \duplicator can 
use rules~\eqref{gbbg_push_pebble_and_move}-\eqref{gbbg_push_pebble_and_wait} to 
`push' the pebble from the position on the partial match closer to a state
in which she can finally match the challenge as above (when the pebble is still on a
$\frownie$-decorated path in terms of Figure~\ref{fig:diagram}), or she
can continue to `push' the pebble closer to the final position
where she can finally cancel the challenge. Note that in rules~\eqref{gbbg_match_and_move}, \eqref{gbbg_match_and_postpone}, \eqref{gbbg_push_pebble_and_move}, and \eqref{gbbg_push_pebble_and_wait} the state $v'$ at the match reflects the progress of \spoiler until the
completion of the matching that cancels it. The matching is completed by
updating the position in Duplicator's rules~\eqref{gbbg_terminate}, \eqref{gbbg_match_move_and_terminate},
and \eqref{gbbg_push_pebble_and_terminate}.
\medskip

Let the set $E(x,y)$ be the smallest set such that $\frownie \in
E(o,y)$ and $\smiley \in E(x,o)$ for all $x,y \in \{o,b\}$.  Using
this notation, we can instantiate our \gb-games
to games for the four instances of $(x,y)$-generic bisimulation.
That is, we have:
\begin{itemize}
  \item the branching bisimulation game, when $E(x,y) = \emptyset$
  \item the delay bisimulation game, when $E(x,y) = \{ \frownie \}$
  \item the $\eta$-bisimulation game, when $E(x,y) = \{ \smiley \}$
  \item the weak bisimulation game, when $E(x,y) = \{ \smiley, \frownie \}$.
\end{itemize}

Therefore, the \emph{faces} $\frownie$ and $\smiley$, included in
each of the sets $E(x,y)$, state exactly the degree of freedom
of the corresponding bisimulation relations. They govern \emph{when}
(before and/or after the matching $a$ actions) we do not have new bisimulation
obligations that would generate new \emph{positions} to check.
That means that the \emph{fewer} faces we have, the \emph{finer} the bisimulation relation
we get.
\medskip

To be exact, our generic game could be presented as a collection of 
(four) games in a single family. The `unity' of this family is supported by the single
(parameterised) definition of its elements. This allows us to claim the following 
correspondence between our
\gb-game and $(x,y)$-generic bisimilarity.
\begin{restatable}{thm}{correspondencetheorem}
\label{thm:correspondence}
For all $x,y \in \{o,b\}$, we have $s \gbbbisim t$
iff $s \gbbg_{E(x,y)} t$.
\end{restatable}
The proof of this claim is discussed in detail in the next section. We
first illustrate the definition using a small example.
\begin{exa}\label{ex:game_illustrated}
Consider the transition system depicted to the left below. Note that
the transition system is essentially
the result of merging all deadlocking and `determined' states  of the two transition
systems depicted in Figure~\ref{fig:example_weak_vs_branching}. That is, 
state A represents state 0 of Figure~\ref{fig:example_weak_vs_branching},
state B represents
the deadlock states 1, 2, 4, 7 and 8 of Figure~\ref{fig:example_weak_vs_branching},
state C represents state 5 of Figure~\ref{fig:example_weak_vs_branching} and 
state D represents states 3 and 6 of Figure~\ref{fig:example_weak_vs_branching}.

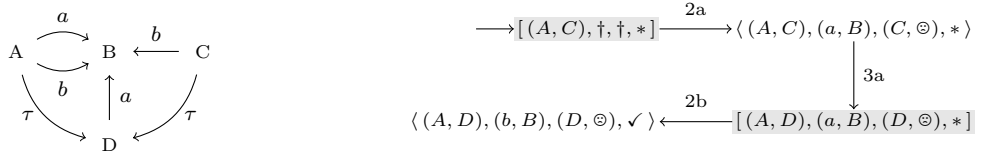
\begin{center}
\begin{minipage}{.35\textwidth}
  \begin{tikzpicture}
  \centering
    [node distance=55pt,inner sep = 1pt,minimum size=10pt]

     \tikzstyle{state}=[circle, draw=none,node distance=35pt]
     \tikzstyle{transition}=[->,>=stealth']

      \node[state] (A) {\tiny A};
      \node[state] [right of=A] (B) {\tiny B};
      \node[state] [right of=B] (C) {\tiny C};
      \node[state] [below of=B] (D) {\tiny D};
      \draw [->] (A) edge[bend left] node[above] {\scriptsize $a$} (B);
      \draw [->] (A) edge[bend right] node[below] {\scriptsize $b$} (B);
      \draw [->] (A) edge[bend right] node[left] {\scriptsize $\tau$} (D);
      \draw [->] (C) edge node[above] {\scriptsize $b$} (B);
      \draw [->] (C) edge[bend left] node[right] {\scriptsize $\tau$} (D);
      \draw [->] (D) edge node[right] {\scriptsize $a$} (B);
  \end{tikzpicture}
\end{minipage}
\begin{minipage}{.35\textwidth}
\centering
  \begin{tikzpicture}
    [node distance=25,inner sep = 1pt,minimum size=5pt,initial text={}]

     \tikzstyle{state}=[rectangle, draw=none]
     \tikzstyle{transition}=[->,>=stealth']

      \node[initial, state, fill=gray!20] (sACdds) {\tiny $\sconf{(A,C),\dagger,\dagger,*}$};
      \node[right of=sACdds, state,xshift=75pt] (ACaBCfs) {\tiny $\dconf{(A,C),(a,B),(C,\frownie),*}$};
      \node[below of=ACaBCfs, state,yshift=-10pt,fill=gray!20] (sADaBDfs) {\tiny $\sconf{(A,D),(a,B),(D,\frownie),*}$};
      \node[left of=sADaBDfs, state,xshift=-95pt] (ADbBDfc) {\tiny $\dconf{(A,D),(b,B),(D,\frownie),\checkmark}$};

      \draw [->] (sACdds) edge node[shape=circle,above,inner sep=2pt] {\tiny\ref{gbbg_spoiler_fresh_challenge}} (ACaBCfs);
      \draw [->] (ACaBCfs) edge node[shape=circle,right,inner sep=2pt] {\tiny\ref{gbbg_push_pebble_and_move}} (sADaBDfs);
      \draw [->] (sADaBDfs) edge node[shape=circle,above,inner sep=2pt] {\tiny\ref{gbbg_spoiler_renew_challenge}} (ADbBDfc);

  \end{tikzpicture}
\vspace{5pt}
\end{minipage}
\end{center}

The states $A$ and $C$ are delay bisimilar and weak bisimilar, but
they are not branching bisimilar or $\eta$-bisimilar as we explained
earlier.  \spoiler's strategy to win the branching bisimulation
game, which we explained on page~\pageref{weak_vs_branching_example},
can be replayed in the \gb-game, as illustrated
by the (solitaire) game graph depicted next to the transition system:
for \spoiler-owned configurations (marked grey), \spoiler uses her
rule~\eqref{gbbg_spoiler_fresh_challenge}, followed by
rule~\eqref{gbbg_spoiler_renew_challenge}, and \duplicator can only
use rule~\eqref{gbbg_push_pebble_and_move}. Notice that \spoiler's strategy
is winning in both the $\gbbg_{E(b,b)}$~game
and the $\gbbg_{E(b,o)}$~game; in both cases, the same (solitaire)
game graph is obtained.

\end{exa}

We elaborate on \duplicator's winning strategy in the $\gbbg_{E(o,o)}$ game
in Example~\ref{ex:game_weak_bisimilar} on page~\pageref{ex:game_weak_bisimilar}.

\subsection{Soundness and Completeness}
\label{sec:generic_game_sound_complete}

\newcommand{\good}{\emph{good}\xspace}
For our completeness result, we essentially need to show that any
pair of states $s,t$ related through a generic bisimulation, yields
a configuration $\sconf{(s,t),\dagger,\dagger,*}$ in our generic
bisimulation game that is won by \duplicator.  A play generally
passes through configurations that have a challenge and partial
match different from $\dagger$, so in our proof of completeness,
we must deal with such configurations, too.  More specifically, our
completeness proof deals with configurations of a particular shape;
we call such configurations \good.

\begin{defi}
Let $\conf{(s,t), c, m, r}$ be an arbitrary configuration in a
\gb-game. We say this configuration is $(x,y)$-\good
with respect to a given relation $R \subseteq S \times S$ iff $s
\rel{R} t$, $m = \dagger$ implies $c = \dagger$, and for $c = (a,s')$
and $m = (t_1,f)$, either $a = \tau$ and $s' \rel{R} t_1$, or there
is some $t'$ such that $s' \rel{R} t'$ and:

\begin{itemize}

\item if $f = \smiley$ then $t_1 \sgentaustar{y,R,s'} t'$;

\item if $f = \frownie$ then
there are $t_2,t_3$ so that $t_1 \sgentaustar{x,R,s} t_2 \pijl{a} t_3 \sgentaustar{y,R,s'} t'$.

\end{itemize}
\end{defi}
When the relation $R$ and both $x$ and $y$ are clear from the context,
we simply write that a configuration is \good rather than
$(x,y)$-\good for $R$.

\newcommand{\strategy}{\varrho\xspace}  
\newcommand{\nxt}[1]{\ensuremath{\textsf{next}(#1)}}

\begin{lem}[Completeness]
\label{lem:gbbbisim_is_gbbg}
For all $x,y \in \{o,b\}$, whenever we have $s \gbbbisim t$, we also have $s \gbbg_{E(x,y)} t$.

\end{lem}
\begin{proof}
We design a partial strategy for \duplicator and show that this
strategy is winning for configurations that are \good with
respect to $\gbbbisim$. Since $\sconf{(s,t),\dagger,\dagger,*}$ is
\good follows from $s \gbbbisim t$, it follows that \duplicator wins 
the $\gbbg_{E(x,y)}$ game and we thus have $s \gbbg_{E(x,y)} t$.

Fix $x,y \in \{o,b\}$. For convenience, we write $R$ for the relation
$\gbbbisim$. We first show that from every configuration
that is \good with respect to $R$,
\spoiler can only move to a configuration that is again \good.
Suppose $\sconf{(s,t),c,m,r}$ is such a configuration.
Then \spoiler may play such that the play will continue in either:
\begin{itemize}
\item $\dconf{(s,t),c,m,*}$ if $c \neq \dagger$ and \spoiler passes
on the configuration unchanged using rule~\eqref{gbbg_move_unchanged};

\item $\dconf{(s,t),(a,s'),(t,\frownie),*}$ for some $s \pijl{a}
s'$ if $c = \dagger$ when \spoiler uses rule~\eqref{gbbg_spoiler_fresh_challenge};

\item $\dconf{(s,t),(a,s'),(t,\frownie),\checkmark}$ for some $s
\pijl{a} s'$ for which $c \neq (a,s')$ when \spoiler uses rule~\eqref{gbbg_spoiler_renew_challenge};

\item $\dconf{(t,s),(a,t'),(s,\frownie),\checkmark}$ for some $t
\pijl{a} t'$ if \spoiler uses rule~\eqref{gbbg_spoiler_switches}.

\end{itemize}
In the first case, the fact that $\dconf{(s,t),c,m,r}$ is \good
follows from the fact that $\sconf{(s,t),c,m,r}$ is
\good.  The configurations of the remaining three
cases can be seen to be \good by an
immediate application of the conditions stating when two states are related by
an $(x,y)$-generic bisimulation, as a consequence of the fact that $s \rel{R}
t$ is imposed by definition of \good configurations. Recall that $R$ here 
is (the largest) $(x,y)$-generic bisimulation.
Moreover, note that every \duplicator-owned configuration reached
by \spoiler---when starting from a configuration that is 
\good---has a non-$\dagger$ challenge; consequently, such configurations---being 
\good---also carry a non-$\dagger$ partial match. 
\medskip

We next focus on configurations of the form
$\dconf{(s,t),(a,s'),(\bar{t},f),r}$ that are \good and we
argue that \duplicator can always move to a \spoiler-owned configuration
that is again \good. Let $\dconf{(s,t),(a,s'),(\bar{t},f),r}$
be such a configuration.  We distinguish three (mutually exclusive)
cases.

\begin{enumerate}
\item Suppose $a = \tau$ and $s' \rel{R} \bar{t}$.
Then, using rule~\eqref{gbbg_terminate},
\duplicator plays to configuration
$\sconf{(s',\bar{t}),\dagger,\dagger,\checkmark}$. Clearly this configuration is \good.

\item Case $f = \frownie$, and $a \neq \tau$ or not $s' \rel{R}
\bar{t}$.  Since $\dconf{(s,t),(a,s'),(\bar{t},f),r}$ is \good and
$f = \frownie$, there are $t',t_2,t_3$ so that $\bar{t} \sgentaustar{x,R,s}
t_2 \pijl{a} t_3 \sgentaustar{y,R,s'} t'$ and $s' \rel{R} t'$. Fix
states $t', t_2$ and $t_3$ on a shortest path witnessing these
properties.  We distinguish three further cases:

\begin{itemize}
\item Case $t_2 = \bar{t}$ but $t_3 \neq t'$.  If $\smiley \notin E(x,y)$ then
\duplicator plays to configuration
$\sconf{(s',t_3),(a,s'),(t_3,\smiley),*}$ using
rule~\eqref{gbbg_match_and_move}; otherwise she moves to configuration
$\sconf{(s,t),(a,s'),(t_3,\smiley),*}$ using
rule~\eqref{gbbg_match_and_postpone}. Both configurations are \good.

\item Case $t_2 = \bar{t}$ and $t_3 = t'$. Then \duplicator plays
to configuration $\sconf{(s',t_3),\dagger,\dagger,\checkmark}$ using
rule~\eqref{gbbg_match_move_and_terminate}, which is again \good.

\item Case $t_2 \neq \bar{t}$. 
We consider the set of paths underlying $\bar{t} \sgentaustar{x,R,s} t_2$,
and in particular the \emph{shortest} paths in this set.
Now we take as  $t^*$ the $\tau$-successor of  $\bar{t}$
in any of these shortest paths. \duplicator
then can play to configuration $\sconf{(s,t^*),(a,s'),(t^*,\frownie),*}$
if $\frownie \notin E(x,y)$, using rule~\eqref{gbbg_push_pebble_and_move}, and
to configuration $\sconf{(s,t),(a,s'),(t^*,\frownie),*}$, otherwise
using rule~\eqref{gbbg_push_pebble_and_wait}. Again, both configurations
are \good.

\end{itemize}

\item Case $f = \smiley$ and not $s' \rel{R} \bar{t}$. Since
$\dconf{(s,t),(a,s'),(\bar{t},f),r}$ is \good and $f = \smiley$,
there is some $t'$ such that $\bar{t} \sgentaustar{y,R,s'} t'$ and
$s' \rel{R} t'$. Consider a $t'$ closest to $\bar{t}$ (with respect to
the lengths of the paths underlying $\sgentaustar{y,R,s'}$)
with the property $s' \rel{R} t'$.  At this point, we can draw two conclusions:
\begin{itemize}
\item $y = o$ since $y = b$ and $\bar{t} \sgentaustar{y,R,s'} t'$ 
contradicts not $s' \rel{R} \bar{t}$;
\item $\bar{t} \neq t'$ since $s' \rel{R} t'$ but not $s' \rel{R} \bar{t}$.

\end{itemize}
From the first, it follows that $\smiley \in E(x,y)$ so \duplicator
may play either rule~\eqref{gbbg_push_pebble_and_move} or rule~\eqref{gbbg_push_pebble_and_terminate}. 
From the second, it follows
that $\bar{t} \sgentaustar{y,R,s'} t'$ reduces to $\bar{t} \taustar
t'$ and since $\bar{t} \neq t'$ there must be some immediate
$\tau$-successor $t^*$ of $\bar{t}$ on the shortest $\tau$-path from $\bar{t}$
to $t'$. Then \duplicator plays to configuration
$\sconf{(s,t),(a,s'),(t^*,\smiley),*}$ if this $t^* \neq t'$, and
configuration $\sconf{(s',t'),\dagger,\dagger,\checkmark}$ otherwise. We observe that both
configurations are again \good.

\end{enumerate}
Since in all cases, \duplicator can move to a configuration that is
again \good, it follows that if \spoiler starts in a configuration
that is \good, no play can pass along configurations that are
not \good. Moreover, \duplicator never gets stuck playing her
strategy.

Finally, we argue that \duplicator wins all plays starting in
configurations that are \good. Let $\Pi$ be the set of all plays
consistent with \duplicator's strategy
that start in some \spoiler-owned configuration that is \good. Since
on plays in $\Pi$, \duplicator never gets stuck, all finite plays
in $\Pi$ are won by \duplicator. It thus suffices to prove that
\duplicator wins all infinite plays in $\Pi$.

Towards a contradiction, assume that $\pi \in \Pi$ is a play won
by \spoiler. That means that there are only finitely many configurations
on $\pi$ in which \duplicator earns a $\checkmark$. Consequently,
\duplicator plays rules~\eqref{gbbg_terminate},~\eqref{gbbg_match_move_and_terminate}
or~\eqref{gbbg_push_pebble_and_terminate} only finitely often.
Furthermore, there are only finitely many configurations on $\pi$
with a non-$\dagger$ challenge. This can be seen as follows: only
\spoiler-owned configurations can have $\dagger$ challenges and
\duplicator only produces a $\dagger$ challenge by applying
one of the rules~\eqref{gbbg_terminate},~\eqref{gbbg_match_move_and_terminate}
or~\eqref{gbbg_push_pebble_and_terminate}, which she does only finitely often.  Thus,
$\pi$ must have an infinite suffix without $\checkmark$ rewards and
only non-$\dagger$ challenges. On this suffix, \spoiler's only moves
consist of passing on the configuration unchanged to \duplicator.

Consider the $i$-th \duplicator-owned configuration
$\dconf{(s_i,t_i),(a,s'),(\bar{t}_i,f_i),*}$ on this suffix.  Since
this configuration is \good, there must be some $t'$ such
that $s' \rel{R} t'$ and there must be some finite path from
$\bar{t}_i$ to $t'$. Now \duplicator
can simply follow this path, as it was done in cases~2
and~3 in the analysis we conducted to prove that goodness
can be preserved by \duplicator. In particular,
\duplicator will apply either
rule~\eqref{gbbg_match_move_and_terminate} or~\eqref{gbbg_push_pebble_and_terminate}
when $t'$ is finally reached, but in such a case she earns a
$\checkmark$, which means a contradiction.  So $\pi$ is won by \duplicator.
\end{proof}
We illustrate \duplicator's winning strategy as constructed in the proof above
using a small example.
\begin{exa}
\label{ex:game_weak_bisimilar}
Reconsider the states $A$ and $C$
of Example~\ref{ex:game_illustrated}, depicted again in Figure~\ref{fig:winning_duplicator}
for convenience. 
As we noted before,
these states are not branching bisimilar, but they are weak bisimilar, 
and, hence, they are $(o,o)$-bisimilar.
Following the strategy explained in the proof of
Lemma~\ref{lem:gbbbisim_is_gbbg}, this is illustrated by the 
$E(o,o)$-bisimulation subgame starting in $\sconf{(A,C),\dagger,\dagger,*}$, 
depicted in Figure~\ref{fig:winning_duplicator}
when \duplicator plays according to the aforementioned strategy.
Note also that \spoiler can only employ
move~\ref{gbbg_spoiler_fresh_challenge} in those configurations in
which there is no pending challenge. She can, however, at all times
change her challenge by playing either
move~\ref{gbbg_spoiler_renew_challenge} or move~\ref{gbbg_spoiler_switches}.

\end{exa}

\begin{figure}[hbpt]
\begin{center}
\begin{minipage}[b]{.10\textwidth}
  \begin{tikzpicture}
  \centering
    [node distance=55pt,inner sep = 1pt,minimum size=10pt]

     \tikzstyle{state}=[circle, draw=none,node distance=35pt]
     \tikzstyle{transition}=[->,>=stealth']

      \node[state] (A) {\tiny A};
      \node[state] [right of=A] (B) {\tiny B};
      \node[state] [right of=B] (C) {\tiny C};
      \node[state] [below of=B] (D) {\tiny D};
      \draw [->] (A) edge[bend left] node[above] {\scriptsize $a$} (B);
      \draw [->] (A) edge[bend right] node[below] {\scriptsize $b$} (B);
      \draw [->] (A) edge[bend right] node[left] {\scriptsize $\tau$} (D);
      \draw [->] (C) edge node[above] {\scriptsize $b$} (B);
      \draw [->] (C) edge[bend left] node[right] {\scriptsize $\tau$} (D);
      \draw [->] (D) edge node[right] {\scriptsize $a$} (B);
      \draw ($ (A.north west) + (0pt,15pt) $) rectangle ($ (C.north east) + (0pt,-57pt) $);
  \end{tikzpicture}
\end{minipage}
\begin{minipage}[b]{.85\textwidth}
  \begin{tikzpicture}
    [node distance=25,inner sep = 1pt,minimum size=5pt,initial text={},every initial by arrow/.style={*->}]

     \tikzstyle{state}=[rectangle, draw=none]
     \tikzstyle{transition}=[->,>=stealth']

      \node[initial, state,fill=gray!20] (ACdds) {\tiny $\sconf{(A,C),\dagger,\dagger,*}$};
      \node[state, below of=ACdds,yshift=-20pt] (ACbBCfs) {\tiny $\dconf{(A,C),(b,B),(C,\frownie),*}$};
      \node[state, right of=ACdds,xshift=75pt] (ACaBCfs) {\tiny $\dconf{(A,C),(a,B),(C,\frownie),*}$};
      \node[state, below of=ACaBCfs,yshift=-20pt] (CAbBAfc) {\tiny $\dconf{(C,A),(b,B),(A,\frownie),\checkmark}$};
      \node[state, right of=ACaBCfs,xshift=95pt,fill=gray!20] (sACaBDfs) {\tiny $\sconf{(A,C),(a,B),(D,\frownie),*}$};
      \node[state, below of=sACaBDfs,yshift=-20pt] (ACaBDfs) {\tiny $\dconf{(A,C),(a,B),(D,\frownie),*}$};
      \node[state, below of=CAbBAfc,yshift=-20pt,fill=gray!20] (BBddc) {\tiny $\sconf{(B,B),\dagger,\dagger,\checkmark}$};
      \node[state, above of=ACdds,yshift=20pt,xshift=-50pt] (ACtDCfs) {\tiny $\dconf{(A,C),(\tau,D),(C,\frownie),*}$};
      \node[state, above of=ACdds,yshift=20pt,xshift=50pt] (CAtDAfc) {\tiny $\dconf{(C,A),(\tau,D),(A,\frownie),\checkmark}$};
      \node[state, right of=CAtDAfc,xshift=95pt] (ACtDCfc) {\tiny $\dconf{(A,C),(\tau,D),(C,\frownie),\checkmark}$};
      \node[state, above of=sACaBDfs,yshift=95pt] (ACbBCfc) {\tiny $\dconf{(A,C),(b,B),(C,\frownie),\checkmark}$};
      \node[state, above of=CAtDAfc,yshift=20pt,xshift=-50pt,fill=gray!20] (sDDddc) {\tiny $\sconf{(D,D),\dagger,\dagger,\checkmark}$};
      \node[state, above of=sDDddc,yshift=20pt,xshift=-50pt] (DDaBDfs) {\tiny $\dconf{(D,D),(a,B),(D,\frownie),*}$};
      \node[state, above of=sDDddc,yshift=20pt,xshift=50pt] (DDaBDfc) {\tiny $\dconf{(D,D),(a,B),(D,\frownie),\checkmark}$};
      \node[state, above of=DDaBDfc,yshift=20pt,xshift=-50pt,fill=gray!20] (sBBddc2) {\tiny $\sconf{(B,B),\dagger,\dagger,\checkmark}$};

      \draw [->] (ACdds) edge node[shape=circle,above,inner sep=2pt] {\tiny\ref{gbbg_spoiler_switches}} (CAbBAfc);
      \draw [->] (ACdds) edge node[shape=circle,right] {\tiny\ref{gbbg_spoiler_fresh_challenge}} (ACbBCfs);
      \draw [->] (ACdds) edge node[shape=circle,above] {\tiny\ref{gbbg_spoiler_fresh_challenge}} (ACaBCfs);
      \draw [->] (ACaBCfs) edge node[shape=circle,above] {\tiny\ref{gbbg_push_pebble_and_wait}} (sACaBDfs);
      \draw [->] (sACaBDfs) edge node[shape=circle,right,inner sep=1pt] {\tiny\ref{gbbg_move_unchanged}} (ACaBDfs);
      \draw [->] (sACaBDfs) edge node[shape=circle,above,inner sep=2pt] {\tiny\ref{gbbg_spoiler_switches}} (CAbBAfc);
      \draw [->] (sACaBDfs) edge node[shape=circle,above,inner sep=2pt] {\tiny\ref{gbbg_spoiler_switches}} (CAtDAfc);
      \draw [->] (sACaBDfs) edge node[shape=circle,right,inner sep=3pt] {\tiny\ref{gbbg_spoiler_renew_challenge}} (ACtDCfc);
      \draw [->] (sACaBDfs) edge node[shape=circle,right,inner sep=2pt] {\tiny\ref{gbbg_spoiler_renew_challenge}} (ACbBCfc);
      \draw [->] (ACaBDfs) edge node[shape=circle,below] {\tiny\ref{gbbg_match_move_and_terminate}} (BBddc);
      \draw [->] (CAbBAfc) edge node[shape=circle,right] {\tiny\ref{gbbg_match_move_and_terminate}} (BBddc);
      \draw [->] (ACbBCfs) edge node[shape=circle,below] {\tiny\ref{gbbg_match_move_and_terminate}} (BBddc);
      \draw [->] (ACdds) edge node[shape=circle,right,inner sep=2pt] {\tiny\ref{gbbg_spoiler_renew_challenge}} (ACtDCfs);
      \draw [->] (ACdds) edge node[shape=circle,left,inner sep=3pt] {\tiny\ref{gbbg_spoiler_switches}} (CAtDAfc);
      \draw [->] (CAtDAfc) edge node[shape=circle,left,inner sep=2pt] {\tiny\ref{gbbg_match_move_and_terminate}} (sDDddc);
      \draw [->] (ACtDCfs) edge node[shape=circle,right,inner sep=2pt] {\tiny\ref{gbbg_match_move_and_terminate}} (sDDddc);
      \draw [->] (ACtDCfc) edge node[shape=circle,above] {\tiny\ref{gbbg_match_move_and_terminate}} (sDDddc);
      \draw [->] (sDDddc) edge node[shape=circle,left,inner sep=3pt] {\tiny\ref{gbbg_spoiler_switches}} (DDaBDfc);
      \draw [->] (sDDddc) edge node[shape=circle,right,inner sep=2pt] {\tiny\ref{gbbg_spoiler_fresh_challenge}} (DDaBDfs);
      \draw [->] (DDaBDfs) edge node[shape=circle,right,inner sep=2pt] {\tiny\ref{gbbg_match_move_and_terminate}} (sBBddc2);
      \draw [->] (DDaBDfc) edge node[shape=circle,left,inner sep=2pt] {\tiny\ref{gbbg_match_move_and_terminate}} (sBBddc2);
      \draw [->] (ACbBCfc) edge node[shape=circle,above] {\tiny\ref{gbbg_match_move_and_terminate}} (sBBddc2);

  \end{tikzpicture}
\end{minipage}
\end{center}
\caption{The (solitaire) game underlying $A \gbbg_{E(o,o)} C$ in
which \duplicator plays according to the (memoryless) winning
strategy of Lemma~\ref{lem:gbbbisim_is_gbbg}.  \spoiler-owned
vertices are marked grey, \duplicator-owned vertices are not marked,
and edges are annotated with the rule \spoiler or \duplicator used
to move to the next configuration. The `starting configuration' of the
game is indicated by the $\drawinitarrow$ symbol. For quick reference, the 
Labelled Transition System of Figure~\ref{fig:example_weak_vs_branching}
serving as the basis for this game is depicted in the below left
corner.}

\label{fig:winning_duplicator}
\end{figure}

We next focus on the soundness of our game; that is, we will show
that \duplicator cannot win configurations
$\sconf{(s,t),\dagger,\dagger,*}$ for states $s$ and $t$ that are
not $(x,y)$-generic bisimilar. Before we give a formal proof of soundness,
we first state two useful observations concerning the
\gb-games.

\begin{prop}\label{prop:invariant_r}
Configurations $\sconf{(s,t),c,m,*}$ and $\sconf{(s,t),c,m,\checkmark}$
are both won by the same player. So are  $\dconf{(s,t),c,m,*}$ and
$\dconf{(s,t),c,m,\checkmark}$.

\end{prop}
\begin{proof}
This follows immediately from the B\"uchi winning condition: any
player that wins some suffix of an infinite play also wins the
infinite play itself. Furthermore, note that neither \spoiler nor
\duplicator can get stuck playing a game by changing a reward from
$*$ to $\checkmark$ or \emph{vice versa}.
\end{proof}

\begin{prop}\label{prop:eqconfigurations}
If \duplicator wins a configuration $\sconf{(s,t),c,m,r}$, then she
also wins the configuration $\sconf{(s,t),\dagger,\dagger,\checkmark}$.

\end{prop}

\begin{proof} 
Let  $\sconf{(s,t),c,m,r}$ be a \spoiler-owned consistent configuration
that is won by \duplicator. We simply observe that any configuration
reached by a move by \spoiler from the configuration
$\sconf{(s,t),\dagger,\dagger,\checkmark}$, can also be reached
from the given $\sconf{(s,t),c,m,r}$ up to possibly the value of
$r$. Next we simply apply Proposition~\ref{prop:invariant_r} to
conclude the result.
\end{proof}

\begin{rem}
The converse of the result above does not hold in general. To
illustrate this, consider a transition system with states $s,s'$
and one transition $s \pijl{a} s'$, where $a \neq \tau$.  Clearly,
\duplicator wins $\sconf{(s,s),\dagger,\dagger,\checkmark}$ but she
does not, for instance, win $\sconf{(s,s),(a,s'),(s',\frownie),*}$.

\end{rem}
Note that as a consequence of Proposition~\ref{prop:eqconfigurations},
\duplicator's rule~\eqref{gbbg_match_and_move} is redundant.  That is,
a game that lacks rule~\eqref{gbbg_match_and_move} is
equivalent (from the point of view of $\gbbg_E$) to the game we
define in Definition~\ref{def:generalised game}; we briefly return to this
matter in Section~\ref{sec:generic_specific}. Although our game contains
some redundancy, we believe that our presentation of the game better
reflects the diagram and intuition of Figure~\ref{fig:generic_picture}.

\begin{lem}[Soundness]\label{lem:game_is_generic_branching_bisim}
For all $x,y \in \{o,b\}$, the relation $\gxybbg$ is an $(x,y)$-generic bisimulation.
\end{lem}
\begin{proof}
We proceed as follows: we construct a relation $R \subseteq S \times
S$ such that all pairs of states related by $R$ represent some
\spoiler-owned configuration that is won by \duplicator. Next, we
show that this relation is an $(x,y)$-generic bisimulation.

Fix some $x,y \in \{o,b\}$ and define $R \subseteq S \times S$ as
follows:
\[
R = \{ (s,t) ~|~ \text{\duplicator wins $\sconf{(s,t),\dagger,\dagger, *}$} \}
\]
We show that $R$ meets the conditions of
Definition~\ref{def:generic_bisimulation}.  First observe that symmetry of $R$
follows from Proposition~\ref{prop:invariant_r} and the fact that starting from
a configuration $\sconf{(s,t),\dagger,\dagger,*}$, \spoiler has the
same options as when starting from configuration
$\sconf{(t,s),\dagger,\dagger,*}$.

Suppose $s \rel{R} t$, and assume that $s \pijl{a} s'$ holds. It suffices
to prove that $R$ meets the transfer condition, \ie either:
\begin{enumerate}
\item[\mylabel{T1}{T1}] $a = \tau$ and $s' \rel{R} t$, or
\item[\mylabel{T2}{T2}] $t \gentaustar{x,R,s} t_1 \pijl{a} t_2 \gentaustar{y,R,s'} t'$ for some
$t_1,t_2,t'$ such that $s' \rel{R} t'$.
\end{enumerate}
Since $s \rel{R} t$, we
know that \duplicator wins $\sconf{(s,t),\dagger,\dagger,*}$. Fix
a strategy $\strategy$ for \duplicator that is winning from this
configuration. Now, consider a play $\pi$ that emerges by
\duplicator playing according to $\strategy$ and \spoiler using her
rule~\eqref{gbbg_spoiler_fresh_challenge}
to move from $\sconf{(s,t),\dagger,\dagger,*}$ to
$\dconf{(s,t),(a,s'),(t,\frownie),*}$ and her rule~\eqref{gbbg_move_unchanged} on all subsequent
non-$\dagger$ configurations. Note that, by assumption,
$s \pijl{a} s'$ holds, so \spoiler's rule~\eqref{gbbg_spoiler_fresh_challenge} is applicable, and a play $\pi$ of the required form therefore is guaranteed to
exist.

We will show that $\pi$ contains all information necessary for concluding that either~\eqref{T1} or~\eqref{T2} holds.  Notice that since
\duplicator wins $\pi$, the play $\pi$ must have a prefix $\pi'$ of the
following form (for some $n \ge 0$ and appropriately chosen $s_i,t_i,\bar{t}_i,f_i$ and
$t'$):
\[
\begin{array}{ll}
\pi' = & \sconf{(s,t),\dagger,\dagger,*}\\
& \dconf{(s_0,t_0),(a,s'),(\bar{t}_0,f_0),*}\\
& \sconf{(s_1,t_1),(a,s'),(\bar{t}_1,f_1),*}\\
& \cdots\\
& \dconf{(s_n,t_n),(a,s'),(\bar{t}_n,f_n),*}\\
& \sconf{(s',t'),\dagger,\dagger,\checkmark}\\
\end{array}
\]
where  either $t' = \bar{t}_n$, or $t'$ is an immediate $\tau$-successor or $a$-successor of
$\bar{t}_n$. Moreover, for all odd $i$, we have $\bar{t}_i = \bar{t}_{i+1}$ and, for
even $i < n-1$, either
$\bar{t}_i \pijl{\tau} \bar{t}_{i+2}$ or $\bar{t}_i \pijl{a} \bar{t}_{i+2}$.
Fix this prefix $\pi'$.  Note that
$\sconf{(s',t'),\dagger,\dagger,\checkmark}$ is again won
by \duplicator, and by Proposition~\ref{prop:invariant_r} and
by definition of $R$ we find $s' \rel{R} t'$. We next distinguish
cases, essentially based on the length of $\pi'$.

In case $n = 0$, \duplicator must have either used rule~\eqref{gbbg_terminate}
or rule~\eqref{gbbg_match_move_and_terminate}.  If \duplicator used
rule~\eqref{gbbg_terminate}, then $a = \tau$ and $t' = \bar{t}_0 =  t$,
and thus $R$ meets transfer condition~\eqref{T1}.  When \duplicator used
rule~\eqref{gbbg_match_move_and_terminate}, then we must have had
$t = \bar{t}_0 \pijl{a} t'$, in which case $R$ meets transfer
condition~\eqref{T2}.


In case $n > 0$, we find $\bar{t}_0 \taustar \bar{t}_{k-1} \pijl{a}
\bar{t}_k \taustar t'$ for some $k$ so it suffices to prove that also $\bar{t}_0
\gentaustar{x,R,s} \bar{t}_{k-1} \pijl{a} \bar{t}_k \gentaustar{y,R,s'}
t'$.  We distinguish two cases:

\begin{enumerate}
\item For all $i \le n$, we have $f_i = \frownie$ in $\pi'$. In
that case, we find that \duplicator only played according to
rules~\eqref{gbbg_push_pebble_and_move} or~\eqref{gbbg_push_pebble_and_wait}
for all $i$-th \duplicator owned configurations ($i < n$) and 
rule~\eqref{gbbg_match_move_and_terminate} 
for its $n$-th configuration. Consider the  sequence
of transitions $t = \bar{t}_0 \pijl{\tau} \bar{t}_2 \pijl{\tau} \cdots
\pijl{\tau} \bar{t}_n \pijl{a} t'$. 
\begin{itemize}

\item Case $x = o$. Then we can immediately conclude
$t \gentaustar{x,R,s} \bar{t}_{n} \pijl{a} t' \gentaustar{y,R,s'} t'$
and since $s' \rel{R} t'$,  $R$ meets transfer condition~\eqref{T2}.

\item Case $x = b$. Then also $\frownie \notin
E(x,y)$ and therefore \duplicator can only have applied
rule~\eqref{gbbg_push_pebble_and_move}. In that case we have $s_i = s$
and $t_i = \bar{t}_i$ for all $i \le n$.  Furthermore,  since \duplicator
wins all configurations on $\pi'$, she also wins all configurations
$\sconf{(s,t_i), (a,s'), (\bar{t}_i,\frownie),*}$. By
Proposition~\ref{prop:eqconfigurations}, she also wins
$\sconf{(s,t_i),\dagger,\dagger,*}$ so we find $s \rel{R} t_i$ (and
therefore $s \rel{R} \bar{t}_i$) for all $i \le n$. As a result, we
find that $t \sgentaustar{x,R,s} \bar{t}_n$
holds and therefore $t \gentaustar{x,R,s} \bar{t}_{n} \pijl{a}
t' \gentaustar{y,R,s'} t'$, which, together with $s' \rel{R} t'$,
is all we needed for
concluding that $R$ meets transfer condition~\eqref{T2}.

\end{itemize}

\item There is some $k$ ($1 \le k \le n$) such that $f_k = \smiley$ in $\pi'$;
fix the smallest $k$ with this property. We note that $k$ must be odd since
the first configuration with a $\smiley$ is owned by \spoiler.
Observe that this means that $\bar{t}_{k-1} \pijl{a} \bar{t}_{k+1}$
must hold since \duplicator used either
rule~\eqref{gbbg_match_and_move} or rule~\eqref{gbbg_match_and_postpone} at
configuration $\dconf{(s_{k-1},t_{k-1}),(a,s'),
(\bar{t}_{k-1},f_{k-1}),*}$. At all other configurations (except
for the last), \duplicator can only have played
rule~\eqref{gbbg_push_pebble_and_move} or~\eqref{gbbg_push_pebble_and_wait},
meaning we must have had $\bar{t}_{i} \pijl{\tau} \bar{t}_{i+2}$ for
all $i \neq k-1$. For the last configuration, \duplicator could only
have used rule~\eqref{gbbg_push_pebble_and_terminate}.
First, consider the sequence of transitions
$t = \bar{t}_0 \pijl{\tau} \bar{t}_2 \pijl{\tau} \cdots
\pijl{\tau} \bar{t}_{k-1}$. 
\begin{itemize}
\item Case $k = 1$. Then $t = \bar{t}_{k-1}$ and since $s \rel{R} t$ we can 
then immediately conclude that $t \gentaustar{x,R,s} \bar{t}_{k-1}$ holds.

\item Case $k > 1$. Following the argument we used in case~1 (reading $k-1$ 
for $n$), we conclude that $t \gentaustar{x,R,s} \bar{t}_{k-1}$.
 
\end{itemize}
Second, consider the sequence of transitions $\bar{t}_{k+1} \pijl{\tau} \bar{t}_{k+3}
\pijl{\tau} \cdots \pijl{\tau} \bar{t}_{n} \pijl{\tau} t'$.
%
Then we can use similar arguments as in case~1
(observing that $y = b$ implies that \duplicator could only have used 
rule~\eqref{gbbg_push_pebble_and_move}), to conclude that
$\bar{t}_{k+1} \gentaustar{y,R,s'} \bar{t}_{n}$. Moreover, from 
$\bar{t}_{k+1} \gentaustar{y,R,s'} \bar{t}_n$, 
$\bar{t}_n \pijl{\tau} t'$ and $s' \rel{R} t'$, we  also obtain
$\bar{t}_{k+1} \gentaustar{y,R,s'} t'$.


Concluding, we find that we have $t \gentaustar{x,R,s} \bar{t}_{k-1}
\pijl{a} \bar{t}_{k+1} \gentaustar{y,R,s'} t'$. Since we already had established
that we have $s' \rel{R} t'$, we conclude that $R$ meets transfer
condition~\eqref{T2}.

\end{enumerate}
In both cases $R$ meets the transfer condition. Thus, $R$
is an $(x,y)$-generic bisimulation relation, and therefore, since
for any $s,t \in S$ such that $s \gxybbg t$ we have $s \rel{R} t$, 
we also have $s \gbbbisim t$.  
\end{proof}
We are now in a position to prove our claimed correspondence.
\correspondencetheorem*
\begin{proof}
The implication from left to right follows from Lemma~\ref{lem:gbbbisim_is_gbbg}.
The implication from right to left follows from Lemma~\ref{lem:game_is_generic_branching_bisim}.
\end{proof}

\subsection{Relating Branching and Generic Bisimulation Games}
\label{sec:generic_specific}

As we have shown in the previous section, our generic bisimulation game 
exactly characterises $(x,y)$-generic bisimilarity. More specifically,
this implies that $\gbbg_{\emptyset}$ coincides with branching bisimilarity.
Since we already have a game that captures branching bisimilarity, one may
wonder what the relation between $\gbbg_{\emptyset}$ and that game is. 
In this section, we formally relate the game play in the generic
bisimulation game to the game play in the branching bisimulation
game.

First, observe that by definition, the game for $\gbbg_{\emptyset}$
consists of all rules by \spoiler, but only
rules~\eqref{gbbg_terminate},~\eqref{gbbg_match_and_move},~\eqref{gbbg_match_move_and_terminate},~\eqref{gbbg_push_pebble_and_move} and~\eqref{gbbg_push_pebble_and_terminate} of \duplicator. Intuitively, this
one-but-last move coincides with \duplicator's third rule in the branching
bisimulation game. Rule~\eqref{gbbg_match_move_and_terminate} can be seen
to match with \duplicator's second rule in the \bb-game. However, rule~\eqref{gbbg_match_and_move} and~\eqref{gbbg_push_pebble_and_terminate} have no counterpart in the
\bb-game. From these observations, we can expect
that every strategy in a \bb-game has a matching strategy
in our \gb-game, but not \emph{vice versa}. \medskip

We next formalise these arguments. Consider the abstraction function
$f$ which maps configurations in the \gb-game to
configurations in the \bb-game, where $f$ is
defined by $f(\conf{(s,t),c,m,r}) = \conf{(s,t),c,r}$.  The function
$f$ can be lifted from configurations to plays in the natural manner.
Using $f$, we can claim that the generic game simulates the branching
bisimulation game.  More formally, we establish that any strategy
of \duplicator in the \bb-game induces a `matching'
strategy in the $\emptyset$-\gb-game that is such that, after
abstraction using $f$, all plays resulting from this matching
strategy are plays in the \bb-game. This would
allow us to conclude that a winning strategy for \duplicator in the
\bb-game induces a winning strategy for her in
the generic game.

\begin{prop}\label{prop:branching_induced_strategy}
For every strategy $\strategy_{bb}$ of \duplicator in the \bb-game,
there is a strategy $\strategy_{gb}$ of \duplicator in the $\emptyset$-\gb-game
such that for every $\strategy_{gb}$-consistent play $\pi$, $f(\pi)$ is a
$\strategy_{bb}$-consistent play in the \bb-game.
\end{prop}
\begin{proof} Let $\strategy_{bb}$ be an arbitrary strategy of
\duplicator in the \bb-game. Consider the partial strategy
$\strategy_{gb}$, defined as follows:
\[
\begin{array}{ll}
& \strategy_{gb}(\dconf{(s,t),(a,s'),(t,\frownie),r})\\
 = & \\
& \begin{cases}
\sconf{(s,t'), (a,s'), (t',\frownie), r'} & \text{ iff
$\strategy_{bb}(\dconf{(s,t),(a,s'),r}) =
\sconf{(s,t'),(a,s'),r'}$} \\
\sconf{(s',t'), \dagger, \dagger, \checkmark} & \text{ iff
$\strategy_{bb}(\dconf{(s,t),(a,s'),r}) =
\sconf{(s',t'),\dagger,\checkmark}$}
\end{cases}
\end{array}
\]
One can check that $\strategy_{gb}$ is well-defined.
Now, consider an arbitrary play $\pi$ starting in some configuration
$\sconf{(s,t),\dagger,\dagger,*}$ that is consistent with
$\strategy_{gb}$. 
We next use induction on the length of the prefix of
$\pi$ to show that $f(\pi)$ is a play in the \bb-game that is consistent with $\strategy_{bb}$ and $\pi$ satisfies the following
invariant for all configurations $\conf{(s_0,t_0),c_0,m_0,r_0}$ on $\pi$:
\[
c_0 = m_0 = \dagger\text{, or } m_0 = (t_0,\frownie)\text{ and } c_0 = (a,s_0') 
\text{ for some $a,s_0'$ such that $s_0 \pijl{a} s_0'$}
\]

Clearly, the prefix
of length 1 of $\pi$, \viz\ $\sconf{(s,t),\dagger,\dagger,*}$,
induces a prefix of length 1, \viz $f(\sconf{(s,t),\dagger,\dagger,*})$
of some play in the \bb-game that is consistent with $\strategy_{bb}$.
Moreover, the single configuration on this prefix satisfies the invariant.

Next, assume that a given prefix $\pi'$ of length $n$ of $\pi$ induces
a prefix $f(\pi')$ of a play, consistent with $\strategy_{bb}$, 
in the \bb-game, and we assume that on this prefix $\pi'$,
the above invariant holds. Consider the prefix
$\pi'\ \conf{(s',t'),c',m',r'}$ of $\pi$. We show that all configurations on
this prefix meet the invariant and, second, we show that also
$f(\pi'\ \conf{(s',t'),c',m',r'})$ is a prefix of a play,
consistent with $\strategy_{bb}$, in the \bb-game.

Suppose
$\sconf{(s_0,t_0),c_0,m_0,r_0}$ is the last configuration on $\pi'$. Since
\duplicator and 
\spoiler  alternate their moves, we find that
$\conf{(s',t'),c',m',r'}$ is a configuration owned by \duplicator. We 
analyse the rule used by \spoiler and show that she could have 
played to a similar configuration in the \bb-game:
\begin{itemize}
\item  If \spoiler used her first rule, then $c_0 \neq \dagger$ and
by our invariant, $c_0 = (a,s_0')$ for some $s_0 \pijl{a} s_0'$, and
$m_0 = (t_0,\frownie)$. Thus
\spoiler played to configuration $\dconf{(s_0,t_0),(a,s_0'),(t_0,\frownie),*}$.
Note that this configuration again meets the invariant.
Moreover, observe that configuration
$f(\dconf{(s_0,t_0),(a,s_0'),(t_0,\frownie),*}) = \dconf{(s_0,t_0),(a,s_0'),*}$ is a
configuration \spoiler could have reached from
$\sconf{(s_0,t_0),(a,s_0'),r_0}$ in the \bb-game
using her first rule, since $s_0 \pijl{a} s_0'$ follows from our invariant.

\item If \spoiler used her second rule, then $c_0 = \dagger$ or $c_0 \neq c'$.
We first analyse the case when $c_0 = \dagger$.  Then for some $s_0'$ such
that $s_0 \pijl{a} s_0'$,
\spoiler played to 
$\dconf{(s_0,t_0),(a,s_0'),(t_0,\frownie),*}$. Note that the resulting
configuration again adheres to the invariant. Moreover, in the \bb-game, \spoiler could have played to
$\dconf{(s_0,t_0),(a,s_0'),*}$.

If $c_0 \neq c'$  then she played to
configuration $\dconf{(s_0,t_0),(a,s_0'),(t_0,\frownie),\checkmark}$ for
some $a,s_0'$ such that $s_0 \pijl{a} s_0'$ and $c_0 \neq (a,s_0')$. Again,
the resulting configuration satisfies our invariant.
It immediately follows that \spoiler
could have played to $\dconf{(s_0,t_0),(a,s_0),\checkmark}$ in the
\bb-game.

\item If \spoiler used her third rule, then she must have played to 
$\dconf{(t_0,s_0),(a,t_0'),(s_0,\frownie),\checkmark}$ for
some $a,t_0'$ such that $t_0 \pijl{a} t_0'$. This configuration satisfies
the invariant. Moreover, we find that \spoiler could have
played to $\dconf{(t_0,s_0),(a,t_0'),\checkmark}$ in the branching
bisimulation game.

\end{itemize}
Next, suppose that $\dconf{(s_0,t_0),c_0,m_0,r_0}$ is the last
configuration on $\pi'$. First note that $\strategy_{gb}$ is defined
for $(\dconf{(s_0,t_0),c_0,m_0,r_0})$ since, by our invariant,
we find that $m_0 = c_0 = \dagger$, or $m_0 = (t_0,\frownie)$ and
$c_0 = (a,s_0')$ for some $a,s_0'$ satisfying $s_0 \pijl{a} s_0'$. 
Since configuration
$f(\dconf{(s_0,t_0),c_0,m_0,r_0})$ is the last configuration on
$f(\pi')$ and \duplicator's strategy from $\dconf{(s_0,t_0),c_0,m_0,r_0}$
matches $f(\dconf{(s_0,t_0),c_0,m_0,r_0})$, the result follows
immediately.

Consequently, in both cases we find that
$f(\pi'\ \conf{(s',t'),c',m',r'})$ is a prefix of length $n+1$
of a play consistent with $\strategy_{bb}$. We thus find that
all prefixes of $\pi$ induce prefixes of $f(\pi)$ that are 
$\strategy_{bb}$-consistent and we thus conclude that $f(\pi)$ is a play 
consistent with $\strategy_{bb}$ in the \bb-game.
\end{proof}
From the above proposition, we find that whenever \duplicator wins the
\bb-game, she also wins the generic game. This follows
from the fact that any winning
strategy by \duplicator in the \bb-game induces a
winning strategy for \duplicator in the \gb-game.\medskip

For the reverse, we first observe that \duplicator can follow an
\emph{eager} strategy in the $\gbbg_{\emptyset}$ game: whenever she
uses rule~\eqref{gbbg_match_and_move} to
play to a configuration $\sconf{(s_0,t_0),(a,s_0'),(t_0,\smiley),*}$,
she could also have used rule~\eqref{gbbg_match_move_and_terminate} to play 
to configuration
$\sconf{(s_0,t_0),\dagger,\dagger,\checkmark}$ instead without
changing the outcome of the play. This follows immediately from
Proposition~\ref{prop:eqconfigurations}. We say that \duplicator
follows an \emph{eager} strategy if she never uses rule~\eqref{gbbg_match_and_move}.
The following result essentially follows immediately from the above.
\begin{lem}
If \duplicator wins a configuration $\sconf{(s,t),\dagger,\dagger,*}$,
she has an eager winning strategy to win this configuration. \qedhere
\end{lem}
Observe that \duplicator's rule~\eqref{gbbg_push_pebble_and_terminate} is
only `enabled' once she plays rule~\eqref{gbbg_match_and_move}. As a consequence,
this rule is effectively disabled on all plays in which \duplicator follows
an eager strategy. Thus, if \duplicator wins a configuration
$\sconf{(s,t),\dagger,\dagger,*}$ she can do so without using
rules~\eqref{gbbg_match_and_move} and~\eqref{gbbg_push_pebble_and_terminate}.

Next, consider an augmentation function $g$ which takes configurations from
the \bb-game and yields configurations in the \gb-game,
where $g(\conf{(s,t),c,r}) = \conf{(s,t),c,(t,\frownie),r}$
if $c \neq \dagger$ and $g(\conf{(s,t),\dagger,r}) = \conf{(s,t),\dagger,\dagger,r}$.
We again lift $g$ from configurations to plays in the natural manner.
\begin{prop}\label{prop:eager_induced_strategy}
For every eager strategy $\strategy_{gb}$ of \duplicator in the \gb-game,
there is a strategy $\strategy_{bb}$ of \duplicator in the \bb-game
such that for every $\strategy_{bb}$-consistent play $\pi$, $g(\pi)$ is a
$\strategy_{gb}$-consistent play in the \gb-game.
\end{prop}
\begin{proof} Analogous to the proof of Proposition~\ref{prop:branching_induced_strategy}.
\end{proof}
As a result, every eager strategy in the \gb-game induces a strategy
in the \bb-game. We then have the following result.
\begin{thm}\label{thm:branching_special_case}
We have $s \bbg t$ iff $s \gbbg_{\emptyset} t$.
\end{thm}
\begin{proof}
The implication from left to right follows from Proposition~\ref{prop:branching_induced_strategy}.
This can be seen as follows:
assuming that $s \bbg t$ holds, then \duplicator has a strategy to win
configuration $\sconf{(s,t),\dagger,*}$. Consequently, \duplicator has a
strategy in the \gb-game such that every play $\pi$ that
is consistent with this play induces a play $f(\pi)$ in the \bb-game that is consistent with \duplicator's winning strategy. Since $f(\pi)$ is
won by \duplicator in the \bb-game, $\pi$ is won by
her in the \gb-game. Therefore, \duplicator wins
$\sconf{(s,t),\dagger,\dagger,*}$ in the \gb-game and
thus $s \gbbg_{\emptyset} t$.

Using identical arguments, the implication from right to left can be shown
to follow from Proposition~\ref{prop:eager_induced_strategy}.
\end{proof}

\subsection{Variations on the Generic Bisimulation Games}
\label{sec:variations}

We finish this section with a brief discussion on two possible
alternative definitions for our \gb-game. The
definition we presented in Section~\ref{sec:generic_game_def} is,
as we illustrated in the previous subsection, closely related to
the \bb-game we presented in Section~\ref{sec:games}.
\medskip

A first alternative definition that is similar in spirit to the
current definition is obtained by uncoupling `termination' from
the walk along the edges in the transition system.  More specifically,
we can drop \duplicator's rules~\eqref{gbbg_match_move_and_terminate}
and~\eqref{gbbg_push_pebble_and_terminate} and rephrase rule~\eqref{gbbg_terminate}
as follows:
\begin{itemize}
\item move to $\sconf{(s_0',s),\dagger,\dagger,\checkmark}$ if
$a = \tau$ or $f = \smiley$.
\end{itemize}
The resulting game definition is more concise than the original
one. While it is not too hard to see that the resulting game is
still the same as the original game play, the additional round that
\duplicator needs to conclude matching \spoiler's challenge makes
that it is a bit more involved to show that it generalises the
\bb-game.\medskip

The second alternative we present is further removed from our current
definition. Rather than parameterising \duplicator's rules, one can
also parameterise the rules of \spoiler. This means that \duplicator's
role is reduced to updating the partial match, moving along the
transition system taking a number of $\tau$ actions and the action
occurring in \spoiler's challenge. Depending on the relation that
is being characterised, \spoiler can `prematurely' decide to pose a new
challenge to \duplicator and continue the game play
from a fresh position.  This way, \eg\ branching bisimilarity
is captured by allowing \spoiler to pose new challenges at any point
in the game play, whereas, \eg\ $\eta$-bisimilarity is captured by
allowing \spoiler to pose new challenges so long as the configurations
contain partial matches of the form $(s,\frownie)$. Formally, the
game is as follows.

\begin{defi}
\label{def:alternative_generic_bisimulation_game}
For each $E \subseteq \{\frownie,\smiley\}$, a 
\emph{dual generic bisimulation (dgb) game} on an LTS $L$
(or $E$-dgb-game) is played by
\spoiler and \duplicator on an arena  of \spoiler-owned
configurations $\sconf{(s,t),c,m,r}$ and \duplicator-owned configurations
$\dconf{(s,t),c,m,r}$, where $(s,t) \in \pos$, $c \in \chaldagger$, $m \in \mov$
and $r \in \rew$.  We write $\conf{(s,t),c,m,r}$
if we do not care about the owner of the configuration. \spoiler's
and \duplicator's moves are given by the following rules:

\begin{itemize}
\item \spoiler moves from a configuration $\sconf{(s,t),c,m,r}$ by:
  \begin{enumerate}
    \item if $c \neq \dagger$, moving to $\dconf{(s,t),c,m,r}$, or
    \item if $c = \dagger$, either:
          \begin{itemize}
            \item selecting $s \pijl{a} s'$ and moving to
          $\dconf{(s,t), (a,s'), (t,\frownie), *}$, or
            \item selecting $t \pijl{a} t'$ and moving to
          $\dconf{(t,s), (a,t'), (s,\frownie), \checkmark}$.
          \end{itemize}
    \item if $c \neq \dagger$, $m = (\bar{t},\frownie)$ and $\frownie \notin E$, either:
         \begin{itemize} 
           \item selecting $s \pijl{a} s'$ and moving to
              $\dconf{(s,\bar{t}),(a,s'), (\bar{t},\frownie), \checkmark}$, or
           \item selecting $\bar{t} \pijl{a} t'$ and moving to
              $\dconf{(\bar{t},s),(a,t'), (s,\frownie), \checkmark}$.
         \end{itemize}
    \item if $c = (a,s')$, $m = (\bar{t},\smiley)$  and $\smiley \notin E$, either:
         \begin{itemize}
           \item selecting $s' \pijl{b} s''$ and moving to
              $\dconf{(s',\bar{t}),(b,s''),(\bar{t},\frownie),\checkmark}$
           \item selecting $\bar{t} \pijl{b} t'$ and moving to
              $\dconf{(\bar{t},s'),(b,t'),(s',\frownie),\checkmark}$
         \end{itemize}
  \end{enumerate}

\item[]

\item \duplicator responds from a configuration $\dconf{(u,v),(a,u'),(\bar{v},f),r}$ by:
\begin{enumerate}

 \item not moving if $a = \tau$ or $f = \smiley$ and continuing in 
       $\sconf{(u',\bar{v}),\dagger,\dagger,\checkmark}$, or

 \item moving $\bar{v} \pijl{a} v'$ if available and
       continuing in $\sconf{(u,v),(a,u'),(v',\smiley),*}$
       if $f = \frownie$, or

 \item moving $\bar{v} \pijl{\tau} v'$ if available and continuing in 
       configuration $\sconf{(u,v),(a,u'),(v',f),*}$.

\end{enumerate}
\end{itemize}
\duplicator wins a finite play starting in a configuration
$\conf{(s,t),c,m,r}$ if \spoiler gets stuck, and she wins an infinite
play if the play yields infinitely many $\checkmark$ rewards.  All
other plays are won by \spoiler.  We say that a configuration is
won by a player when she has a strategy that wins all plays starting
in it.  \emph{Full} plays of the game start in a configuration
$\sconf{(s,t),\dagger,\dagger,*}$; we say that \duplicator wins the
$E$-dgb-game for a position $(s,t)$, if the configuration
$\sconf{(s,t),\dagger,\dagger,*}$ is won by her; in this case, we write
$s \gbbg_{dE} t$. Otherwise, we say that \spoiler wins that game.

\end{defi}
We claim, without further proof, that the dual generic game characterises
(again) branching bisimilarity, when $E = \emptyset$; weak
bisimilarity, when $E = \{\frownie,\smiley\}$; $\eta$-bisimilarity, when  $E =
\{\smiley\}$; and delay bisimilarity, whenever $E = \{\frownie\}$.


\section{Extensions}\label{sec:extensions} In this section we
investigate how our generic bisimulation games can be modified
to obtain other relations. 

We observe that
all relations discussed so far are not sensitive to divergences,
in the sense that a state in which a divergence (an infinite
$\tau$-path) is possible, cannot be distinguished from a state in
which no divergence is possible, but which otherwise exhibits the
exact same behaviour. Our first modification of our game is thus
to make it such that it can distinguish divergent states from
non-divergent states. For the branching bisimulation case our
modified game characterises \emph{branching bisimilarity with explicit
divergence} (also sometimes called divergence preserving branching
bisimilarity).  Second, we show that, with minimal changes, our
game can be modified to obtain the \emph{simulation} counterparts
of the bisimulations we discussed so far.

\subsection{Explicit Divergences}
\label{sec:explicit_divergences}
%

Intuitively, branching bisimulation with explicit divergence imposes the following additional constraint on top of branching bisimulation.
Whenever there is a divergent path from one of the states in a related pair 
for which all states on that path are related, the related state is also part of such a divergent path.
In other words, these branching bisimulations will preserve the existence of infinite executions
of internal actions through states with the same (behavioural) potentials.

Note that this is not the simplest possible way of defining a refinement of branching bisimulation that captures divergences. In fact, there are several competing proposals in the literature. For instance, the one studied by Bergstra \etal~\cite{bergstra_failures_1987} simply imposes that in order to be related by such a bisimulation, whenever one of the states is divergent, \ie, admits an infinite sequence of $\tau$-steps, then the other is too. This, however, produces a coarser divergence sensitive branching bisimulation than the one we study in this section, which is essentially based on~\cite{vanglabbeek_branching_2009, 
vanglabbeek_branching_1996, vanglabbeek_linear_1993, yin_branching_2014}. In addition to the (technical) argument that the resulting relation is a congruence for parallel composition and distinguishes livelocked states from deadlocked states~\cite{vanglabbeek_computation_2009}, it is argued that branching bisimulations are able 
to see (up to branching bisimulation itself) 
all the intermediate states along a computation, and therefore this idea should be preserved when considering divergent computations. Other authors preferred the opposite approach, where an even
coarser treatment of divergence to that in~\cite{bergstra_failures_1987} is proposed. For instance, 
 Walker made a quite detailed study of bisimulation and divergence in~\cite{walker90}, where a
 notion of {\em local divergence} that takes into account the possibility of executing new observable
 actions in the future, is considered. This notion was already proposed by Hennessy and Plotkin 
 as early as in 1980 ~\cite{hennessy80}, although its development had to wait ten more years.

There are also various formalisations of branching bisimilarity with explicit divergence.
Van Glabbeek \etal investigated these variations
in~\cite{vanglabbeek_branching_2009}, proving that all such variations were in essence equivalent. We here use one of their (many) equivalent characterisations:
%
%
%
%

\begin{defi}[{\cite[Condition D$_4$]{vanglabbeek_branching_2009}}]
\label{def:branching_bisimulation_with_div}
A symmetric relation $R \subseteq S \times S$ is 
a \emph{branching bisimulation with explicit divergence} if and only if
$R$ is a branching bisimulation and for all $s \rel{R} t$, if
there is an infinite sequence $s = s_0 \pijl{\tau} s_1 \pijl{\tau} s_2 \cdots$,
then there is a state $t'$ such that $t \tauplus t'$ and
for some $k$, $s_k \rel{R} t'$.
We write $s \bbedbisim t$ iff there is a branching bisimulation with explicit
divergence $R$ such that $s \rel{R} t$.
\end{defi}
We here opt to use condition D$_4$ instead of, \eg\ their condition D, in
which all states on the divergent paths are related. Partly, this
is to allow for local arguments in the game based definition and
the corresponding proof, and partly because not all their conditions
turn out to generalise straightforwardly to our $(x,y)$-generic
bisimulation. For the moment we defer a discussion on this topic
to Section~\ref{sec:diverging_discussion}.\medskip

In the previous section we have presented our $(x,y)$-generic bisimulation as a natural way of capturing, in a generic way, four kinds of bisimulation capturing abstraction. We continue this generalising approach in this section, defining $(x,y)$-\emph{generic bisimulation with explicit divergences} by adding the constraint on divergent computations from Definition~\ref{def:branching_bisimulation_with_div} to our definition of $(x,y)$-generic bisimulation from Definition~\ref{def:generic_bisimulation}.

\begin{defi}

\label{def:generic_bisimulation_with_div}
A symmetric relation $R \subseteq S \times S$ is an $(x,y)$-\emph{generic 
bisimulation with explicit divergence} if and only if
$R$ is an $(x,y)$-generic bisimulation and for all $s \rel{R} t$, if
there is an infinite sequence $s = s_0 \pijl{\tau} s_1 \pijl{\tau} s_2 \cdots$,
then there is a state $t'$ such that $t \tauplus t'$ and
for some $k$, $s_k \rel{R} t'$.
We write $s \gbbedbisim t$ iff there is an $(x,y)$-generic bisimulation with explicit
divergence $R$ such that $s \rel{R} t$.
\end{defi}

Let us next argue why this is an adequate way of defining a divergence sensitive 
extension of all the instances of our $(x,y)$-generic bisimulation. One could argue that we are adding a condition which is too strong for some of the relations since it originates from the branching bisimulation case. However, the added condition uses the transitive $\tau$-closure and refers to the particular relation we are defining in each case, and therefore it allows to consider the intermediate states along a divergent computation with respect to the corresponding equivalence. 
Essentially, we are just applying a `categorical' approach, by asking for the preservation of divergent
computations including the `observation' of the semantical information that we are capturing in 
each case. Moreover, in this way we are obtaining a uniformly defined extension of all the instances of our $(x,y)$-generic bisimulation, something that cannot be claimed if for any 
reason we would prefer the alternative finer (or coarser) extension for some (but not all) of the instances of our $(x,y)$-generic bisimulation.

Note that with the above definition at hand, it is not too hard to
prove that $\gbbedbisim$ is an equivalence relation. Moreover, it
follows immediately that any relation $R$ that is a $(b,b)$-generic
bisimulation with explicit divergence is also a branching bisimulation
with explicit divergence.
Also, $\gbbedbisim$ has the stuttering property.
\begin{lem}
  \label{lem:gbbedbisim_stuttering_property}
Let $x,y \in \{o,b\}$. Then $\gbbedbisim$ satisfies the stuttering property (Definition~\ref{def:stuttering property}).
\end{lem}
\begin{proof}
Let $t_0 \pijl{\tau} t_1 \cdots \pijl{\tau} t_k$ with $t_0 \gbbedbisim t_k$. We define relation $R$ as follows:
\[
R = \{ (t_0, t_i), (t_i, t_0) \mid 0 \leq i < k \} \cup \gbbedbisim
\]
We can prove that $R$ is a $(x,y)$-generic bisimulation relation with explicit divergence. Proving the transfer condition for all pairs $t_0 \mathrel{R} t_i$ is again analogous to the proof of \cite[Lemma 4.9.2]{sangiorgi_introduction_2012}.
For condition D$_4$ we sketch the proof.
We distinguish two cases. Suppose $i > 0$ (for $i = 0$ the result follows immediately from reflexivity).
\begin{itemize}
  \item $t_0 \mathrel{R} t_i$. Suppose there exists an infinite sequence $t_0 = \bar{t}_0 \pijl{\tau} \bar{t}_1 \pijl{\tau} \cdots$. First, observe that $t_0 \gbbedbisim t_k$, so there exists $t_k \tauplus t'$ and $\ell$ such that $t' \gbbedbisim \bar{t}_\ell$. Since $t_i \taustar t_k$, also $t_i \tauplus t'$.
  \item $t_i \mathrel{R} t_0$. Suppose there exists an infinite sequence $t_i = \bar{t}_0 \pijl{\tau} \bar{t}_1 \pijl{\tau} \cdots$. Since $i > 0$ and $t_0 \taustar t_i$, also $t_0 \tauplus t_i$, and the result follows from $t_i \gbbedbisim t_i$ due to reflexivity.
    \qedhere
\end{itemize}
\end{proof}

In the remainder of this section we present a generic game characterisation of
$(x,y)$-generic bisimilarity with explicit divergence and prove its correctness,
seeing that the proofs presented in \cite{branching_games_2016}, for the particular case of branching bisimilarity can be easily transferred to the general case.
We first recall the branching bisimulation with explicit divergence game introduced in \cite{branching_games_2016}.

\begin{defi}
\label{def:branching_bisimulation_with_divergence_game}
A \emph{branching bisimulation with explicit divergence (\bbed) game} on an LTS  $L$ is played
by players \spoiler and \duplicator on an arena  of \spoiler-owned
configurations $\sconf{(s,t),c,r}$ and \duplicator-owned configurations
$\dconf{(s,t),c,r}$, where $((s,t),c,r) \in \pos \times \chaldagger \times
\rew$, and $\pos$, $\chaldagger$ and $\rew$ are as before.
By convention, we write $\conf{(s,t),c,r}$ if we do not care about the
owner of the configuration. 
\begin{itemize}
\item \spoiler moves from a configuration $\sconf{(s,t),c,r}$ by:
  \begin{enumerate}
    \item selecting $s \pijl{a} s'$ and moving to
$\dconf{(s,t),(a,s'),*}$ if
$c = (a,s')$ or $c = \dagger$, and to $\dconf{(s,t),(a,s'),\checkmark}$,
otherwise; or

    \item picking some $t \pijl{a} t'$ and moving to
$\dconf{(t,s),(a,t'),\checkmark}$.
  \end{enumerate}

\vspace{0.3cm}
\item \duplicator responds from a configuration $\dconf{(u,v),(a,u'),r}$ by:
\begin{enumerate}

 \item \label{bbged_terminate} not moving if $a = \tau$ and continuing in
    configuration $\sconf{(u',v),\dagger,*}$, or, 

  \item moving $v \pijl{a} v'$ if available
    and continuing in configuration $\sconf{(u',v'), \dagger, \checkmark}$, or

 \item moving $v \pijl{\tau} v'$ if available
    and continuing in configuration $\sconf{(u,v'), (a,u'), *}$.
\end{enumerate}
\end{itemize}
\duplicator wins a finite play starting in a configuration
$\conf{(s,t),c,r}$ if \spoiler gets stuck, and she wins an infinite
play if the play yields infinitely many $\checkmark$ rewards. All
other plays are won by \spoiler. 
We say that a configuration is won by a player when 
she has a strategy that wins all plays starting in it.
{\em Full} plays of the game start in a configuration  $\sconf{(s,t),\dagger,*}$;
we say that 
 \duplicator wins the \bbed-game for a position $(s,t)$, if the configuration $\sconf{(s,t),\dagger,*}$
is won by it; in this case, we write $s \bbedg t$. Otherwise, we say that \spoiler wins that game.
\end{defi}

Comparing this definition with Definition~\ref{def:branching_bisimulation_game}, it is hard to pinpoint the difference due to the similarities between the games. In fact, the only difference is in \duplicator's first rule (rule ~\eqref{bbg_terminate} in Definition~\ref{def:branching_bisimulation_game}), in which idling in response to an internal move is no longer rewarded, whereas it was in the \bb-game.
As we proved in~\cite{branching_games_2016}, this single change is sufficient to 
capture branching bisimulation with explicit divergence. We here only repeat this result.
\begin{thm}
We have $\bbedbisim$ = $\bbedg$.
\end{thm}

While the formal proof of correctness of this game in~\cite{branching_games_2016} is fairly involved, the intuition is quite simple. First, observe that in order to check a divergence, \spoiler can challenge \duplicator by presenting the internal steps on the divergence one by one. \duplicator is forced to reply to infinitely many of these internal steps by an internal step of her own, otherwise she loses the game since she does not collect infinitely many $\checkmark$ rewards.
Instead, in the original \bb-game \duplicator could win by replying to a divergence by just remaining idle, because her first rule also rewarded her with a $\checkmark$. We also need to check that when not rewarding `isolated' $\tau$ moves that \duplicator replies to by remaining idle, \duplicator cannot (incorrectly) be declared loser of a game. This cannot happen, simply because an infinite play that is not checking a divergence will include infinitely many challenges that correspond to visible actions, and \duplicator should match all of them, so that she will collect infinitely many $\checkmark$ rewards that she needs to win the play, even if she does not get the $\checkmark$ rewards corresponding to the idling responses to $\tau$ moves.\medskip

We next use the idea from Definition~\ref{def:branching_bisimulation_with_divergence_game} to introduce the \emph{generic bisimulation with explicit divergence (\gbed) game}.

\begin{defi}
\label{def:generalised game_with_divergences}
Let 
$E \subseteq \{ \smiley, \frownie \} $.
A \emph{generic bisimulation with explicit divergence (\gbed) game} (or $E$-bisimulation game with explicit divergences) on an LTS $L$ is played by players \spoiler and \duplicator
on an arena  of \spoiler-owned configurations $\sconf{(s,t),c,m,r}$
and \duplicator-owned configurations $\dconf{(s,t),c,m,r}$, where
$(s,t) \in \pos$, $c \in \chaldagger$, $m \in \mov$, $r \in \rew$, and
$\pos, \chaldagger, \rew$ and $\mov$ are as before.
We write $\conf{(s,t),c,m,r}$ if we do not care about
the owner of the configuration. 
\begin{itemize}
\item From a configuration $\sconf{(s,t),c,m,r}$, \spoiler can:

  \begin{enumerate}
    \item \label{gbbged_move_unchanged} move to $\dconf{(s,t),c,m,*}$ if $c \neq \dagger$, or

    \item for some $s \pijl{a} s'$, move to either:
       \begin{enumerate}
   \item\label{gbbged_spoiler_fresh_challenge} $\dconf{(s,t),(a,s'),(t,\frownie),*}$, if $c = \dagger$, or
\item\label{gbbged_spoiler_renew_challenge}
       $\dconf{(s,t),(a,s'),(t,\frownie),\checkmark}$, if $c \neq (a,s')$
      \end{enumerate}
    \item\label{gbbged_spoiler_switches} for some $t \pijl{a} t'$, move to:
$\dconf{(t,s),(a,t'),(s,\frownie),\checkmark}$.

  \end{enumerate} 
\item[]

\item From a configuration $\dconf{(u,v),(a,u'),(\bar{v},f),r}$, \duplicator can:
\begin{enumerate}

 \item \label{gbbged_terminate} 
move to $\sconf{(u',\bar{v}),\dagger,\dagger,*}$ when $a = \tau$, or

 \item if $f = \frownie$ and $\bar{v} \pijl{a} v'$, move to one of the following:
\begin{enumerate}
   \item \label{gbbged_match_and_move}  $\sconf{(u',v'),(a,u'),(v',\smiley),*}$, in any case, or
   \item \label{gbbged_match_move_and_terminate} $\sconf{(u',v'),\dagger,\dagger,\checkmark}$, in any case, or
   \item \label{gbbged_match_and_postpone} $\sconf{(u,v),(a,u'),(v',\smiley),*}$, only if  $\smiley \in E$
 \end{enumerate} 
     
 \item for some $\bar{v} \pijl{\tau} v'$, move to one of the following:
 \begin{enumerate}
   \item \label{gbbged_push_pebble_and_move} $\sconf{(u,v'),(a,u'),(v',f),*}$,  in any case, or
   \item \label{gbbged_push_pebble_and_terminate} $\sconf{(u',v'),\dagger,\dagger,\checkmark}$, only if $f = \smiley$
   \item \label{gbbged_push_pebble_and_wait} $\sconf{(u,v),(a,u'),(v',f),*}$, only if  $f \in E$
 \end{enumerate} 

\end{enumerate}

\end{itemize}

\duplicator wins a finite play starting in a configuration
$\conf{(s,t),c,m,r}$ if \spoiler gets stuck, and she wins an infinite
play if the play yields infinitely many $\checkmark$ rewards.  All
other plays are won by \spoiler.  We say that a configuration is
won by a player when she has a strategy that wins all plays starting
in it.  {\em Full} plays of the game start in a configuration
$\sconf{(s,t),\dagger,\dagger,*}$; we say that \duplicator wins the
game for a position $(s,t)$, if the configuration
$\sconf{(s,t),\dagger,\dagger,*}$ is won by her.
In this case, we write $s \gbbg_E^{ed} t$. Otherwise, we say that \spoiler wins that game.
 \end{defi}
As before for the \bb-game, the modification compared to Definition~\ref{def:generalised game} is marginal: only the $\checkmark$ reward in \duplicator's rule ~\eqref{gbbg_terminate} has been turned into $*$. The idea here is the same as for the \bbed-game.\medskip

We next prove that for all $x,y \in \{o,b\}$ the relation induced by the $E(x,y)$-bisimulation game exactly captures
$(x,y)$-generic bisimilarity with explicit divergence. 
We split the proof obligations into three separate lemmata, first addressing
completeness (Lemma~\ref{lem:bbedbisim_is_a_bbedg}) and next addressing
soundness (Lemmata~\ref{lem:game_with_divergence_is_generic_branching_bisim} and~\ref{lem:bbedg_is_a_bbbedbisim}).

\begin{lem}[Completeness]\label{lem:bbedbisim_is_a_bbedg}
For all $x,y \in \{ o, b \}$, whenever we have $s \gbbedbisim t$, we also have $s \gdxybbg t$.
\end{lem}

\begin{proof}
We again design a partial strategy for \duplicator 
for the \gbed-game that starts in $\sconf{(s,t), \dagger,\dagger,*}$. 
We first construct a \duplicator strategy using the same construction as used for defining the strategy in the proof of Lemma~\ref{lem:gbbbisim_is_gbbg} to win the corresponding \gb-game, but using $\gbbedbisim$ instead of $\gbbbisim$ as relation $R$ in the construction.
%
However, if we do not change anything in this 
strategy, it could be the case that
\spoiler now wins the \gbed-game, since the strategy does not take divergences into account. Let us see which changes are needed to 
guarantee that \duplicator will also win the \gbed-game. 

First, note that all the {\em positions} along any play consistent with
that winning strategy for \duplicator contain two $\gbbbisim$ equivalent states, reusing the proof in Lemma~\ref{lem:game_is_generic_branching_bisim}. 
Second, observe that we start from a
configuration $\sconf{(s,t), \dagger,\dagger,*}$ containing two $\gbbedbisim$ equivalent states, and 
in order to be able to repeat our arguments after any move of \duplicator, we must
preserve the $\gbbedbisim$ relation, and not just $\gbbbisim$, as in the proof of Lemma~\ref{lem:gbbbisim_is_gbbg}.


If we apply the strategy that we have defined for the \gb-game, the only case in which 
\duplicator loses the game 
is that in which she is generating infinitely many $\dagger$ challenges, but only finitely many $\checkmark$ rewards.
In particular, there would be some suffix of a play in which \duplicator generates infinitely many $\dagger$ challenges, but no $\checkmark$ reward.
We consider that suffix as a full play, and denote by $(s_0, t_0)$ the position at the configuration that suffix starts.

Let us first make a few observations about the moves played by both players along this suffix:
\begin{itemize}
  \item \spoiler only plays moves \eqref{gbbged_move_unchanged} or  \eqref{gbbged_spoiler_fresh_challenge};
  \item \duplicator only plays move \eqref{gbbged_terminate}
\end{itemize}
We analyse why this must be the case. First, observe that in \spoiler's moves \eqref{gbbged_spoiler_renew_challenge} and \eqref{gbbged_spoiler_switches} and \duplicator's moves \eqref{gbbged_match_move_and_terminate} and \eqref{gbbged_push_pebble_and_terminate} immediately a $\checkmark$ is rewarded, which is a contradiction to the assumption that there are no $\checkmark$ rewards on the suffix. Now, suppose \duplicator plays any of the moves \eqref{gbbged_match_and_move}, \eqref{gbbged_match_and_postpone}, \eqref{gbbged_push_pebble_and_move}, or \eqref{gbbged_push_pebble_and_wait}. In each of these cases, when playing according to the strategy defined in the proof of Lemma~\ref{lem:gbbbisim_is_gbbg}, when this move is played, eventually the $\checkmark$ is obtained through either \duplicator's move \eqref{gbbged_match_move_and_terminate} or \eqref{gbbged_push_pebble_and_terminate}.\footnote{This can be checked using a tedious but elementary analysis on the cases in the defined strategy.}

Now, since \duplicator is always playing using rule~\eqref{gbbged_terminate}, all challenges involved in the considered infinite suffix concern $\tau$ actions (by definition of rule~\eqref{gbbged_terminate}), and generate a divergent sequence
$ s_0 \pijl{\tau} s_1 \pijl{\tau} s_2 \pijl{\tau} \cdots$, and \duplicator always responds by leaving $t_0$ the
same, so that the invariance of $\gbbedbisim$ implies that $ s_i \bbedbisim t_0$, for all $i$.
But then, by definition of $\gbbedbisim$,
there must be some sequence of transitions $t_0 \taustar^+ t'$ such that for some $k$, $s_k \gbbedbisim t'$. Write this sequence as $t_0 \pijl{\tau} t_1 \pijl{\tau} \cdots \pijl{\tau} t_{\ell} = t'$ (for $\ell > 0$).
Since $s_i \gbbedbisim t_0$ for all $i$, and $\gbbedbisim$ is an equivalence relation, we have $s_i \gbbedbisim s_j$ for all $i, j$. In a similar vein, we can also conclude $t_0 \gbbedbisim t_{\ell}$.
Since $\gbbedbisim$ has the stuttering property according to Lemma~\ref{lem:gbbedbisim_stuttering_property}, we now find that $t_i \gbbedbisim t_j$ for all $0 \leq i,j \leq \ell$, and hence $s_0 \gbbedbisim t_i$ for $0 \leq i \leq \ell$.

This allows us to modify the strategy \duplicator plays from the configurations $\dconf{(s_0, t_i), (\tau, s_1), (t_i, \frownie), *}$ for $i < \ell - 1$ in such a way that she plays the $\tau$-step from $t_i$ using rule \eqref{gbbged_push_pebble_and_move} to $\sconf{(s_0, t_{i+1}), (\tau, s_1), (t_i, \frownie), *}$. From $\dconf{(s_0, t_{\ell-1}), (\tau, s_1), (t_{\ell-1}, \frownie), *}$, she now plays to $\sconf{(s_1, t_{\ell}), \dagger, \dagger, \checkmark}$.

Note that all configurations that we play to in this changed strategy are \good.
From $\sconf{(s_1, t_{\ell}), \dagger, \dagger, *)}$, \duplicator can proceed as she would from $\sconf{(s_1, t_{\ell}), \dagger, \dagger, *)}$, and apply the same modifications as described above.

In this way we get a revised strategy for \duplicator
that will allow her to win the \gbed-game that starts 
in $\sconf{(s,t), \dagger,\dagger,*}$, thus proving $s \gdxybbg t$.
\end{proof}

We continue by showing that the relation induced by the $\gbed$-game is, in fact, an $(x,y)$-generic bisimulation (without the divergence requirement). Note that this result follows more or less by design since the \gbed-game is stricter than our original \gb-game since it rejects plays that were once winning for \duplicator.
\begin{lem}\label{lem:game_with_divergence_is_generic_branching_bisim}
For all $x,y \in \{o,b\}$, the relation $\gdxybbg$ is an  $(x,y)$-generic bisimulation.
\end{lem}
\begin{proof}
Since the $E(x,y)$-bisimulation with explicit divergence games are obtained from the corresponding  $E(x,y)$-bisimulation games simply
by turning some $\checkmark$ rewards into $*$, and in this way any configuration that is won by 
\duplicator at the former is also a winning configuration for her at the latter, we can repeat 
the reasoning in the proof of Lemma~\ref{lem:game_is_generic_branching_bisim}
substituting the  $\checkmark$ reward by a $*$ reward whenever \duplicator resorts to choosing her first 
option, to obtain the proof that $ \gdxybbg$ is an  $(x,y)$-generic bisimulation.
\end{proof}

The following lemma confirms that our new game is indeed capable of discerning
states that are divergent and states that are non-divergent. More specifically,
it states that the relation induced by $\gdxybbg$ meets divergence condition D$_4$.
\begin{lem}\label{lem:bbedg_is_a_bbbedbisim}
Let  $s \gdxybbg t$, and assume that we have a divergent sequence
$s = s_0 \pijl{\tau} s_1 \pijl{\tau} s_2 \pijl{\tau} \cdots$. Then there
is some $t'$ and some $k$ such that $t \tauplus t'$ 
and $s_k \gdxybbg t'$.

\end{lem}
\begin{proof}
Towards a contradiction, suppose that for all $t'$ for which $t
\tauplus t'$, and for all $k$, we have $s_k \not\gdxybbg t'$.
Consider \spoiler's partial strategy to use
rule~\eqref{gbbged_spoiler_fresh_challenge} with challenge $s_i
\pijl{\tau} s_{i+1}$ for all configurations of the form
$\sconf{(s_i,t'),\dagger,\dagger,r}$, and for
configurations of the form
$\sconf{(s_i,t'),(\tau,s_{i+1}),(\bar{t},f),r}$, with $t
\tauplus \bar{t}$, use rule~\eqref{gbbged_move_unchanged} when
$t' = t$, and when $t' \neq t$, the strategy that mimics \spoiler's winning strategy
for $\sconf{(s_i,t'),\dagger,\dagger,*}$.
Note that since \duplicator wins
$\sconf{(s_0,t),\dagger,\dagger,*}$, all plays consistent with
\duplicator's winning strategy visit \spoiler-owned configurations
for which the above strategy is defined.

Consider the play that emerges from $\sconf{(s_0,t),\dagger,\dagger,*}$
by following \spoiler's strategy and \duplicator's winning strategy.
Since \spoiler would only make moves generated by $\tau$ transitions
when following her fixed strategy, \duplicator can only reply using $\tau$ transitions as well. Moreover, to win the play she needs to collect infinitely many 
$\checkmark$ rewards. In order to collect the first of them, she needs to sometime apply either 
rule~\eqref{gbbged_match_move_and_terminate}
or~\eqref{gbbged_push_pebble_and_terminate}. When doing this she
is advancing her state by executing a $\tau$ transition, thus moving to a
configuration with position $(s_k,t')$, for some $k \ge 0$ and $t'$ such 
that $t \tauplus t'$. But by definition, \spoiler's strategy is winning for 
configurations with such positions, and therefore the play is won by
\spoiler. 
As a result we find
$s_0 \not\gdxybbg t$, which contradicts our assumptions. So there must be some $k$ and
some $t'$ such that $t \tauplus t'$ and $s_k \gdxybbg t'$.
\end{proof}

By combining the three preceding lemmata, we can next conclude that 
the \gbed-game characterises
our relational definition of $(x,y)$-generic bisimilarity with explicit divergence.
\begin{thm}\label{th:bbbedbisim_is_bbedg} For all $x,y \in \{o,b\}$, we have $\gbbedbisim\, =\; \gdxybbg $.
\end{thm}
\begin{proof}
The implication from left to right follows from Lemma~\ref{lem:bbedbisim_is_a_bbedg}, the implication from right to left follows from Lemmata~\ref{lem:game_with_divergence_is_generic_branching_bisim}~and~\ref{lem:bbedg_is_a_bbbedbisim}.
\end{proof}

\begin{exa}
Recall the labelled transition system from Figure~\ref{fig:need_for_challenges}
(depicted left below), with the corresponding \gbed-game with $E = \emptyset$ 
(shown right below).
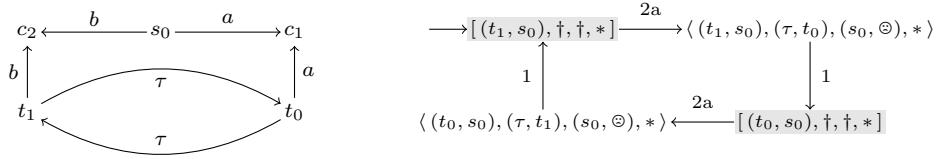
\begin{center}
\begin{minipage}{.35\textwidth}
  \begin{tikzpicture}
    [node distance=30pt,inner sep = 1pt,minimum size=10pt]
        \scriptsize
    \node (u_)  {$s_0$};
    \node[right of=u_,xshift=20pt] (c1_) {$c_1$};
    \node[left of=u_,xshift=-20pt] (c2_) { $c_2$};
    \node[below of=c1_] (v_) {$t_0$};
    \node[below of=c2_] (w_) {$t_1$};

    \path[->]
      (u_) edge node[above] {$a$} (c1_) edge node[above] {$b$} (c2_)
      (v_) edge[bend left] node[above] {$\tau$} (w_) edge node [right] {$a$}  (c1_)
      (w_) edge[bend left] node[below] {$\tau$} (v_) edge node [left] {$b$} (c2_)
    ;
    \end{tikzpicture}
\end{minipage}
\begin{minipage}{.35\textwidth}
\centering
  \begin{tikzpicture}
    [node distance=25,inner sep = 1pt,minimum size=5pt,initial text={}]

     \tikzstyle{state}=[rectangle, draw=none]
     \tikzstyle{transition}=[->,>=stealth']

      \node[initial, state, fill=gray!20] (sACdds) {\tiny $\sconf{(t_1,s_0),\dagger,\dagger,*}$};
      \node[right of=sACdds, state,xshift=75pt] (ACaBCfs) {\tiny $\dconf{(t_1,s_0),(\tau,t_0),(s_0,\frownie),*}$};
      \node[below of=ACaBCfs, state,yshift=-10pt,fill=gray!20] (sADaBDfs) {\tiny $\sconf{(t_0,s_0),\dagger,\dagger,*}$};
      \node[below of=sACdds, state,yshift=-10pt] (ADbBDfc) {\tiny $\dconf{(t_0,s_0),(\tau,t_1),(s_0,\frownie),*}$};

      \draw [->] (sACdds) edge node[shape=circle,above,inner sep=2pt] {\tiny\ref{gbbged_spoiler_fresh_challenge}} (ACaBCfs);
      \draw [->] (ACaBCfs) edge node[shape=circle,right,inner sep=2pt] {\tiny\ref{gbbged_terminate}} (sADaBDfs);
      \draw [->] (sADaBDfs) edge node[shape=circle,above,inner sep=2pt] {\tiny\ref{gbbged_spoiler_fresh_challenge}} (ADbBDfc);
      \draw [->] (ADbBDfc) edge node[shape=circle,left,inner sep=2pt] {\tiny\ref{gbbged_terminate}} (sACdds);

  \end{tikzpicture}
\vspace{5pt}
\end{minipage}
\end{center}
 Observe that $s_0$ and $t_1$ are branching bisimilar, but not branching bisimilar with explicit divergence. \spoiler tries to show they are not equivalent by challenging \duplicator with $\tau$ transitions from $t_1$, and subsequently $t_0$ using her rule \eqref{gbbged_spoiler_fresh_challenge}. Since in $s_0$ \duplicator cannot play any $\tau$ transitions, she can only respond with her move \eqref{gbbged_terminate}, and she loses since she never gains any $\checkmark$ reward.
\end{exa}

In Section~\ref{sec:smallapplication}, we present a larger application of the \gbed-game.

\subsection{A Note on Divergence Conditions}
\label{sec:diverging_discussion}

Branching bisimilarity with explicit divergence can be defined in
terms of a relational characterisation as we did in the previous
section, but also in terms of a modal characterisation or in terms
of coloured traces. The relational characterisation of branching
bisimilarity with explicit divergence essentially builds on the
relational definition of branching bisimulation, adding an
orthogonal divergence condition to it. The literature, however,
defines several incomparable divergence conditions which are all
used to this effect. It turns out that all of these alternatives
(referred to as condition~D, D$_0$, D$_1$, D$_2$, D$_3$ and D$_4$)
give rise to the same behavioural equivalence relation,
see~\cite{vanglabbeek_branching_2009}.

Our definition of $(x,y)$-generic bisimilarity with explicit
divergence is based on divergence condition D$_4$, see also
Definition~\ref{def:branching_bisimulation_with_div}. Yet, a natural
question is whether the alternative divergence conditions also give
rise to equivalent characterisations for our $(x,y)$-generic
bisimilarity with explicit divergence. The answer to this question is
negative: using condition D$_2$ gives rise to a different behavioural
relation than condition D$_4$. 

Formally condition D$_2$ requires of a relation $R$ 
that for every $s \rel{R} t$ for which we have $s = s_0 \pijl{\tau} s_1
\pijl{\tau} \dots$, there is some immediate $\tau$-successor $t'$ of $t$ 
such that $s_k \rel{R} t'$ for some $k$. Note that this strengthens
condition D$_4$, which requires that $t'$ is some state that can be
reached via one or more $\tau$-steps. 
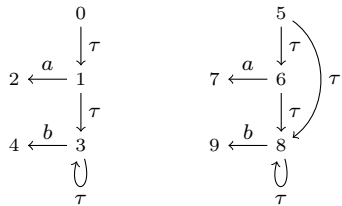
\begin{figure}[hbpt]
  \begin{center}
  \begin{tikzpicture}
    [node distance=25,inner sep = 1pt,minimum size=10pt]

     \tikzstyle{state}=[circle, draw=none]
     \tikzstyle{transition}=[->,>=stealth']

      \node[state] (t0) {\tiny 0};
      \node[state] [below of=t0] (t1) {\tiny 1};
      \node[state] [left of=t1] (t2) {\tiny 2};
      \node[state] [below of=t1] (t3) {\tiny 3};
      \node[state] [left of=t3] (t4) {\tiny 4};
      \draw [->] (t0) edge node[right] {\scriptsize $\tau$} (t1);
      \draw [->] (t1) edge node[above] {\scriptsize $a$} (t2);
      \draw [->] (t1) edge node[right] {\scriptsize $\tau$} (t3);
      \draw [->] (t3) edge[loop below] node[below] {\scriptsize $\tau$} (t3);
      \draw [->] (t3) edge node[above] {\scriptsize $b$} (t4);

      \node[state] [right of=t0, xshift=50pt] (s0) {\tiny 5};
      \node[state] [below of=s0] (s1) {\tiny 6};
      \node[state] [left of=s1] (s2) {\tiny 7};
      \node[state] [below of=s1] (s3) {\tiny 8};
      \node[state] [left of=s3] (s4) {\tiny 9};
      \draw [->] (s0) edge node[right] {\scriptsize $\tau$} (s1);
      \draw [->] (s1) edge node[above] {\scriptsize $a$} (s2);
      \draw [->] (s1) edge node[right] {\scriptsize $\tau$} (s3);
      \draw [->] (s3) edge[loop below] node[below] {\scriptsize $\tau$} (s3);
      \draw [->] (s3) edge node[above] {\scriptsize $b$} (s4);
      \draw [->] (s0) edge[bend left=55] node[right] {\scriptsize $\tau$} (s3);

  \end{tikzpicture}
  \end{center}
\caption{An illustration of the difference between condition D$_2$ and
D$_4$.}
\label{fig:counterexample}
\end{figure}

The difference between these two conditions is illustrated through
the transition system depicted in Figure~\ref{fig:counterexample}.
While we cannot expect the divergence condition D$_2$ and D$_4$ for
$(b,b)$-generic
bisimilarity with explicit divergence to give different results, they do
give different results for
$(o,o)$-generic bisimilarity with explicit divergence. First, observe
that we have $0 \gbbbisimargs{o,o} 5$. States $0$ and $5$ can,
however, not be related through a weak bisimulation relation
that meets condition D$_2$. The offending behaviour here is exactly
the divergent computation $5 \pijl{\tau} 8 \pijl{\tau} 8 \pijl{\tau}
\dots$ which cannot be mimicked from state $0$ because state $1$
cannot be related to any state on the divergent computation. On the
other hand, we do have $0 \gbbedbisimargs{o,o} 5$. This can be seen
using our game characterisation: if \spoiler challenges \duplicator
to mimic the divergent computation from state $5$, then \duplicator
can respond by moving from state $0$ to state $3$ using
rule~\eqref{gbbged_push_pebble_and_wait} and subsequently using
rule~\eqref{gbbged_match_move_and_terminate}.

The above observation raises several questions. For instance, it
is not clear whether there is a game characterisation of $(x,y)$-generic
bisimilarity with explicit divergence based on divergence condition
D$_2$. The same can be asked for some of the other divergence
conditions.  Moreover, one can wonder for which of the conditions
the various instances of $(x,y)$-generic bisimilarity with explicit
divergence would coincide; do they all give rise to equivalence
relations, \emph{etcetera}. We defer answering these questions to
future research.

\subsection{Simulation Games}
So far we have considered four bisimulation relations and showed that explicit divergences can be added orthogonally. All of these relations are equivalence relations. When checking an implementation relation, sometimes it is desirable to drop the symmetry requirement, and use simulation relations instead. In this section we sketch how our generic bisimulation game can be adapted to reflect the corresponding simulation relations. We here forego a formal treatment, since most of the required modifications are straightforward, and the corresponding proofs only require minor changes compared to what we have presented so far.


The similarity versions of the four instances of our generic bisimulation already 
appeared in the preliminary version of~\cite{vanglabbeek_linear_1993}, and some of them we also
mentioned later, see, \eg~\cite{gerth_partial_1999}, where the induced bisimilarity notions
were studied in depth. Of course, these similarity notions generate {\em preorders},
although the corresponding {\em kernels} would again provide equivalence relations. Van Glabbeek and Weijland \cite{vanglabbeek_branching_1996} also briefly mention the possibility of defining $\eta$, delay- and branching bisimulation preorders with explicit divergences. However, in this case the starting point
is not a simulation relation, but the corresponding bisimulation equivalence, that does not (still) take into account divergences. Then, the treatment of these divergences induces a finer preorder, where behaviours
without divergences are `better' than the equivalent (up to divergences) ones that contain some divergent computations. Also, the notion of divergence preserving branching simulation defined in~\cite{reniers_results_2014} comes quite close to a branching bisimulation preorder with explicit divergences.

We next define \emph{generic simulation}, by dropping the symmetry requirement from Definition~\ref{def:generic_bisimulation}.
\begin{defi}
\label{def:generic_simulation}
 A relation $R \subseteq S \times
 S$ is an $(x,y)$-\emph{generic simulation}, whenever $s
 \rel{R} t$ and $s \pijl{a} s'$ imply either:
\begin{itemize}
 \item $a = \tau$ and $s' \rel{R} t$, or 
 \item there exist states $t', t_1,t_2$ such that $t
 \gentaustar{x,R,s} t_1 \pijl{a} t_2 \gentaustar{y,R,s'} t'$ and $s' \rel{R} t'$.  
\end{itemize}
 Now, for the corresponding values of $x$ and $y$, we write $s \gbbsim
 t$ and say that $s$ is $(x,y)$-simulated by  $t$  iff there is an
 $(x,y)$-generic simulation $R$ such that $s \!\rel{R} t$.
 Typically, we simply write $\gbbsim$ to denote  $(x,y)$-\emph{generic
 similarity}, and we write $\gbbsimeq$ to denote the $(x,y)$-\emph{generic
 simulation equivalence} $\gbbsim \cap \gbbsim^{-1}$.

\end{defi}
As before, we have the following correspondences for $R \subseteq S \times S$:
\begin{itemize}
  \item $R$ is a weak simulation iff it is an $(o,o)$-generic simulation;
  \item $R$ is a delay simulation iff it is an $(o,b)$-generic simulation;
  \item $R$ is an $\eta$-simulation iff it is a $(b,o)$-generic simulation;
  \item $R$ is a branching simulation iff it is a $(b,b)$-generic simulation.
\end{itemize}

Van Glabbeek already proved in~\cite{vanglabbeek_linear_1993}, using his testing characterisations,
that the preorders induced by the first two classes of simulations above are the same, and that this is also the case for those preorders induced by the latter two.
Next we provide a direct proof based on the definition of our (coinductive) generic simulations.

\begin{prop} \label{reduced_set_generic_simulations}
For all $x,y,z \in \{o,b \}$ we have  $ s \leq_{(z,x)} t   \Leftrightarrow  s \leq_{(z,y)} t$.
\end{prop}
\begin{proof}
For each relation $R$ we define the derived relation $R^{\taustarm}$ by  
\[
 R^{\taustarm} = \{ (s,t) \mid \exists t' \colon t \taustar t' \wedge s \rel{R} t' \}.
 \]
Let us see that whenever $R$ is a $(z,o)$-generic simulation, the derived relation 
$R^{\taustarm}$ is too. Indeed, when we have
 $ s \rel{R^{\taustarm}} t$, we will have some $t'$ with both
   $ t \taustar t' $ and $s \rel{R} t'$. Now, if 
   $s \pijl{a} s'$, then there exist states $t'', t_1',t_2'$ such that 
   $t' \gentaustar{z,R,s} t_1' \pijl{a} t_2' \gentaustar{o,R,s'} t''$ and $s' \rel{R} t''$.
   We also have $t' \gentaustar{z,R^{\taustarm},s} t_1' \pijl{a} t_2' $,
   since for each intermediate state $t'''$ in the computation  $ t \taustar t' $
   we always have   $ s \rel{R^{\taustarm}} t'''$. And from 
    $s' \rel{R} t''$  we can infer $ s' \rel{R^{\taustarm}} t_2'$. In this way we finally obtain
      $t' \gentaustar{z,R^{\taustarm},s} t_1' \pijl{a} t_2' \gentaustar{o,R^{\taustarm},s'} t_2'$ 
      with $s' \rel{R^{\taustarm}} t_2'$.
But since the last weak transition in the obtained weak computation is an empty transition, we
also have  $t' \gentaustar{z,R^{\taustarm},s} t_1' \pijl{a} t_2' \gentaustar{b,R^{\taustarm},s'} t_2'$ 
with $s' \rel{R^{\taustarm}} t_2'$,
thus proving that $R^{\taustarm}$ is also a $(z,b)$-generic simulation.
\end{proof}

A generic weak simulation game can be obtained from Definition~\ref{def:branching_bisimulation_game} by disallowing
\spoiler to choose her third option.

We remark that we could use this game characterisation to provide an alternative proof
for Proposition~\ref{reduced_set_generic_simulations} along the following line of
reasoning. Observe that \duplicator
can avoid the use of move~\eqref{gbbged_match_and_postpone} even when
playing an $E$-simulation game with $\smiley \in E$. The moves that she plays when
playing those games are always valid for the $E'$-game, where $E' = E \setminus \{ \smiley \} $, hence
proving the desired game equivalence. 

Note that when playing move~\eqref{gbbg_match_and_move} 
instead of move~\eqref{gbbg_match_and_postpone}, \duplicator allows the next challenge from
\spoiler to correspond to the reached pair of states $(u',v')$, instead of some pair $(u',v'')$,
with  $ v' \taustar v'' $, that she could propose after playing the sequence of 
\eqref{gbbg_push_pebble_and_wait} moves corresponding to that weak computation.
In the case of the simulation game \spoiler can, however, not take any advantage from a
new challenge starting from  $(u',v')$ instead of $(u',v'')$, since her move always uses a
transition $u' \pijl{a} u''$, and then \duplicator can start her reply playing the moves corresponding
to the weak computation  $ v' \taustar v'' $ again.\footnote{In the bisimulation game, of course,
\spoiler is much stronger, as she can pose a challenge using any transition  $ v' \pijl{a} v''' $, changing the side on which she plays.}

\begin{exa}
If we now reconsider the example we took from~\cite{korver_computing_1992} in Figure~\ref{fig:example_weak_vs_branching} on page~\pageref{fig:example_weak_vs_branching}, we note
that state $0$ is not branching simulated by state $5$, which can be proved
following the same arguments as used in that section. 
Instead, state $5$ is branching
simulated by state $0$, as the last can copy any move from the former, eventually arriving at states that are trivially equivalent.
\end{exa}

A game characterisation of the simulation equivalences can equally 
straightforwardly be obtained from our definitions, by only allowing \spoiler 
to choose her third option for her moves (in the bisimulation game) during the first round of the game, and disallowing this option in any subsequent rounds.
Of course, the corresponding simulation equivalence relation that one obtains 
in this way is coarser than the corresponding bisimulation: \spoiler has a
much bigger power if she can switch the board at any round.

Similarly, by combining the modification to obtain the \gbed-game (Definition~\ref{def:generalised game_with_divergences}) with the change needed to obtain the simulation game, we could obtain games for the corresponding simulation relations with explicit divergence and the corresponding simulation equivalence by restricting \spoiler's options.


\section {A Small Application}\label{sec:smallapplication}

The game characterisation we discussed in the preceding sections
first and foremost provide a useful operational understanding of
the four weak bisimulations.  In particular, our games allow for a nice
intuitive operational explanation of the \emph{inequivalence} of
states: \spoiler's winning strategy essentially guides one to
the offending states. 
Our operational explanation of the inequivalence of states can be
used complementary to a distinguishing formula, if available for the
equivalence used to compare the two states.\footnote{We
are not aware of distinguishing formulae for $\eta$-bisimulation
and delay-bisimulation, nor of any of their divergence-aware variants
or their simulation equivalence variants. Distinguishing formulae for
branching bisimulation can be found in~\cite{korver_computing_1992} and
distinguishing formulae for branching bisimulation with explicit divergence
are studied in~\cite{vanglabbeek_branching_2009}. For weak bisimulation,
distinguishing formulae are essentially obtained by a `weak' Hennessy-Milner
logic.} 

A distinguishing
formula is a modal formula that holds for exactly one of the two (inequivalent) transition
systems being compared and it provides a high-level explanation for their inequivalence.
By high-level, we mean that it will not
only hold for one of the two transition systems at hand, but by their nature, it
will also hold for every
transition system that is equivalent to the one for which it holds. Since
such transition systems are in general structurally quite different, the
formula can therefore only explain the difference between the two given transition
systems in terms of their capabilities and not in terms of their structure.

In contrast, \spoiler's winning strategy is specific for the two given
transition systems and therefore provides a much more fine-grained explanation as to the
inequivalence of the two transition systems than a distinguishing formula.
Such a structure-based operational explanation
of inequivalence can be more accessible than a
distinguishing formula, pointing at concrete states that cause the inequivalence. 
We will forego a precise comparison between the game-theoretical and 
the various logical approaches, leaving a comparison and deeper
connections between the two approaches to future work.
Instead, we will illustrate some of the differences between the two techniques 
through two examples.

As a first example, reconsider the LTSs of Figure~\ref{fig:example_weak_vs_branching}
on page~\pageref{fig:example_weak_vs_branching}, taken
from~\cite{korver_computing_1992}.  The
distinguishing formula given in~\cite{korver_computing_1992} for
illustrating that states $0$ and $5$ 
are not branching bisimilar is
$\neg ( (\textit{tt}\langle b \rangle \textit{tt}) \langle a \rangle
\textit{tt})$. 
This formula explains the inequivalence between states $0$ and $5$ by asserting that
state $0$ (for which the formula fails)
may `engage in an $a$-step, while in all intermediate states (state
$0$ in this case) a $b$-step is available'~\cite{korver_computing_1992},
whereas this is not true of state $5$ (for which the formula holds).\footnote{Informally, a formula
of the form $\phi \langle a \rangle \psi$ states that there is a
finite $\tau$-path on which $\phi$ holds until action $a$, after which
$\psi$ holds. By nesting these \emph{until}-operators one can state more complex 
properties: $\textit{tt} \langle b \rangle \textit{tt}$ means that a $b$-action
can be `weakly' executed, and using this formula as the first argument in
$( (\textit{tt}\langle b \rangle \textit{tt}) \langle a \rangle
\textit{tt})$, we state that $a$ can be fired as first visible
action, possibly after an internal computation that only
passes intermediate states where $b$ can be weakly executed.} 

Of course, conceptually, the
argument is the same as the strategy we gave on page~\pageref{discussion},
but it is given in a syntax and semantics that is quite remote from
the concept of a transition system. 
Whereas the formula explains the inequivalence in terms of high-level concepts such as 
`behavioural potentials', \ie the ability to execute certain behaviour,
the game-theoretical
approach works at the level of the excution of concrete states and transitions.
\medskip

Next we give another example that illustrates how the \bbed-game (and, of course,
the \gbed-game) can be put to use
to explain inequivalence of states in a more involved setting.
Consider the LTS below, which depicts a simple specification of a
one-place buffer for exchanging two types of message ($d_1$ and
$d_2$).

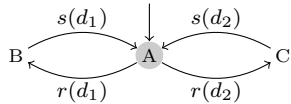
\begin{center}
\begin{tikzpicture}
  [node distance=50,inner sep = 1pt,minimum size=10pt,initial text={}]

   \tikzstyle{state}=[circle, draw=none]
   \tikzstyle{initstate}=[state,fill=black!20]
   \tikzstyle{transition}=[->,>=stealth']

    \node [initial above, initstate] (s0) {\tiny A};
    \node[state] [left of=s0] (s1) {\tiny B};
    \node[state] [right of=s0] (s2) {\tiny C};
    \draw [->] (s0) edge[bend left] node[below] {\scriptsize $r(d_1$)} (s1);
    \draw [->] (s1) edge[bend left] node[above] {\scriptsize $s(d_1$)} (s0);
    \draw [->] (s0) edge[bend right] node[below] {\scriptsize $r(d_2$)} (s2);
    \draw [->] (s2) edge[bend right] node[above] {\scriptsize $s(d_2$)} (s0);
\end{tikzpicture}
\end{center}
Suppose one tries to implement this one-place buffer using the
Alternating Bit Protocol (see Figure~\ref{fig:ABP}), only to find
out that states $A$ and $0$ are not branching bisimilar with explicit
divergence. Note that states $A$ and $0$ \emph{are} branching
bisimilar, so the difference must be in the lack of divergence in
the specification. Indeed, a distinguishing formula states just that:
formula $\textit{tt}\ \langle r(d_1) \rangle\ (\Delta \textit{tt})$,\footnote{The unary
divergence predicate $\Delta\ \phi$ of~\cite{vanglabbeek_branching_2009} holds in
a state iff it admits
an infinite $\tau$-path on which $\phi$ holds.} holds in state
$0$ but not in state $A$. The formula concisely explains  the
inequivalence, but it does not point at a concrete cause for the divergence
in the involved transition systems and which states are involved. 
In this case, \spoiler's winning strategy can
be used to `play' against the designer of the implementation in a
way similar to that of~\cite{stevens_practical_1998}, allowing the
designer to better understand the reason why this implementation
is not satisfactory. By solving the (automatically generated)
\bbed-game we obtain a winning strategy for player \spoiler, which can
be used in an interactive setting as follows:\medskip

\noindent
{\small
\begin{verbatim}
Spoiler moves A --r(d1)--> B
You respond with 0 --r(d1)--> 1
Spoiler switches positions and moves 1 --tau--> 3
You respond by not moving
...
Spoiler moves 19 --tau--> 1
You respond by not moving
You explored all options. You lose.
\end{verbatim}
}
In a similar vein, one can check also that states $B$ and $9$ are not
branching bisimilar with explicit divergence. 

\begin{figure}[hbtp]
\centering
\begin{tikzpicture}[node distance=40,inner sep = 1pt,minimum size=10pt,initial text={}]

   \tikzstyle{state}=[circle, draw=none]
   \tikzstyle{initstate}=[state,fill=black!20]
   \tikzstyle{transition}=[->,>=stealth']
\node at (-50pt, 5pt) [initstate,initial left] (state0) {\tiny 0};
\node[state] (state1) [above right of=state0] {\tiny 1}; 
\node[state] (state3) [right of=state1] {\tiny 3}; 
\node[state] (state5) [right of=state3] {\tiny 5}; 
\node[state] (state9) [right of=state5] {\tiny 9}; 
\node[state] (state13) [right of=state9] {\tiny 13}; 
\node[state] (state17) [right of=state13] {\tiny 17}; 

\node[state] (state24) [above right of=state17] {\tiny 24}; 
\node[state] (state28) [above left of=state24] {\tiny 28}; 
\node[state] (state32) [left of=state28] {\tiny 32}; 
\node[state] (state36) [left of=state32] {\tiny 36}; 
\node[state] (state37) [below right of=state32] {\tiny 37}; 
\node[state] (state44) [below right of=state36] {\tiny 44};

\node[state] (state23) [right of=state17] {\tiny 23}; 
\node[state] (state6) [above of=state3] {\tiny 6}; 
\node[state] (state10) [left of=state6,xshift=15pt] {\tiny 10}; 
\node[state] (state14) [left of=state10,xshift=15pt] {\tiny 14}; 
\node[state] (state18) [right of=state14,xshift=-15pt,yshift=-15pt] {\tiny 18}; 
\node[state] (state19) [below of=state14,yshift=15pt] {\tiny 19}; 

\node[state] (state2) [below right of=state0] {\tiny 2}; 
\node[state] (state4) [right of=state2] {\tiny 4}; 
\node[state] (state7) [right of=state4] {\tiny 7}; 
\node[state] (state11) [right of=state7] {\tiny 11}; 
\node[state] (state15) [right of=state11] {\tiny 15}; 
\node[state] (state20) [right of=state15] {\tiny 20}; 
\node[state] (state26) [below right of=state20] {\tiny 26}; 
\node[state] (state29) [below left of=state26] {\tiny 29}; 
\node[state] (state33) [left of=state29] {\tiny 33}; 
\node[state] (state38) [left of=state33] {\tiny 38}; 
\node[state] (state39) [above right of=state33] {\tiny 39}; 
\node[state] (state45) [above right of=state38] {\tiny 45};

\node[state] (state25) [right of=state20] {\tiny 25}; 
\node[state] (state8) [below of=state4] {\tiny 8}; 
\node[state] (state12) [left of=state8,xshift=15pt] {\tiny 12}; 
\node[state] (state16) [left of=state12,xshift=15pt] {\tiny 16}; 
\node[state] (state21) [right of=state16,xshift=-15pt,yshift=15pt] {\tiny 21}; 
\node[state] (state22) [above of=state16,yshift=-15pt] {\tiny 22}; 

\node[state] (state60) [above of=state0,yshift=25pt] {\tiny 60};
\node[state] (state54) [above right of=state60] {\tiny 54};
\node[state] (state50) [right of=state54] {\tiny 50}; 
\node[state] (state46) [right of=state50,xshift=60pt] {\tiny 46}; 
\node[state] (state40) [right of=state46] {\tiny 40}; 
\node[state] (state34) [right of=state40] {\tiny 34}; 

\node[state] (state41) [above left of=state34] {\tiny 41}; 
\node[state] (state47) [above right of=state41] {\tiny 47}; 
\node[state] (state51) [right of=state47] {\tiny 51}; 
\node[state] (state55) [below left of=state51] {\tiny 55}; 
\node[state] (state56) [below right of=state51] {\tiny 56};

\node[state] (state30) [right of=state34] {\tiny 30}; 
\node[state] (state27) at (state0-|state30) {\tiny 27}; 
\node[state] (state61) [above right of=state54] {\tiny 61}; 
\node[state] (state64) [above left of=state61] {\tiny 64}; 
\node[state] (state66) [left of=state64] {\tiny 66}; 
\node[state] (state68) [below right of=state66] {\tiny 68}; 
\node[state] (state69) [below left of=state66] {\tiny 69}; 
\node[state] (state72) [below right of=state69] {\tiny 72}; 

\node[state] (state62) [below of=state0,yshift=-25pt] {\tiny 62};
\node[state] (state57) [below right of=state62] {\tiny 57}; 
\node[state] (state52) [right of=state57] {\tiny 52}; 
\node[state] (state63) [below right of=state57] {\tiny 63}; 
\node[state] (state65) [below left of=state63] {\tiny 65}; 
\node[state] (state67) [left of=state65] {\tiny 67}; 
\node[state] (state70) [above left of=state67] {\tiny 70}; 
\node[state] (state71) [above right of=state67] {\tiny 71}; 
\node[state] (state73) [above right of=state70] {\tiny 73}; 

\node[state] (state48) [right of=state52,xshift=60pt] {\tiny 48}; 
\node[state] (state42) [right of=state48] {\tiny 42}; 
\node[state] (state35) [right of=state42] {\tiny 35}; 
\node[state] (state43) [below left of=state35] {\tiny 43}; 
\node[state] (state49) [below right of=state43] {\tiny 49}; 
\node[state] (state53) [right of=state49] {\tiny 53}; 
\node[state] (state58) [above left of=state53] {\tiny 58}; 
\node[state] (state59) [above right of=state53] {\tiny 59};

\node[state] (state31) [right of=state35] {\tiny 31}; 
\draw [->] (state0) edge node[right,xshift=1pt,yshift=-1pt] {\scriptsize $r(d_1$)}  (state1);
\draw [->] (state0) edge node[right,xshift=1pt,yshift=1pt] {\scriptsize $r(d_2$)}  (state2);
\draw [->] (state1) edge (state3);
\draw [->] (state2) edge (state4);
\draw [->] (state3) edge (state5);
\draw [->] (state3) edge (state6);
\draw [->] (state4) edge (state7);
\draw [->] (state4) edge (state8);
\draw [->] (state5) edge (state9);
\draw [->] (state6) edge (state10);
\draw [->] (state7) edge (state11);
\draw [->] (state8) edge (state12);
\draw [->] (state9) edge node[above] {\scriptsize $s(d_1$)} (state13);
\draw [->] (state10) edge (state14);
\draw [->] (state11) edge node[below] {\scriptsize $s(d_2$)}  (state15);
\draw [->] (state12) edge (state16);
\draw [->] (state13) edge (state17);
\draw [->] (state14) edge (state18);
\draw [->] (state14) edge (state19);

\draw [->] (state15) edge (state20);
\draw [->] (state16) edge (state21);
\draw [->] (state16) edge (state22);
\draw [->] (state17) edge (state23);
\draw [->] (state17) edge (state24);
\draw [->] (state18) edge (state1);
\draw [->] (state19) edge (state1);
\draw [->] (state20) edge (state25);
\draw [->] (state20) edge (state26);
\draw [->] (state21) edge (state2);
\draw [->] (state22) edge (state2);
\draw [->] (state23) edge (state27);
\draw [->] (state24) edge (state28);
\draw [->] (state25) edge (state27);
\draw [->] (state26) edge (state29);
\draw [->] (state27) edge node[right] {\scriptsize $r(d_1$)} (state30);
\draw [->] (state27) edge node[right] {\scriptsize $r(d_2$)} (state31);
\draw [->] (state28) edge (state32);
\draw [->] (state29) edge (state33);
\draw [->] (state30) edge (state34);
\draw [->] (state31) edge (state35);
\draw [->] (state32) edge (state36);
\draw [->] (state32) edge (state37);
\draw [->] (state33) edge (state38);
\draw [->] (state33) edge (state39);
\draw [->] (state34) edge (state40);
\draw [->] (state34) edge (state41);

\draw [->] (state35) edge (state42);
\draw [->] (state35) edge (state43);
\draw [->] (state36) edge (state44);
\draw [->] (state37) edge (state44);
\draw [->] (state38) edge (state45);
\draw [->] (state39) edge (state45);
\draw[->] (state40) edge (state46);
\draw [->] (state41) edge (state47);
\draw [->] (state42) edge (state48);
\draw [->] (state43) edge (state49);
\draw [->] (state44) edge (state17);
\draw [->] (state45) edge (state20);
\draw[->] (state46) edge node[above] {\scriptsize $s(d_1$)}  (state50);
\draw [->] (state47) edge (state51);
\draw [->] (state48) edge node[below] {\scriptsize $s(d_2$)} (state52);
\draw [->] (state49) edge (state53);
\draw [->] (state50) edge (state54);
\draw [->] (state51) edge (state55);
\draw [->] (state51) edge (state56);
\draw [->] (state52) edge (state57);
\draw [->] (state53) edge (state58);
\draw [->] (state53) edge (state59);
\draw [->] (state54) edge (state60);
\draw [->] (state54) edge (state61);
\draw [->] (state55) edge (state30);
\draw [->] (state56) edge (state30);
\draw [->] (state57) edge (state62);

\draw [->] (state57) edge (state63);
\draw [->] (state58) edge (state31);
\draw [->] (state59) edge (state31);
\draw [->] (state60) edge (state0);
\draw [->] (state61) edge (state64);
\draw [->] (state62) edge (state0);
\draw [->] (state63) edge (state65);
\draw [->] (state64) edge (state66);
\draw [->] (state65) edge (state67);
\draw [->] (state66) edge (state68);
\draw [->] (state66) edge (state69);
\draw [->] (state67) edge (state70);
\draw [->] (state67) edge (state71);
\draw [->] (state68) edge (state72);
\draw [->] (state69) edge (state72);
\draw [->] (state70) edge (state73);
\draw [->] (state71) edge (state73);
\draw [->] (state72) edge (state54);
\draw [->] (state73) edge (state57);

\end{tikzpicture}
\caption{The ABP with two messages; unlabelled transitions
are $\tau$ transitions.}
\label{fig:ABP}
\end{figure}
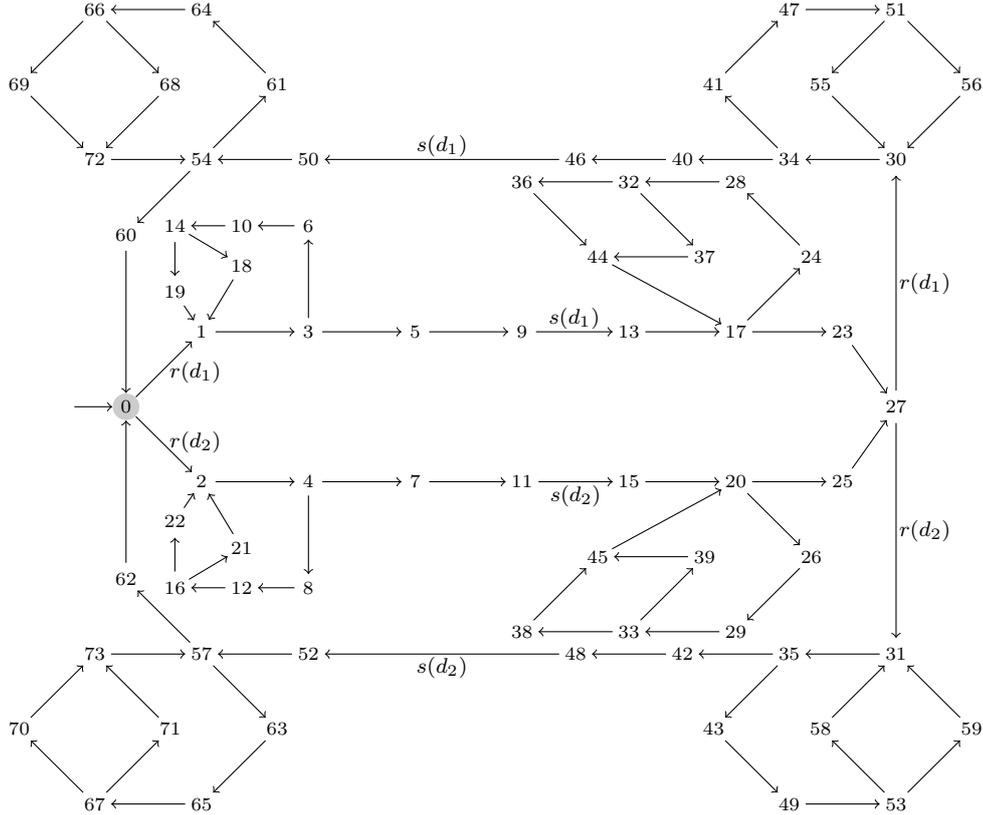


\section{Closing Remarks}\label{sec:conclusions}
In this paper we introduced a generic, game-theoretic definition of weak bisimulation relations. In particular, we showed how our games can be instantiated to obtain games that characterise branching-, $\eta$-, delay- and weak bisimilarity. We also illustrated that our generic bisimulation games generalise the branching bisimulation games we presented in \cite{branching_games_2016}.

The definition that we have presented does not require any transitive closure of $\tau$-transitions in the game definition, so that a `local' assessment is obtained when two states are found to be inequivalent. Any sequence of $\tau$-transitions is handled explicitly in the game, in a step-by-step fashion. This enables a more straightforward use of our generic bisimulation games for debugging whenever states are not equivalent.

We generalised our generic bisimulation games to deal with divergence as a first-class citizen: no precomputation of divergences, and subsequent modification of the game, is needed. 
In particular, we showed that our game, when suitably instantiated, characterises branching bisimulation with explicit divergences, just as it was the case for the game we previously presented in~\cite{branching_games_2016}. 

Additionally, we illustrated the (simple) generalisation of our games to cover simulation relations.

\subsection*{Future work}
We have experimented with a prototype of the game-theoretic definitions of branching bisimulation (also with explicit divergence), both using an interactive, command-line application, as well as with a graphical user interface. We intend to make a proper implementation available in the mCRL2 toolset \cite{cranen_overview_2013}.

In this paper, we have used condition D$_4$ from \cite{vanglabbeek_branching_2009} to obtain an extension sensitive to divergence. As we illustrated in Section~\ref{sec:diverging_discussion}, 
using any of their other conditions is non-trivial;
indeed, it is far from immediate that the use of any of their other conditions would produce equivalent relations. 
These conditions should be explored further to answer the questions raised in this paper.
Additionally, the notion of divergence considered in this paper is not the only notion of divergence considered in the literature. One could, for example, drop the requirement that vertices on divergent paths are related~\cite{bergstra_failures_1987}. 
Game-based characterisations for this and other notions such as 
\emph{divergence sensitive branching bisimulation}~\cite{denicola_three_1995}, but also notions such as
the \emph{next-preserving branching bisimulations} of~\cite{Yatapanage15} can be studied. Finally, it would be interesting to extend the approach of~\cite{CKW:17} and investigate whether the games we have studied in this paper can be applied to B\"uchi automata or parity games, \ie models other than labelled transition systems. 

\paragraph{\emph{Acknowledgements and Dedicatory.}} The authors would like to thank the reviewers for
their constructive feedback and their insights which enabled us to
improve the paper considerably. The first author wants to dedicate this paper to the memory of his mother Virginia, that passed away during its elaboration.


\bibliographystyle{plain}

\begin{thebibliography}{10}

\bibitem{AndersenAEHLOSW15}
J.R. Andersen, N.~Andersen, S.~Enevoldsen, M.M. Hansen, K.G. Larsen, S.R.
  Olesen, J.~Srba, and J.K. Wortmann.
\newblock {CAAL:} {C}oncurrency {W}orkbench, {A}alborg {E}dition.
\newblock In {\em {ICTAC}}, volume 9399 of {\em Lecture Notes in Computer
  Science}, pages 573--582. Springer, 2015.

\bibitem{baeten_another_1987}
J.C.M. Baeten and R.J.~van Glabbeek.
\newblock Another look at abstraction in process algebra.
\newblock In T.~Ottmann, editor, {\em Proc. {ICALP}'87}, volume 267 of {\em
  {LNCS}}, pages 84--94. Springer, 1987.

\bibitem{basten_branching_1996}
T.~Basten.
\newblock Branching bisimilarity is an equivalence indeed!
\newblock {\em Inform. Process. Lett.}, 58(3):141--147, May 1996.

\bibitem{bergstra_failures_1987}
J.A. Bergstra, J.W. Klop, and E.R. Olderog.
\newblock Failures without chaos: A new process semantics for fair abstraction.
\newblock In {\em Formal Description of Programming Concepts III - Proceedings
  of the IFIP TC2/WG 2.2 Working Conference}, pages 77--103. North Holland,
  1987.

\bibitem{blom_mcrl_2001}
S.C.C. Blom, W.J. Fokkink, J.F. Groote, I.~van Langevelde, B.~Lisser, and J.C.
  van~de Pol.
\newblock \ensuremath{\mu}{CRL}: A toolset for analysing algebraic
  specifications.
\newblock In {\em Proc. {CAV}'01}, volume 2102 of {\em {LNCS}}, pages 250--254,
  2001.

\bibitem{bulychev_computing_2007}
P.E. Bulychev, I.V. Konnov, and V.A. Zakharov.
\newblock Computing (bi)simulation relations preserving {CTL}*-{X} for ordinary
  and fair {K}ripke structures.
\newblock {\em Inst. for Syst. Progr., Russian Acad. of Sci., Math. Meth. and
  Algor.}, 12, 2007.

\bibitem{cranen_overview_2013}
S.~Cranen, J.F. Groote, J.J.A. Keiren, F.P.M. Stappers, E.P. de~Vink, J.W.
  Wesselink, and T.A.C. Willemse.
\newblock An overview of the {mCRL}2 toolset and its recent advances.
\newblock In {\em Proc. {TACAS}'13}, volume 7795 of {\em {LNCS}}, pages
  199--213, 2013.

\bibitem{CKW:17}
S.~Cranen, J.J.A. Keiren, and T.A.C. Willemse.
\newblock Parity game reductions.
\newblock {\em Acta Informatica}, pages 1--44, 8 2017.

\bibitem{branching_games_2016}
D.~de~Frutos~Escrig, J.J.A. Keiren, and T.A.C. Willemse.
\newblock Branching bisimulation games.
\newblock In E.~Albert and I.~Lanese, editors, {\em Proc. {FORTE}'16}, pages
  142--157. Springer, 2016.

\bibitem{denicola_three_1995}
R.~de~Nicola and F.W. Vaandrager.
\newblock Three logics for branching bisimulation.
\newblock {\em J. {ACM}}, 42(2):458--487, March 1995.

\bibitem{Etessami:01}
K.~Etessami, Th. Wilke, and R.A. Schuller.
\newblock Fair simulation relations, parity games, and state space reduction
  for {B\"{u}}chi automata.
\newblock In {\em {ICALP}}, volume 2076 of {\em Lecture Notes in Computer
  Science}, pages 694--707. Springer, 2001.

\bibitem{Fritz06}
C.~Fritz and Th. Wilke.
\newblock Simulation relations for alternating parity automata and parity
  games.
\newblock In {\em Developments in Language Theory}, volume 4036 of {\em Lecture
  Notes in Computer Science}, pages 59--70. Springer, 2006.

\bibitem{garavel_cadp_2013}
H.~Garavel, F.~Lang, R.~Mateescu, and W.~Serwe.
\newblock {CADP} 2011: a toolbox for the construction and analysis of
  distributed processes.
\newblock {\em Int. Journ. on Softw. Tools for Techn. Transfer}, 15(2):89--107,
  April 2013.

\bibitem{gerth_partial_1999}
R.~Gerth, R.~Kuiper, D.~Peled, and W.~Penczek.
\newblock A partial order approach to branching time logic model checking.
\newblock {\em Inform. and Comput.}, 150(2):132--152, 1999.

\bibitem{fdr}
T.~Gibson-Robinson, P.~Armstrong, A.~Boulgakov, and A.W. Roscoe.
\newblock {FDR3}: a parallel refinement checker for {CSP}.
\newblock {\em International Journal on Software Tools for Technology
  Transfer}, 18(2):149--167, 2016.

\bibitem{gradel_automata_2002}
E.~Gr{\"a}del, W.~Thomas, and T.~Wilke, editors.
\newblock {\em Automata Logics, and Infinite Games}, volume 2500 of {\em LNCS}.
\newblock Springer, 2002.

\bibitem{hennessy80}
M.~Hennessy and G.D. Plotkin.
\newblock A term model for {CCS}.
\newblock In {\em Mathematical Foundations of Computer Science 1980 (MFCS'80),
  Proceedings of the 9th Symposium}, volume~88 of {\em Lecture Notes in
  Computer Science}, pages 261--274. Springer, 1980.

\bibitem{Hutagalung16}
M.~Hutagalung, N.~Hundeshagen, D.~Kuske, M.~Lange, and {\'{E}}.~Lozes.
\newblock Multi-buffer simulations for trace language inclusion.
\newblock In {\em GandALF}, volume 226 of {\em {EPTCS}}, pages 213--227, 2016.

\bibitem{korver_computing_1992}
H.~Korver.
\newblock Computing distinguishing formulas for branching bisimulation.
\newblock In {\em Proc. {CAV}'92}, volume 575, pages 13--23. Springer, 1992.

\bibitem{milner_calculus_1980}
R.~Milner.
\newblock {\em A Calculus of Communicating Systems}, volume~92 of {\em {LNCS}}.
\newblock Springer, 1980.

\bibitem{milner_modal_1981}
R.~Milner.
\newblock A modal characterisation of observable machine-behaviour.
\newblock In E.~Astesiano and C.~Böhm, editors, {\em {CAAP} '81}, volume 112
  of {\em {LNCS}}, pages 25--34. Springer Berlin Heidelberg, 1981.

\bibitem{namjoshi_simple_1997}
K.S. Namjoshi.
\newblock A simple characterization of stuttering bisimulation.
\newblock In {\em Proc. {FSTTCS}'97}, volume 1346 of {\em LNCS}, pages
  284--296. Springer, 1997.

\bibitem{park_concurrency_1981}
D.~Park.
\newblock Concurrency and automata on infinite sequences.
\newblock In {\em Proc. {TCS}'81}, volume 104 of {\em LNCS}, pages 167--183.
  Springer, 1981.

\bibitem{reniers_results_2014}
M.A. Reniers, R.~Schoren, and T.A.C. Willemse.
\newblock Results on embeddings between state-based and event-based systems.
\newblock {\em Comput. J.}, 57(1):73--92, 2014.

\bibitem{sangiorgi_introduction_2012}
D.~Sangiorgi.
\newblock {\em Introduction to Bisimulation and Coinduction}.
\newblock Cambridge Universtity Press, 2012.

\bibitem{stevens_practical_1998}
P.~Stevens and C.~Stirling.
\newblock Practical model-checking using games.
\newblock In {\em Proc. {TACAS}'98}, volume 1384 of {\em LNCS}, pages 85--101.
  Springer, 1998.

\bibitem{stirling_bisimulation_1999}
C.~Stirling.
\newblock Bisimulation, modal logic and model checking games.
\newblock {\em Logic Journal of IGPL}, 7(1):103--124, January 1999.

\bibitem{thomas_ehrenfeucht_1993}
W.~Thomas.
\newblock On the {E}hrenfeucht-{F}ra{\"\i}ss{\'e} game in theoretical computer
  science.
\newblock In {\em Proc. {TAPSOFT}'93}, volume 668 of {\em LNCS}, pages
  559--568. Springer, 1993.

\bibitem{vanglabbeek_linear_1993}
R.J. van Glabbeek.
\newblock The linear time - {Branching} time spectrum {II}.
\newblock In E.~Best, editor, {\em {CONCUR}'93}, volume 715 of {\em {LNCS}},
  pages 66--81. Springer, 1993.
\newblock An extended preliminary version can be reached at
  http://citeseerx.ist.psu.edu/viewdoc/summary?doi=10.1.1.29.1931.

\bibitem{vanglabbeek_branching_2009}
R.J. van Glabbeek, B.~Luttik, and N.~Tr\c{c}ka.
\newblock Branching bisimilarity with explicit divergence.
\newblock {\em Fundam. Inform.}, 93(4):371--392, 2009.

\bibitem{vanglabbeek_computation_2009}
R.J. van Glabbeek, B.~Luttik, and N.~Tr\c{c}ka.
\newblock Computation tree logic with deadlock detection.
\newblock {\em Logical Methods in Computer Science}, 5(4), 2009.

\bibitem{vanglabbeek_branching_1996}
R.J. van Glabbeek and W.P. Weijland.
\newblock Branching time and abstraction in bisimulation semantics.
\newblock {\em J. {ACM}}, 43(3):555--600, May 1996.

\bibitem{walker90}
D.J. Walker.
\newblock Bisimulation and divergence.
\newblock {\em Inf. Comput.}, 85(2):202--241, 1990.

\bibitem{Yatapanage15}
N.~Yatapanage and K.~Winter.
\newblock Next-preserving branching bisimulation.
\newblock {\em Theor. Comput. Sci.}, 594:120--142, 2015.

\bibitem{yin_branching_2014}
Q.~Yin, Y.~Fu, C.~He, M.~Huang, and X.~Tao.
\newblock Branching bisimilarity checking for {PRS}.
\newblock In {\em Proc. {ICALP}'14}, volume 8573 of {\em LNCS}, pages 363--374.
  Springer, 2014.

\end{thebibliography}

\end{document}